\documentclass[conference]{IEEEtran}
\IEEEoverridecommandlockouts
\pdfoutput=1

\usepackage[utf8]{inputenc}
\usepackage{mathtools,bm}
\usepackage{cite}
\usepackage{amsmath,amssymb,amsfonts,amsthm}
\usepackage{algorithmic}
\usepackage{graphicx}
\usepackage{textcomp}
\usepackage{xcolor}
\def\BibTeX{{\rm B\kern-.05em{\sc i\kern-.025em b}\kern-.08em
    T\kern-.1667em\lower.7ex\hbox{E}\kern-.125emX}}
\usepackage{subfig}
\usepackage{hyperref}
\usepackage{multirow}
\usepackage{cleveref}

\usepackage{todonotes}
\usepackage{paralist}
\usepackage{xspace}
\usepackage{empheq}
\usepackage[thinlines]{easytable}
\usepackage{tikz}
\usepackage{thm-restate}
\usepackage{orcidlink}
\usetikzlibrary{hobby,backgrounds,calc,trees}
\usetikzlibrary{automata, positioning, arrows}

\newif\ifcomment
\commenttrue
\ifcomment
\newcommand{\cmahsa}[1]{\todo[color=green!40,inline]{{\bf Mahsa:} #1}}
\newcommand{\cben}[1]{\todo[color=blue!40,inline]{{\bf Ben:} #1}}
\newcommand{\crida}[2][]{\todo[color=teal!40,#1]{{\bf Alex:} #2}}
\newcommand{\cvincent}[2][]{\todo[color=red!40,#1]{{\bf Vincent:} #2}}

\newcommand{\mahsa}[1]{\color{green!70!black} #1 \color{black}\xspace}
\newcommand{\ben}[1]{\color{blue} #1 \color{black}\xspace}
\newcommand{\rida}[1]{\color{teal} #1 \color{black}\xspace}
\newcommand{\vincent}[1]{\color{red} #1 \color{black}\xspace}
\else
\newcommand{\cmahsa}[2][]{}
\newcommand{\cben}[2][]{}
\newcommand{\crida}[2][]{}
\newcommand{\cvincent}[2][]{}

\newcommand{\mahsa}[1]{}
\newcommand{\ben}[1]{}
\newcommand{\rida}[1]{}
\newcommand{\vincent}[1]{}
\fi

\newcommand{\vbeta}{\boldsymbol{\beta}} 

\newcommand{\ve}{\boldsymbol{e}}
\newcommand{\va}{\boldsymbol{a}}
\newcommand{\vb}{\boldsymbol{b}}
\newcommand{\vc}{\boldsymbol{c}}

\newcommand{\vf}{\boldsymbol{f}}

\newcommand{\vzero}{\boldsymbol{0}}
\newcommand{\vone}{\boldsymbol{1}}

\newcommand{\A}{{\mathcal A}}                   
\newcommand{\B}{{\mathcal B}}                   
\newcommand{\hmu}{\boldsymbol{\mu}}

\newcommand{\sem}[2]{[\![#1]\!](#2)}
\newcommand{\genS}[1]{f_{#1}}
\newcommand{\LgenS}[1]{\tilde{f}_{#1}}

\newcommand{\N}{\mathbb{N}}
\newcommand{\K}{\mathbb{K}}
\newcommand{\Z}{\mathbb{Z}}

\newcommand{\Q}{\mathbb{Q}}

\newcommand{\KK}{\K_{u}}

\newcommand{\size}[1]{\|#1\|}

\newcommand{\SDerive}[1]{\Theta #1}
\newcommand{\SDeriveN}[2]{\Theta^{(#2)} #1}

\newcommand{\sequence}{\mathsf{sequence}}
\newcommand{\set}{\mathsf{set}}
\newcommand{\cycle}{\mathsf{cycle}}

\allowdisplaybreaks

\newtheorem{theorem}{Theorem}

\newtheorem{corollary}{Corollary}

\newtheorem{proposition}{Proposition}

\newtheorem{example}{Example}

\begin{document}

\title{Differential Tree Automata}

\author{\IEEEauthorblockN{Rida Ait El Manssour\,\orcidlink{0000-0001-6228-9071}}
\IEEEauthorblockA{\textit{University of Oxford, UK}}
\and
\IEEEauthorblockN{Vincent Cheval\,\orcidlink{0000-0002-3622-2129}}
\IEEEauthorblockA{\textit{University of Oxford, UK}}
\and
\IEEEauthorblockN{Mahsa Shirmohammadi\,\orcidlink{0000-0002-7779-2339}}
\IEEEauthorblockA{\textit{CNRS, IRIF, France}}
\and
\IEEEauthorblockN{James Worrell\,\orcidlink{0000-0001-8151-2443}}
\IEEEauthorblockA{\textit{University of Oxford, UK}}
}

\maketitle

\begin{abstract}
A rationally dynamically algebraic (RDA) power series is one that 
arises as (a component of) the solution of 
a system of differential equations of the form
$\boldsymbol{y}' = F(\boldsymbol y)$, where $F$ is a vector of rational functions
that is defined at $\boldsymbol y(0)$.
RDA power series subsume algebraic power series and are a proper subclass of differentially algebraic power series (those that satisfy a univariate polynomial-differential equation).
We give a combinatorial characterisation of RDA power series in terms of 
exponential generating functions of regular languages of labelled trees.  
Motivated by this connection, we define the notion of a differential tree automaton.  Differential tree automata generalise weighted tree automata by allowing the transition weights to be rational functions of the tree size.  Our main result is that the ordinary generating functions of the formal tree series recognised by differential tree automata are exactly the differentially algebraic power series.   The proof of this result establishes a general form of recurrence satisfied by the
sequence of coefficients of a differentially algebraic power series,
generalising Reutenauer's matrix representation of polynomially recursive sequences.
As a corollary we obtain a procedure for determining equality of differential tree automata.

\end{abstract}



\section{Introduction}


\subsection{Combinatorial Classes, Automata, and Power Series}
Many combinatorial classes admit natural representations as languages of words and trees, which can be specified using automata or grammars.  This provides a method for systematically extracting generating functions for combinatorial classes from their specifications (as is also done in  the theory of combinatorial species~\cite{bergeron1998combinatorial}) and leads to a natural correspondence between classes of automata and classes of power series.  For example, 
regular languages of words have rational generating functions~\cite[Section I.4]{flajolet2009analytic}, while unambiguous context-free languages of words have algebraic generating functions~\cite[Sections I.5]{flajolet2009analytic}.
Here, recall that for a field $\K$ a formal power series $f \in \K[\![x]\!]$ is \emph{rational} if there are polynomials $P,Q \in \K[x]$ such that $f=\frac{P}{Q}$ and
$f$ is \emph{algebraic} if there are polynomials 
$P_0,\ldots,P_d \in \K[x]$ such that $\sum_{i=0}^d P_i f^i =0$.

In this paper we develop the correspondence between automata, power series, and combinatorial classes to relate tree automata and
differentially algebraic power series.  Our main contribution is to define the class of differential tree automata and to show that their generating functions are precisely the differentially algebraic power series.
Recall  here that a formal power series $f \in \K[\![x]\!]$ is \emph{differentially algebraic} if it satisfies a differential equation 
$P(f,f',\ldots,f^{(k)})=0$
for some polynomial $P(y_1,\ldots,y_k)$ with coefficients in $\K[x]$.
We also give a combinatorial characterisation of 
a proper subclass of differentially algebraic power series---the so-called \emph{rationally dynamically algebraic} (RDA) series~\cite{FORSMAN1992341,ovchinnikov2022bounds}.
These arise as (components of) solutions of
systems of differential equations
\begin{gather}y_1'=F_1(y_1,\ldots,y_k), \, \ldots, \, y'_k=F_k(y_1,\ldots,y_k),
\label{eq:RDA}
\end{gather}
where $F_1,\ldots,F_k \in \K(x,y_1,\ldots,y_k)$ are rational functions that are defined at $(0,y_1(0),\ldots,y_k(0))$. 
 In Section~\ref{sec:SPECIES}
we give a collection of rational
dynamical systems that characterise the exponential generation
functions of standard constructions of combinatorial
species.

\subsection{Tree Automata and their Generating Functions}
 Our main results concern the notion of 
\emph{differential tree automata}: 
a generalisation of weighted tree automata~\cite{fülöp2024weightedtreeautomata}.  Let $\Sigma$ be a finite ranked alphabet and write 
$T_\Sigma$ for the set of $\Sigma$-trees.  
As with a $\K$-weighted automaton,
a differential tree automaton~$\mathcal A$ computes a function 
$[\![\mathcal A]\!] : T_\Sigma \rightarrow \K$.
Recall that a transition in a $\K$-weighted tree automaton has the form
\[
\sigma(q_1,\ldots,q_\ell) \stackrel{a}{\longrightarrow} q \, ,
\]
where $\sigma$ is a $\ell$-ary alphabet symbol, $q_1,\ldots,q_\ell$ and $q$ are states of the automaton,
and $a \in \K$  is the weight of the transition. 
Intuitively, the automaton
reads the symbol $\sigma$ and makes a transition from the tuple of states $q_1,\ldots,q_\ell$ to $q$ with weight~$a$.
Generalising this pattern, in a differential tree automaton the weight of a transition is a rational function $R(x)$ that has no positive integer poles.  
When instantiated on a tree~$t$, the transition has weight $R(\size{t})$, where~$\size{t}$ denotes the size of~$t$ (only counting  the number of internal nodes in $t$).

We consider two different generating functions associated with (weighted and differential) tree automata.
The \emph{ordinary generating function} of an automaton $\mathcal A$
is the power series $f_{\mathcal A} \in \K[\![x]\!]$ defined by
\begin{gather}
    f_{\mathcal A} (x) = \sum_{n \in \mathbb N} \, \Bigg( \sum_{\substack{t \in T_\Sigma\\ \size{t}=n}} \sem{\A}{t} \Bigg) x^n.
\label{eq:ordinary}
\end{gather}
A formulation of the Chomsky-Schutzenberger theorem is that
the class of algebraic power series in $\K[\![x]\!]$ coincides with the class of ordinary generating functions of tree series recognized by
$\K$-weighted tree automata.   

We also introduce the notion of a \emph{labelled generating function}.
A \emph{monotonic labelling} of a tree $t\in T_\Sigma$ of size $n$ is a bijective mapping from the set of internal nodes of $t$ to
the set $\{1,\ldots,n\}$ such that the label of any node is greater than the label of any of its children.
Define
$\lambda :T_\Sigma\rightarrow\mathbb Q$ such that $\lambda(t)$ is the proportion of labellings of $t \in T_\Sigma$ that are monotonic. That is, 
the number of monotonic labellings of $t$ is 
$\lambda(t) n!$ .
The labelled generating function of $\mathcal A$ is the formal power series
$\widetilde{f}_{\mathcal A} \in \K[\![x]\!]$ defined by
\begin{equation} 
    \widetilde{f}_{\mathcal A}(x) := \sum_{n\in\mathbb N} \Bigg(  \sum_{\substack{t\in T_\Sigma\\ \size{t}=n} }
\lambda(t) \sem{\A}{t} \Bigg) x^n.\label{eq:labelled}
\end{equation}

\subsection{Main Results} 
Our first result is as follows:

\begin{restatable}{theorem}{thmainone}
    The following are equivalent for a power series $g(x) = \sum_{n=0}^\infty a_n x^n$ in $\K[\![x]\!]$:
\begin{enumerate}
    \item $g$ is rationally dynamically algebraic; \label{enum:thm:MAIN1-RDA}
    
    \item $g$ is equal to the ordinary generating function of a differential tree automaton in which the rational weight of 
 every transition  has the form $\frac{ax+b}{x}$ for some $a,b\in \K$;\label{enum:thm:MAIN1-ordinary}
 \item $g$ is equal to the labelled generating function of a $\K$-weighted tree automaton.\label{enum:thm:MAIN1-labelled}
 \end{enumerate}
\label{thm:MAIN1}
\end{restatable}
From Theorem~\ref{thm:MAIN1} it follows that every RDA power series $g \in \mathbb Z[\![x]\!]$ 
can be written $g=\widetilde{f}_{\mathcal A}-\widetilde{f}_{\mathcal B}$, where $\mathcal A$ and $\mathcal B$ are
deterministic bottom-up tree automata.

Removing the restriction on the form of rational function in Item~\ref{enum:thm:MAIN1-ordinary} of~\Cref{thm:MAIN1}, we show that the ordinary generating functions of differential tree automata 
are precisely the differentially algebraic series.


\begin{restatable}{theorem}{thmaintwo}
\label{thm:MAIN2}
The following are equivalent for a power series $g(x) = \sum_{n=0}^\infty a_n x^n$ in $\K[\![x]\!]$:
\begin{enumerate}
    \item $g$ is differentially algebraic;\label{enum:thm:MAIN2-D-algebraic}
    \item $g$ is equal to the ordinary generating function
    of a differential tree automaton;\label{enum:thm:MAIN2-ordinary}
\item  $(a_n)_{n=0}^\infty$ is a component of a sequence 
$(\boldsymbol b_n)_{n=0}^\infty \in \K^d$ of vectors that satisfies a bilinear recurrence 
with polynomial coefficients, of the form
\[
\boldsymbol b_{n} = \vb_{n-1} \, P(n) + 
 \sum_{k=0}^{n-1}
 (\boldsymbol{b}_{k} \otimes \boldsymbol{b}_{n-k-1})Q(n) \, ,
\]
where $P \in \K(x)^{d\times d}$ and $Q \in \K(x)^{d^2\times d}$ have entries defined on positive integers.
\label{item:thm:MAIN2}\label{enum:thm:MAIN2-recurrence}
\end{enumerate}
\end{restatable}

The proof of~\Cref{thm:MAIN2} builds on a construction of Denef and Lipschitz~\cite{denef1984power}, who showed that 
the sequence of coefficients $(a_n)_{n=0}^\infty$ of a
differentially algebraic power series satisfies a recurrence of the form
\[P(n) \,  a_{n} = Q_n(a_0,\ldots,a_{n-1}) \, , \]
where $P_n$ is a polynomial whose form depends on $n$.  Building on the construction of~\cite{denef1984power}
we show that $(a_n)_{n=0}^\infty$ can be recovered as a component of vectors satisfying a recurrence of the form shown in \Cref{enum:thm:MAIN2-recurrence} of Theorem~\ref{thm:MAIN2}.

As a consequence of \Cref{enum:thm:MAIN2-ordinary} of~\Cref{thm:MAIN2} we obtain a decision procedure for equivalence of differential tree automata (\Cref{th:equivalence}).
We further show that the representation of differentially algebraic power series by recurrences in \Cref{enum:thm:MAIN2-ordinary,enum:thm:MAIN2-recurrence} 
is convenient for 
for computing the sum, Cauchy product,
inverse, derivative, integration, and forward and backward shift of such series (\Cref{prop:closure} in Appendix~\ref{sec:operations}).

\subsection{Related Work}
The special case of RDA series in which the functions $F_1,\ldots,F_k$ in~\eqref{eq:RDA} are polynomials 
is studied in Bergeron and Reutenauer~\cite{Bergeron1990CombinatorialRO} under the name \emph{constructibly differentiably algebraic (CDA)} series.
We observe below that the classes of RDA and CDA series coincide. 
 A combinatorial interpretation of CDA power series via so-called \emph{tree labelling tables}, which are an automata-like 
notion on labelled non-plane trees, is given in~\cite[Theorem 1]{Bergeron1990CombinatorialRO}.
This result is closely related to Item~\ref{enum:thm:MAIN1-labelled} of Theorem~\ref{thm:MAIN1}. 

A power series $f \in \K[\![x]\!]$
 is \emph{differentially finite} (D-finite) if it satisfies a relation of the form
\[\sum_{i=0}^d P_i(x)f^{(i)}(x) =0\] for polynomials $P_0,\ldots,P_d \in \K[x]$.  Thus differentially finite power series 
are a subclass of differentially algebraic series.  
It is well-known~\cite{kauers2023d} that $f(x)=\sum_{n=0}^\infty a_nx^n$ is differentially finite
if and only if its sequence 
$(a_n)_{n=0}^\infty$ of coefficients satisfies a linear recurrence with polynomial coefficients.
Reutenauer~\cite{Reutenauer12} shows that such sequences can equivalently be characterised by vector 
recurrences corresponding  to a special case of the recurrence in \Cref{enum:thm:MAIN2-recurrence} of Theorem~\ref{thm:MAIN2} in which $R$ is identically zero.
See~\cite{BunaMargineanCSW24} for a development of these ideas, including an exact learning algorithm for such recurrences and their automata extension.  

Clemente~\cite{Clemente24} uses techniques from differential algebra to show that multiplicity equivalence of \emph{basic parallel processes} can be decided in doubly exponential space.  He also shows that commutative weighted basic parallel processes (viewed as power series) coincide with the class of CDA series,  thereby obtaining a doubly exponential time procedure for deciding zeroness of CDA series.  

Boreale~\cite{Boreale19} defines a notion of bisimulation for systems of ordinary differential equations, based on the Lie derivative.  He gives an algorithm to compute the largest bisimulation on a given system, and shows how this can be used for minimisation and determining equivalence.  As observed in~\cite{Clemente24}, this procedure decides zeroness of CDA power series; however no complexity analysis is given in~\cite{Boreale19}.

Resolving a conjecture of Castiglione and Massazza~\cite{CASTIGLIONE201774},
Bostan \emph{et al.}~\cite{BostanCKN20} show that the multivariate power series arising as the generating functions  of languages accepted by weakly-unambiguous Parikh automata are $D$-finite.  This result is used to prove the inherent unambiguity of certain languages based on analytic properties of univariate $D$-finite series considered as functions.

We focus on automata in this paper.  Other works consider the relationship between recurrences and generating functions on the one hand, and combinatorial classes 
defined by logical formulas on the other hand~\cite{BellBY12,fischer2011application}.


\section{Overview}

\subsection{Tree Automata}

Let $\Sigma$ be a ranked alphabet, that is, $\Sigma$ is a finite set of function symbols, each having a non-negative integer arity.
For $k \in \N$, we denote by $\Sigma_k$ the set of symbols of arity $k$ in $\Sigma$.
A  (non-deterministic) \emph{tree automaton} $\A$ is a tuple \((Q, \Delta, F)\), where
$Q$ is a finite set of states, $\Delta$ is a set of transition rules of the form
  \[
  \sigma(q_1, q_2, \ldots, q_k) \longrightarrow q,
  \]
  where \(\sigma \in \Sigma\) is a \(k\)-ary symbol and \(q_1, q_2, \ldots, q_k,q \in Q\) are 
  states, and $F$ is a set of accepting states~\cite{Comon1997TreeAT}.
  
The automaton processes a tree bottom-up, starting at the leaves. 
Leaf nodes are assigned states according to the rules in \(\Delta\) for nullary symbols. States are then propagated upwards using the transition rules, combining the states of child nodes to determine the state of their parent. The value $\sem{\A}{t}$ of a tree $t\in T_\Sigma$ is the number of runs on $t$ in which the root is labelled with state in $F$.
The $n$-th coefficient of the \emph{ordinary generating function} power series  of an automaton $\A$
is the sum of $\sem{\A}{t}$ over all trees of size $n$ (see~\eqref{eq:ordinary}).
Similarly, the $n$-th coefficient of the \emph{labelled generating function} power series is the number of monotonic labelled trees accepted by the automaton divided by $n!$ (see~\eqref{eq:labelled}).

\begin{example}[Partitions]
\label{ex:bell}
Recall that the sequence 
\[(b_n)_{n=0}^\infty=(1,1,2,5,15,52,203,\ldots)\]
of Bell numbers, where $b_n$ is the number  of partitions of
the set $\{1,\ldots,n\}$.  
As noted in~\cite{Bergeron1990CombinatorialRO}, the exponential generating function
\[f_B(x):=\sum_{n=0}^\infty \frac{b_n}{n!}x^n\] is (the $y_1$-component of) a solution of the system of differential equations 
\[ y_1'=y_1y_2, \quad y_2'=y_2.\]  Here we show that $f_B$
arises as the labelled generating function of a tree automaton.

Let $\Sigma=\{\sigma_0,\sigma_1,\sigma_2\}$ be the signature for ranked trees of degree at most 2. Here we consider $\sigma_0, \sigma_1$ and $\sigma_2$ to have respectively arities $0$, $1$ and $2$. 
Let the regular language $L\subseteq T_\Sigma$ comprise all trees $t$ accepted by the automaton 
$\mathcal A$ with set of states $\{q_1,q_2\}$, where $q_1$ is accepting, with transitions
\begin{align*}
    \sigma_0& \longrightarrow q_1 \quad
   &\sigma_0 &\longrightarrow q_2 \\
   \sigma_1&(q_2) \longrightarrow q_2 \quad
   &\sigma_2&(q_1,q_2) \longrightarrow q_1
\end{align*} 
We note that a tree $t\in L$ can be considered as a list of chains, where each chain (depicted as a coloured block in 
Figure~\ref{fig:EX}),
consists of a set of nodes all having a common ancestor under the right-child relation.

Each labelling of a tree $t \in L$ thus determines a unique partition of $\{1,\ldots,n\}$: the components of the partition are 
the sets of labels of each chain in $t$.  Moreover every partition arises as the labelling of a unique tree $t \in L$.
We thus have
\[ \sum_{{t\in L,\, \size{t}=n}} \lambda(t) = \frac{b_n}{n!}\]
and we recover the series $f_B$ as the labelled generating function $\widetilde{f}_{\mathcal A}$
of the automaton $\mathcal A$
\end{example}

\begin{figure*}
    
    \subfloat[Ranked binary tree (on left) encoding partition (on right)\label{fig:EX}]{
    \begin{tikzpicture}[
        nodet/.style = {draw,circle,font=\footnotesize,inner sep=0.1cm},
        nodeta/.style = {draw,circle,inner sep=0.05cm,fill}
    ]
        \node[nodet] (a) {$8$};
        \node[nodet] (b) [below left=0.5cm and 0.5cm of a] {$7$};
        \node[nodet] (c) [below right=0.5cm and 0.5cm of a] {$3$};
        \node[nodet] (d) [below left=0.5cm and 0.5cm of b] {$4$};
        \node[nodet] (e) [below right=0.5cm and 0.5cm of b] {$6$};
        \node[nodet] (f) [below right=0.5cm and 0.5cm of d] {$2$};
        \node[nodet] (g) [below =0.5cm of e] {$5$};
        \node[nodet] (h) [below =0.5cm of g] {$1$};

        \node[nodet] (i1) [below =0.5cm of h] {};
        \node[nodet] (i2) [below =0.5cm of f] {};
        \node[nodet] (i3) [below left=0.5cm and 0.5cm of d] {};
        \node[nodet] (i4) [below =0.5cm of c] {};

        \path[-]
            (a) edge (c)
            (a) edge (b)
            (b) edge (d)
            (b) edge (e)
            (d) edge (f)
            (e) edge (g)
            (g) edge (h)
            (h) edge (i1)
            (f) edge (i2)
            (d) edge (i3)
            (c) edge (i4)
        ;
        \draw[cyan,opacity=0.2,line width=0.7cm,line join=round,line cap=round] plot [smooth, tension=0.6] coordinates { (a) (c) (i4) };
        \draw[red,opacity=0.2,line width=0.7cm,line join=round,line cap=round] plot [smooth, tension=0.6] coordinates { (b) (e) (g) (h) (i1) };
        \draw[green,opacity=0.2,line width=0.7cm,line join=round,line cap=round] plot [smooth, tension=0.6] coordinates { (d) (f) (i2) };

        \node[nodeta] (d) [right = 2cm of a,label=right:{$4$}]{};
        \node[nodeta] (b) [right = 1cm of d,label=right:{$7$}]{};
        \node[nodeta] (a) [right = 1cm of b,label=right:{$8$}]{};
        \node[nodeta] (f) [below=0.7cm of d,label=right:{$2$}]{};
        \node[nodeta] (e) [below=0.7cm of b,label=right:{$6$}]{};
        \node[nodeta] (c) [below=0.7cm of a,label=right:{$3$}]{};
        \node[nodeta] (g) [below=0.7cm of e,label=right:{$5$}]{};
        \node[nodeta] (h) [below=0.7cm of g,label=right:{$1$}]{};

        \path[-]
            (a) edge (c)
            (b) edge (e)
            (d) edge (f)
            (e) edge (g)
            (g) edge (h)
        ;
        \draw[cyan,opacity=0.2,line width=0.4cm,line join=round,line cap=round] plot [smooth, tension=0.6] coordinates { (a) (c)  };
        \draw[red,opacity=0.2,line width=0.4cm,line join=round,line cap=round] plot [smooth, tension=0.6] coordinates { (b) (e) (g) (h) };
        \draw[green,opacity=0.2,line width=0.4cm,line join=round,line cap=round] plot [smooth, tension=0.6] coordinates { (d) (f)  };
    \end{tikzpicture}}
\subfloat[Ranked binary tree (on left), encoding rooted non-plane trees (on right).\label{fig:EX2}]{
    \begin{tikzpicture}[
        nodet/.style = {draw,circle,font=\footnotesize,inner sep=0.1cm},
        nodeta/.style = {draw,circle,inner sep=0.05cm,fill}
    ]
        \node[nodet] (a) {$2$};
        \node[nodet] (b) [below =0.5cm of a] {$8$};
        \node[nodet] (c) [below left=0.5cm and 1cm of b] {$6$};
        \node[nodet] (d) [below left=0.5cm and 0.5cm of c] {$5$};
        \node[nodet] (e) [below right=0.5cm and 0.5cm of c] {$7$};
        \node[nodet] (f) [below left=0.5cm and 0.2cm of e] {$1$};
        \node[nodet] (g) [below right=0.5cm and 1cm of b] {$3$};
        \node[nodet] (h) [below right=0.5cm and 0.5cm of g] {$9$};
        \node[nodet] (i) [below left=0.5cm and 0.2cm of h] {$4$};

        \node[nodet] (i1) [below left=0.5cm and 0.2cm of d] {};
        \node[nodet] (i2) [below right=0.5cm and 0.2cm of d] {};
        \node[nodet] (i3) [below left=0.5cm and 0.2cm of f] {};
        \node[nodet] (i4) [below right=0.5cm and 0.2cm of f] {};
        \node[nodet] (i5) [below right=0.5cm and 0.2cm of e] {};
        \node[nodet] (i6) [below left=0.5cm and 0.5cm of g] {};
        \node[nodet] (i7) [below left=0.5cm and 0.2cm of i] {};
        \node[nodet] (i8) [below right=0.5cm and 0.2cm of i] {};
        \node[nodet] (i9) [below right=0.5cm and 0.2cm of h] {};

        \path[-]
            (a) edge (b)
            (b) edge (c)
            (b) edge (g)
            (c) edge (d)
            (c) edge (e)
            (d) edge (i1)
            (d) edge (i2)
            (e) edge (f)
            (e) edge (i5)
            (f) edge (i3)
            (f) edge (i4)
            (g) edge (i6)
            (g) edge (h)
            (h) edge (i)
            (h) edge (i9)
            (i) edge (i7)
            (i) edge (i8)
        ;

        \coordinate[left=0.1cm of g]  (g1) ;
        \coordinate[above right=0.1cm and 0.1cm of g]  (g2) ;
        \coordinate[right=0.2cm of h]  (h2) ;

        \coordinate[left=0.1cm of c]  (c1) ;
        \coordinate[above right=0.1cm and 0.1cm of c]  (c2) ;
        \coordinate[right=0.2cm of e]  (e2) ;

        \draw[cyan,opacity=0.2,line join=round,line cap=round,fill] plot [smooth cycle, tension=0.6] coordinates { (g1) (g2) (h2) (i9.east) (i8.south east) (i7.south west) (i.west) };
        \draw[red,opacity=0.2,line join=round,line cap=round,fill] plot [smooth cycle, tension=0.6] coordinates { (c1) (c2) (e2) (i5.east) (i4.south east) (i3.south west) (f.west)};

        \def\d{1cm}
        \def\dbelow{0.95cm}

        \node[nodeta] (a) [right = 3.5cm of a,label=right:{$2$}]{};
        \node[nodeta] (d) [below left = \d and 1.2cm of a,label=right:{$5$}]{};
        \node[nodeta] (c) [below = \dbelow of a,label=right:{$6$}]{};
        \node[nodeta] (b) [below right = \d and 1.2cm of a,label=right:{$8$}]{};
        \node[nodeta] (f) [below left = \d and 0.3cm of c,label=right:{$1$}]{};
        \node[nodeta] (e) [below right = \d and 0.3cm of c,label=right:{$7$}]{};
        \node[nodeta] (g) [below = \d of b,label=right:{$3$}]{};
        \node[nodeta] (h) [below right = \d and 0.3cm of g,label=right:{$9$}]{};
        \node[nodeta] (i) [below left = \d and 0.3cm of g,label=right:{$4$}]{};

        \path[-]
            (a) edge (c)
            (a) edge (b)
            (a) edge (d)
            (c) edge (e)
            (c) edge (f)
            (b) edge (g)
            (g) edge (h)
            (g) edge (i)
        ;
        \draw[cyan,opacity=0.2,fill,line join=round,line cap=round] plot [smooth cycle, tension=0.6] coordinates { (i.south west) (g.north) (h.south east) };
        \draw[red,opacity=0.2,fill,line join=round,line cap=round] plot [smooth cycle, tension=0.6] coordinates { (f.south west) (c.north) (e.south east) };
    \end{tikzpicture}
}
    \caption{Examples of Combinatorial encodings}
\end{figure*}


\subsection{Weighted Tree Automata}

A \emph{$\K$-weighted tree automaton} $\A$ (also called a multilinear representation~\cite{BerstelR82}) is a bottom-up tree automaton $(Q,\Delta,q_f)$, with a single final state $q_f \in Q$, where the transition rules in $\Delta$ are equipped with a weight in $\K$.  Transitions are denoted \[ \sigma(q_1,q_2,\ldots,q_k) \xrightarrow{a} q , \] where $\sigma \in \Sigma_k$, $a \in \K$ and $q_1,\ldots,q_k,q \in Q$. 
A run of such an automaton is a run of the underlying non-deterministic tree automaton.  The weight of the run is the product of the weights
of the transitions comprising the run.  The value $\sem{\A}{t}$ of word $t\in T_\Sigma$ is the sum of the weights of all accepting runs on $t$.

When reasoning with weighted tree automata it is convenient to work with a matricial presentation in which a $\K$-weighted tree automaton over  alphabet $\Sigma$ is a pair $(d,\mu)$, where $d\in \mathbb N$ is the dimension and $\mu$ is a function with domain $\Sigma$ such that $\mu(\sigma) \in \K^{d^k \times d}$ is a matrix for $\sigma \in \Sigma_k$. The dimension corresponds to the number of states in the automaton, whereas $\mu$ is the matrix representation of the set of transition rules $\Delta$ with the set of states $Q = \{1,\ldots,d\}$:
\[
\sigma(q_1,\ldots,q_k) \xrightarrow{a} q \in \Delta \quad \text{iff} \quad \mu(\sigma)_{(q_1,\ldots,q_k),q} = a
\]
The value $\sem{\A}{t}$ of $t\in T_\Sigma$ is defined using the Kronecker product.
Recall that the Kronecker product of two row vectors $\boldsymbol{u} = \begin{bmatrix}
    u_1 & \cdots & u_m 
 \end{bmatrix}$ and $\boldsymbol{v} = \begin{bmatrix}
    v_1 & \cdots & v_n
\end{bmatrix}$, denoted $\boldsymbol{u} \otimes \boldsymbol{v}$, 
is defined by 
\[ 
\begin{bmatrix}
    u_1v_1 & u_1v_2 & \cdots & u_mv_{n-1} & u_mv_n
\end{bmatrix} \,  . \]
The operation is associative and the nullary Kronecker product is the $1\times 1$ identity matrix.

The map $\mu$ induces a function $\hmu:T_\Sigma \rightarrow \K^{1\times d}$ that is defined inductively by
\begin{gather*} \hmu(\sigma(t_1,\ldots,t_k)) :=
(\hmu(t_1)\otimes \cdots \otimes \hmu(t_k)) \cdot \mu(\sigma) 
\label{eq:sum-prod}
\end{gather*}
for a  $k$-ary symbol $\sigma$. In other words, $\hmu(t)$ is given by an iterated matrix product determined by parsing the tree $t$ from the leaves to the root. From the properties of the Kronecker product, $\hmu(t)$ is a vector $\begin{bmatrix} w_1 & \ldots & w_d \end{bmatrix}$ where $w_i$ is the sum of the weights of all runs over $t$ that label the root with state~$i$. Without loss of generality, we will consider that the final state $q_f$ corresponds to index 1, meaning that $\sem{\A}{t} = \mu(t)_1$. By convention, in this paper, a vector will be represented by a bold symbol, say $\boldsymbol{v}$, and its elements will be denoted $v_1, v_2, \ldots$.

The generating function of a weighted tree automaton, given by \Cref{eq:ordinary}, is an algebraic power series~\cite[Proposition 7.2]{BerstelR82}.
Note that two automata that represent different tree series can have the same generating function, since the latter aggregates all trees of the same size.


\subsection{Differential Tree Automata}
The notion of differential tree automata generalises that of weighted tree automata.  
A \emph{$\K$-differential tree automaton} over alphabet $\Sigma$ is 
a pair $\mathcal A=(d,\mu)$,  where $d\in \mathbb N$ is the dimension and 
$\mu$ is a map with domain $\Sigma$ such that 
\begin{itemize}
    \item $\mu(\sigma)\in \K^{1\times d}$ for $\sigma \in \Sigma_0$ and,
    \item $\mu(\sigma)\in \K(x)^{d^\ell\times d}$ is a matrix of rational functions, having no positive integer poles,
 for $\ell \geq 1$ and $\sigma \in \Sigma_\ell$.
\end{itemize}

The map $\mu$ induces a function  $\hmu:T_\Sigma \rightarrow \K^{1\times d}$, inductively defined by specifying that for all
$\sigma \in \Sigma_0$, 
$\hmu(\sigma) := \mu(\sigma)$; and for all $k\geq 1$ and $\sigma \in \Sigma_k$ and $t_0=\sigma(t_1,\ldots,t_k)$,
\begin{gather*} \hmu(t_0) :=
(\hmu(t_1)\otimes \cdots \otimes \hmu(t_k)) \cdot \mu(\sigma) (\|t_0\|) \, .
\label{eq:sum-prod2}
\end{gather*}
We define the formal tree series represented by $\mathcal A$ to be the function $[\![\mathcal A]\!] : T_\Sigma \rightarrow \K$ defined by $\sem{\A}{t} := \hmu(t)_1$.

We also introduce the notion of a generalised differential tree automaton, we will see later that this does not change the expressiveness.
To set up this definition, we first define
 $\KK(x_0,\ldots,x_k)$ to be the subring of $\K(x_0,\ldots,x_k)$
consisting of rational functions of the form
\[ \frac{P(x_0,\ldots,x_k)}{Q_0(x_0) \cdots Q_k(x_k)}\]
where $P \in \K[x_0,\ldots,x_k]$ is a multivariate polynomial and $Q_i \in \K[x_i]$ for $i \in \{0,\ldots,k\}$ are univariate polynomials such that $Q_0$ has no positive integer root and $Q_1,\ldots,Q_k$
have no nonnegative integer root.

A  generalised $\K$-differential tree automaton is defined as above except that for $\sigma \in \Sigma_{\ell}$ with $\ell\geq 1$ we now take  
 $\mu(\sigma)$ to be an element of $ \KK(x_0,\ldots,x_k)^{d^k\times d}$.
Such an automaton likewise induces a function $\hmu:T_\Sigma \rightarrow \K^{1\times d}$, which is defined as above except that now we have 
\begin{gather*} \hmu(t_0) :=
(\hmu(t_1)\otimes \cdots \otimes \hmu(t_k)) \cdot \mu(\sigma) (\|t_0\|,\ldots,\|t_k\|) \, .
\label{eq:sum-prod2g}
\end{gather*}

\begin{restatable}{proposition}{generalised}
\label{prop:generalised}
For every generalised differential tree automaton~$\mathcal A$, there exists a 
differential tree automaton~$\mathcal B$ such that~$[\![\mathcal A]\!]=[\![\mathcal B]\!]$. 
\end{restatable}
The proof is provided in Appendix~\ref{app:generalised}. Henceforth we drop the term \emph{generalised} when referring to differential tree automata.

The definition of the \emph{generating function} $\genS{\A}$ and \emph{labelled generating function} $\LgenS{\A}$ of a differential tree automaton $\A$ is exactly as for a weighted tree automaton, namely via \Cref{eq:ordinary,eq:labelled}. 
In general $\genS{\A}$ need not be algebraic, unlike for weighted tree automata, but we will show that it is differentially algebraic.


\begin{example}[Labelled trees]
\label{ex:labelled-trees}
Let \[(t_n)_{n=0}^\infty=(0,1,2,9,\allowbreak 64,625,\ldots)\] be the sequence whose general term counts the number of labelled rooted 
non-plane trees with $n$ nodes and write  
\[f_T(x)=\sum_{n=0}^\infty \frac{t_n}{n!}x^n\] for its exponential generating function.
From the recursive definition of labelled rooted trees as a combinatorial species we have that $f_T = x \exp(f_T)$.  
Differentiating, we see that $f_T$ is a solution of the differential equation
\begin{equation}
    x(1-y)y' = y\,.
    \label{eq:ex-labelled trees}
\end{equation}
From there, the sequence $(a_n)_{n=0}^\infty = (\frac{t_n}{n!})_{n=0}^\infty$ satisfies the non-linear recurrence relation
\begin{equation}
\label{eq:recurrence labelled tree}
a_{n+1} = \frac{1}{n} \sum_{k=0}^{n-1} (n-k)a_{k+1}a_{n-k}
\end{equation}

From this recurrence relation we obtain a differential tree automaton $\mathcal A$
over alphabet $\Sigma=\{\sigma_0,\sigma_1,\sigma_2\}$, 
with states $q_1,q_2$ and transitions
\[ \sigma_0\stackrel{1}{\longrightarrow} q_2,\qquad \sigma_1(q_2) \stackrel{1}{\longrightarrow}  q_1, \qquad 
\sigma_2(q_2,q_2) \stackrel{\frac{x_2+1}{x_0}}{\longrightarrow} q_2 
\]
such that $f_{\mathcal A}=f_T$. 

In its algebraic representation, the automaton $\A$ has dimension 2 with the transition function $\mu$ defined as $\mu(\sigma_0) = \begin{bmatrix} 0 & 1 \end{bmatrix}$ and
\[
\mu(\sigma_1) = 
\begin{bmatrix}
0 & 0 \\
1 & 0 \\
\end{bmatrix}
\quad \text{and} \quad
\mu(\sigma_2) = 
\begin{bmatrix}
\vzero_{3 \times 1} & \vzero_{3 \times 1} \\
0  & \frac{x_2+1}{x_0}\\
\end{bmatrix}
\]
With a simple induction, we notice that for all $t \in T_\Sigma$, \[\hmu(t) = \begin{bmatrix} \frac{t_n}{n!} & \frac{t_{n+1}}{(n+1)!} \end{bmatrix}\,.\]

The equality $f_{\mathcal A}=f_T$ can also be established using combinatorial reasoning.
The set \[\{ t \in T_\Sigma : \sem{\A}{t}\neq 0\}\] comprises those trees $t\in T_\Sigma$
in which the root has arity one and all other internal nodes have arity two.
Every such tree $t$ encodes a rooted plane tree~$\tau$ (see Figure~\ref{fig:EX2}) as follows: 
\begin{itemize}
    \item the nodes of $\tau$ are the non-leaf nodes of $t$,
    \item the root of $\tau$ is the root of $t$, 
    \item the two children of a binary node in $t$ are respectively the next sibling and first child of that node in $\tau$.
\end{itemize}  

Given a rooted non-plane tree $\tau$ with $n$ nodes, let $S \subseteq T_\Sigma$ be the set of all binary trees that encode
some plane tree obtained by imposing an order on the children of each node in $\tau$.
If $G$ is the group of automorphisms of $\tau$, then each tree $t\in S$ encodes 
$|G|$ different
plane trees.  
By construction of $\mathcal A$
we have
\[|G| \cdot \sum_{t\in S} \sem{\A}{t} = 1 \, .\] 
Thus the number of labellings of $\tau$ is 
$\frac{n!}{|G|}=n!\sum_{t\in S} \sem{\A}{t}$.  Summing over all $\tau$  of size $n$,
it follows that \[\sum_{\substack{t\in T_\Sigma\\\size{t}=n}} \sem{\A}{t}=\frac{t_n}{n!}\] and hence 
$f_{\mathcal A}=f_T$.
\end{example}

In the previous example the power series $f_T$ satisfied the equation $y' = \frac{y}{x(1-y)}$ (see \Cref{eq:ex-labelled trees}).  This does
not yet imply that $f_T$ is RDA since 
$\frac{y}{x(1-y)}$ is not defined at $(0,y(0))$. However, by introducing a second power series variable $z$ such that $y = x\,z$, we can see that \Cref{eq:ex-labelled trees} is equivalent to 
\[
y' = z + \, \frac{x\, z^2}{1- x\,z} \qquad z' =  \frac{z^2}{1- x\,z}
\]
which shows that $f_T$ is RDA.

One approach to show that a power series $f$ is RDA consists of differentiating once the polynomial 
differential equation satisfied by $f$ to obtain a new polynomial equation where the degree of highest order is always one. 


\begin{example}
Consider the solution $y_1(x)$ 
of the differential equation \[(y_1')^3 + y_1^3 = 1\] such that $y_1(0) = 0$ and $y'_1(0) = 1$.  Differentiating the 
equation we deduce that \[3(y'_1)^2y''_1 + 3(y_1)^2y'_1 = 0\,,\] which leads to the following rational dynamical system:
\[
y'_1 = y_2 \qquad \qquad y'_2 = - \frac{y^2_1}{y_2}
\]
As $y_2(0) = y'_1(0) \neq 0$, the rational functions in the system are well defined at $(0,y_1(0),y_2(0))$, implying that $y_1(x)$ is RDA. The proof of  \Cref{prop:RDA implies CDA,prop:CDA implies multi-holonomic degree 0} constructs a differential tree automaton $\A$ over alphabet 
\[\Sigma=\{\sigma_0,\sigma_1,\sigma_3,\sigma_5\}\] with states $q_1,q_2,q_3$ and transitions
\begin{align*}
\sigma_0 \xrightarrow{1} q_2 &\qquad \qquad \sigma_0 \xrightarrow{1} q_3\\
\sigma_1&(q_2) \xrightarrow{\frac{1}{x_0}} q_1\\
\sigma_3&(q_1,q_1,q_3) \xrightarrow{-\frac{1}{x_0}} q_2\\
 \sigma_5&(q_1,q_1,q_3,q_3,q_3) \xrightarrow{\frac{1}{x_0}} q_3
\end{align*}
such that $f_\A(x) = y_1(x)$.
Running the automaton, we compute its generating  function $f_{\mathcal A}$ as:
\[
x - \frac{2}{4!}x^4 - \frac{20}{7!}x^7 - \frac{3320}{10!}x^{10} - \frac{1598960}{13!}x^{13} - \ldots\,.
\]
\end{example}

Thus far, all power series presented in this section were RDA.  
However, not all differentially algebraic power series are RDA. Indeed, we show in \Cref{prop:RDA implies CDA} that RDA power series actually coincide with the class of constructibly differentiably algebraic (CDA) series, the latter being incomparable with the class of differentially finite power series. (It was shown in \cite[Example 2.5]{Stanley80} that $\frac{1}{\cos(x)}$ is not D-finite but CDA, whereas $\sum_{n=0}^\infty n! x^x$ is D-finite but not CDA.) Our automata allow us to go strictly beyond these two classes.


\begin{example}
Consider the power series $f(x) = \sum_{n=0}^\infty b_n x^n$ satisfying the equation 
\begin{equation}
\label{eq:example diff}
y - 1 = x^2 y' y + xy^2 \, .
\end{equation}
Below we show that $f(x)$ is neither D-finite nor CDA.

From~\eqref{eq:example diff} the sequence $(b_n)_{n=0}^\infty$ satisfies recurrence
\begin{equation}
\label{eq:example not CDA}
b_n = \sum_{k=0}^{n-1} (n-k) b_k b_{n-k-1}
\end{equation}
with $b_0 = 1$.   This ensures that $b_n\geq n!$ for all $n$.

If $f$ were CDA then we would have $b_n \leq \alpha^n$~ for  some constant~$\alpha>0$~\cite[Theorem 3]{Bergeron1990CombinatorialRO}.  Hence $f$ is not CDA.
Moreover, differentiating \Cref{eq:example diff}, we have that for each positive integer $k$, $f$ satisfies an equation of the form $y^{(k)} = \sum_{i=0}^{2k} P_{k,i}(x)y^{i+1-2k}$ 
with $P_{k,0}(x)$ not identically 
Suppose for a contradiction that
$f(x)$ is differentially finite. Then it satisfies an equation of the form:
\[
\sum_{k=0}^d Q_k(x)y^{(k)} =0\,.
\]
which, together with the above expression for each $y^{(k)}$, entails that $f(x)$ is the root of a non-zero polynomial $P(x,y)$, i.e., $f$ is algebraic.  But the growth of an algebraic power series is $O(c^n n^d)$ for some constants $c,d$, which contradicts $a_n \geq n!$.

Notice that the similarity between the recurrences
\Cref{eq:recurrence labelled tree,eq:example not CDA}.
This is explained by the fact that the corresponding power series are respectively ordinary and labelled generating functions of a single common automaton. 
Specifically, define 
automaton $\A$ over the alphabet $\Sigma = \{ \sigma_0,\sigma_2\}$ with states $q_1$ and transitions
\[ \sigma_0\stackrel{1}{\longrightarrow} q_1, \qquad 
\sigma_2(q_1,q_1) \stackrel{x_2+1}{\longrightarrow} q_1
\]
The power series $f(x)$ is the ordinary generating function of~$\A$.  On the other hand,  the following power series \[\sum_{n=0}^\infty \frac{t_{n+1}}{(n+1)!} x^n\] with $(t_n)_n^\infty$ as defined in \Cref{ex:labelled-trees} is the labelled generating function of $\A$.
This observation is generalized in~\Cref{prop:automata to recurrence}.
\end{example}


\subsection{Multi-holonomic sequences}

Motivated by the relation between differentially algebraic power series and recurrence relations in the previous examples, we introduce the notion of a \emph{multi-holonomic recurrence} for a sequence of vectors $(\boldsymbol a_n)_{n=0}^\infty$ in $\K^d$.


Fix $d \in \mathbb N$.  A sequence $(\boldsymbol a_n)_{n=0}^\infty$ of elements of $\K^d$ is \emph{multi-holonomic}
if there exist $m \in \mathbb N$, and for all $\ell\in\{1,\ldots,m\}$ a matrix of polynomials
\[\mu_\ell \in \KK(x)^{d^{\ell} \times d} \, , \]  such that for all $n \geq 1$ we have
\begin{equation}
\label{eq:multiholonomic}
  \boldsymbol a_n =  \, \sum_{\substack{\ell \in \{1,\ldots,m\}\\\bm{j} \in J_n(\ell)}} (\va_{j_1}\otimes \ldots \otimes \va_{j_\ell})  \cdot\mu_\ell(n) \,  ,
\end{equation}
where $J_n(\ell)$ is the set of tuples $(j_1,\ldots,j_\ell) \in \N^{\ell}$ such that $j_1 + \ldots + j_\ell = n-1$.
We call~$d \in \mathbb N$ the \emph{order} 
of the recurrence and $m$ its arity. The \emph{degree} of the recurrence is the maximum degree of the numerator and  denominator  over all entries of the matrices $\mu_\ell$.

Syntactically speaking, a multi-holonomic recurrence can be regarded a differential tree automaton over an alphabet with exactly one symbol of each arity.
The difference is that we view the recurrence as defining a sequence of vectors rather than a formal tree series.
As with automata, we introduce a generalisation of multi-holonomic recurrences that permits additional predicates on the summation indices.
Formally, a sequence $(\boldsymbol a_n)_{n=0}^\infty$ of elements of $\K^d$ is \emph{generalised multi-holonomic} if there exist a set $I \subseteq  \bigcup_{\ell\geq 0} \Z \times  \N^\ell$ and for all tuples $\bm{s} = (s_0,\ldots,s_\ell)\in I$ a matrix of polynomials \[\mu_{\bm{s}} \in \KK(x_0,\ldots,x_\ell)^{d^\ell \times d}\] such that for all $n \geq 1$,
\[
\boldsymbol a_n = \sum_{\substack{\bm{s}  \in I \\ \bm{j} \in G_{n}(\bm{s})}} (\va_{j_1} \otimes \ldots \otimes \va_{j_\ell}) \cdot \mu_{\bm{s}}(n,\bm{j} ) \, ,
\]
where $G_n(s_0,\ldots,s_\ell)$ is the set of tuples $(j_1,\ldots, j_\ell) \in \N^\ell$ such that $j_1 + \ldots + j_\ell = n - 1 + s_0$ and for all $k \in \{1, \ldots, \ell\}$, $j_k \leq n-1-s_k$.

A multi-holonomic sequence of arity $m$ defined by matrices $\mu_1, \ldots, \mu_m$ can be seen as a generalised multi-holonomic sequence by taking the set $I = \{ \bm{0}_{2}, \ldots,\bm{0}_{(m+1)} \}$, with $\bm{0}_i$ being the tuple composed of $i$ integers $0$, and for all $\ell \in \{1,\ldots,m\}$, \[\mu_{\bm{0}_{(\ell+1)}}(x_0,\ldots,x_\ell) = \mu_{\ell}(x_0)\,.\]

The third main result of our paper compares these notions of multi-holonomic sequences. For this purpose,
given $d_1\geq d_2$ 
we say that a sequence $(\vb_n)_{n=0}^\infty$ of elements of $\K^{d_1}$ \emph{recognises} a sequence $(\va_n)_{n=0}^\infty$ of elements of $\K^{d_2}$ 
if $\pi(\vb_n)=\va_n$ for all $n$, where 
$\pi : \K^{d_1}\rightarrow \K^{d_2}$ is the projection onto the first $d_2$ coordinates.
The following result is a useful step in the proof of \Cref{thm:MAIN2}.
Detailed proofs can be found in Appendix~\ref{app:generalised-multi-holonomic}.

\begin{restatable}{theorem}{thequivalencerecurrence}
\label{thm:MAIN3}
Let $(\va_n)_{n=0}^\infty$ be a sequence of elements of $\K^d$. The following statements are equivalent:
\begin{enumerate}
\item a multi-holonomic sequence of arity 2 recognises $(\va_n)_{n=0}^\infty$;
\item a multi-holonomic sequence recognises $(\va_n)_{n=0}^\infty$;
\item a generalised multi-holonomic sequence recognises $(\va_n)_{n=0}^\infty$.
\end{enumerate}
\end{restatable}


\section{Proof of \texorpdfstring{\Cref{thm:MAIN1}}{Theorem~\ref{thm:MAIN1}}}

In this section we provide  an overview of the proof of \Cref{thm:MAIN1}. Each main step 
is formulated as a separate proposition, accompanied by a brief proof sketch. Complete proofs are presented in detail in Appendix~\ref{app:theorem1}.

The first step of the proof of \Cref{thm:MAIN1} consists of showing that constructibly differentiably algebraic (CDA) power series coincide with rationally dynamically algebraic (RDA) power series. Recall that a RDA power series $f(x)$ arises as the first component of a solution $(y_1(x),\ldots,y_k(x))$ of a system of differential equations
\begin{gather}
y_1'=\frac{P_1(y_1,\ldots,y_k)}{Q_1(y_1,\ldots,y_k)}, \, \ldots, \, y'_k=\frac{P_k(y_1,\ldots,y_k)}{Q_k(y_1,\ldots,y_k)},
\label{eq:RDS}
\end{gather}
where $P_1, Q_1,\ldots,P_k,Q_k \in \K[x,y_1,\ldots,y_k]$ are polynomial functions with $Q_1,\ldots,Q_k$ defined at $(0,y_1(0),\ldots,y_k(0))$.
 A power series is CDA when $Q_1 = \ldots = Q_k = 1$.
The CDA power series are thus RDA power series by definition. The converse is given by the following proposition.


\begin{restatable}{proposition}{propRDAimpliesCDA}
\label{prop:RDA implies CDA}
The class of  RDA and CDA power series coincide.
\end{restatable}

\begin{proof}[Proof sketch]
Our aim is to  transform the system shown in~\eqref{eq:RDS} into an equivalent one in which all denominators are constant.  To this end 
we introduce  new variables 
$z_i$, for all $i \in \{1, \ldots, k\}$,  with the equations \[ z_i = \frac{1}{Q_i(y_1,\ldots,y_k)} \quad\text{and}\quad y'_i = P_i(y_1,\ldots,y_k) z_i\,.\]  Differentiating these equations leads to \[ z'_i = -z_i^2 \sum_{j=1}^k \frac{\partial Q_i}{\partial y_j} y'_j = -z_i^2 \sum_{j=1}^k \frac{\partial Q_i}{\partial y_j} P_j(y_1,\ldots,y_k) z_j\,.\]
But $\frac{\partial Q_i}{\partial y_j} $ is also a polynomial $\K[x,y_1,\ldots,y_k]$,  which concludes the proof.
\end{proof}

In combination with~\Cref{prop:RDA implies CDA}, the following result establishes
the implication \ref{enum:thm:MAIN1-RDA} $\Rightarrow$ \ref{enum:thm:MAIN1-ordinary} in \Cref{thm:MAIN1}.


\begin{restatable}{proposition}{propCDAtoholonomic}
\label{prop:CDA implies multi-holonomic degree 0}
Let $f(x)$ be a CDA power series defined by the system of  differential equations 
\begin{equation}
\label{eq:prop:CDS implies multi}
y_1'= P_1(y_1,\ldots,y_k),\ldots, y'_k = P_k(y_1,\ldots,y_k)\,.
\end{equation}
Then 
 $f(x)$ is the ordinary generating function of a
differential tree automaton in which the rational weight of 
 every transition  has the form $\frac{a}{x}$ for some $a\in \K$.
\end{restatable}
    
\begin{proof}[Proof sketch]
We can assume that 
the polynomials $P_1,\ldots,P_k$ lie in  $\K[y_1,\ldots,y_k]$ instead of $\K[x][y_1,\ldots,y_k]$. This assumption is  without loss of generality as one can rewrite 
 all polynomials  $P_1,\ldots,P_k$ 
by  replacing every occurrence of $x$ in these polynomials by a new power series variable $z$, and
by adding the equation 
 $z' = 1$. 
 
Denote by $m$  the maximal total degree of the polynomials~$P_i$. Let the tuple 
 $(y_1(x),\ldots,y_k(x))$ of power series be a solution of the system and write $y_j(x) = \sum_{n=0}^\infty a_{n,j} x^n$ for all~$j \in \{1, \ldots, k\}$. Substituting  this solution in  the defining system of equations in \eqref{eq:prop:CDS implies multi} entails that  for all $n \geq 1$, $n \, a_{n,j}$  equals the $(n-1)$-th coefficient of the power series $P_j(y_1(x),\ldots,y_k(x))$. 

To represent the coefficients of the $P_j$s, we decompose each polynomial $P_j$ into the sum of its monomials as follows:
\[
    P_j(y_1,\ldots,y_k) = \sum_{\ell=0}^m \sum_{\bm{i} \in \{1, \ldots,k\}^\ell} \alpha_{\bm{i},j} y_{i_1} \ldots y_{i_\ell} ,
\]
with all $\alpha_{\bm{i},j}$s being in $\K$. 
Hence, the $(n-1)$-th coefficient of each monomial $\alpha_{\bm{i},j} y_{i_1}(x) \ldots y_{i_\ell}(x)$ is given by 
\[
\sum_{\bm{n} \in J_n(\ell)} \alpha_{\bm{i},j} a_{n_1,i_1}\ldots a_{n_\ell,i_\ell}\,,
\]
where $J_n(\ell)$ is the set of tuples $(j_1,\ldots,j_\ell) \in \N^{\ell}$ such that $j_1 + \ldots + j_\ell = n-1$.

Following the above observation, we build a differential tree automaton $\A = (k+1,\mu)$ over an alphabet $\Sigma$  such that the $n$-th coefficient of $\vf_\A(x)$ is the vector $\begin{bmatrix} a_{n,1} & \ldots & a_{n,k} & 1 \end{bmatrix}$. For this, when $\ell \geq 1$, for each monomial $\alpha_{\bm{i},j} a_{n_1,i_1}\ldots a_{n_\ell,i_\ell}$, we introduce a symbol $\sigma_{\bm{i},j}$ in the alphabet $\Sigma$ with the transition rule
\[
\sigma_{\bm{i},j}(q_{i_1},\ldots,q_{i_\ell}) \xlongrightarrow{\alpha_{\bm{i},j}} q_j
\]
The $k+1$ component in the vector $\begin{bmatrix} a_{n,1} & \ldots & a_{n,j} & 1 \end{bmatrix}$ is in fact used to handle the case $\ell = 0$ by introducing an additional unary symbol. 
\end{proof}

To establish the remaining implications in~\Cref{thm:MAIN1} we use~\Cref{prop:automata to recurrence}, which shows  that the labelled and ordinary generating function of an automaton induces multi-linear recurrence relations that solely differ by a factor of $\frac{1}{n}$.

Given a differential tree automaton $\A = (d,\mu)$, 
 define  
$\vf_\A(x) = \sum_{n=0}^\infty \va_n x^n$ and $\widetilde{\vf}_\A(x) = \sum_{n=0}^\infty \vb_n x^n$, where
\[
\va_n = \sum_{\substack{t\in T_\Sigma\\ \size{t}=n}} \hmu(t) \quad \text{and} \quad \vb_n = \sum_{\substack{t\in T_\Sigma\\ \size{t}=n}} \lambda(t)\hmu(t)\,.
\]
Observe that $f_{\A,1}(x)$ is the ordinary generating function~$f_\A(x)$, and similarly $\widetilde{f}_{\A,1}(x)$ is the labelled generating function~$\widetilde{f}_{\A}(x)$.  More generally,
the power series $f_{\A,i}(x)$ and $\widetilde{f}_{\A,i}(x)$ can be seen respectively as the ordinary and labelled generating functions of $\A$ when $q_i$ is considered the final state.


\begin{restatable}{proposition}{propautomatatorecurrence}
\label{prop:automata to recurrence}
Let $\A = (d,\mu)$ be a differential tree automaton over $\Sigma$ of  maximal arity $m$. Write
$\vf_\A(x) = \sum_{n=0}^\infty \va_n x^n$ and $\widetilde{\vf}_\A(x) = \sum_{n=0}^\infty \vb_n x^n$ (as above).
Then 
\[\vb_0 = \va_0 = \sum_{\sigma \in \Sigma_0} \mu(\sigma)\,,\] 
and for all $n \geq 1$
\begin{align*}
\va_n &= \sum_{\substack{\ell \in \{1, \ldots,m\}\\\bm{j} \in J_n(\ell)}} (\va_{j_1} \otimes \ldots \otimes \va_{j_\ell}) \cdot \sum_{\sigma \in \Sigma_\ell}\mu(\sigma)(n)\\
\vb_n &= \frac{1}{n}\sum_{\substack{\ell \in \{1, \ldots,m\}\\\bm{j} \in J_n(\ell)}} (\vb_{j_1} \otimes \ldots \otimes \vb_{j_\ell}) \cdot \sum_{\sigma  \in \Sigma_\ell} \mu(\sigma)(n) \,
\end{align*}
where  $J_n(\ell)$ is the set of tuples $(j_1,\ldots,j_\ell) \in \N^{\ell}$ such that $j_1 + \ldots + j_\ell = n-1$.

\end{restatable}

\begin{proof}[Proof sketch]
The key idea of the proof is to rewrite the sum~$\sum_{\substack{t\in T_\Sigma\\ \size{t}=n}} \hmu(t)$, for trees~$t\in T_\Sigma$ with $\size{t}\geq 1$, as 
a summation over possible root symbol and subtrees. 
Since $t$ has size larger than one, it is of the form $\sigma(t_1,\ldots,t_\ell)$ for some $\sigma \in \Sigma_\ell$. Hence
\[
\va_n = \sum_{\ell=1}^m \sum_{\substack{\sigma \in \Sigma_\ell\\\bm{j} \in J_n(\ell)}} \sum_{\substack{t_1\in T_\Sigma\\ \size{t}=j_1}} \ldots \sum_{\substack{t_\ell\in T_\Sigma\\ \size{t}=j_\ell}} \hmu(\sigma(t_1,\ldots,t_\ell))\,.
\]
Unfolding the definition of $\hmu(\sigma(t_1,\ldots,t_\ell))$ and using the bilinearity of the Kronecker product we derive the claimed  recurrence for~$\va_n$.

The claimed recurrence relation for $\vb_n$ follows 
by similar reasoning combined with the following  additional observation.
Since  the multinomial coefficient
$\binom{n-1}{j_1,\ldots,j_\ell}$  represents the number of partitions of $\{1,\ldots,n-1\}$ into subsets of respective sizes 
$j_1,\ldots,j_\ell$,
the number of labellings of the tree~$t$ equals the product of
$\binom{n-1}{j_1,\ldots,j_\ell}$ with the respective numbers of labellings of~$t_1,\ldots,t_\ell$.
It follows that
\[ \lambda(\sigma(t_1,\ldots,t_\ell))= \frac{1}{n} \lambda(t_1) \cdots \lambda(t_\ell) \,. \qedhere\]
\end{proof}

The sequences $(\va_n)_{n=0}^\infty$ and $(\vb_n)_{n=0}^\infty$ defined in~\Cref{prop:automata to recurrence} 
are multi-holonomic  with defining matrices $\mu_\ell = \sum_{\sigma \in \Sigma_\ell} \mu(\sigma)$ for all $\ell \in \{0, \ldots, m\}$. 

The implication \ref{enum:thm:MAIN1-ordinary} $\Rightarrow$ \ref{enum:thm:MAIN1-RDA}  in~\Cref{thm:MAIN1} follows
 directly from~\Cref{prop:automata to recurrence} and~\Cref{prop:multi-holonomic degree 1 to RDA}, given below. 

\begin{restatable}{proposition}{propHolonomicOneToRDA}
\label{prop:multi-holonomic degree 1 to RDA}
Let $(\va_n)_{n=0}^\infty$ be a multi-holonomic sequence of order $d$ and degree at most $1$ with common denominator $x$. The  power series of coordinates of~$\vf(x) = \sum_{n=0}^\infty \va_n\ x^n$ are all RDA.
\end{restatable}

\begin{proof}[Sketch proof]
Denote by $\SDerive{f}$  the differential operator defined by $\SDerive{f}(x) := x f'(x)$. 

We show that the vector $\vf(x)$ of power series is in fact a solution of a system of differential equations of the form 
\[
\begin{cases}
\SDerive{y_1} = x Q_1(y_1,\ldots,y_d) + \sum_{i = 1}^d \SDerive{y_i }\cdot xP_{1,i}(y_1,\ldots,y_d) \\
 \vdots\\
\SDerive{y_d}  = x Q_d(y_1,\ldots,y_d) + \sum_{i = 1}^d \SDerive{y_d}\cdot xP_{d,i}(y_1,\ldots,y_d) \\
\end{cases}
\]
where $Q_i,P_{i,j}$s are polynomials in $\K[y_1,\ldots,y_d]$. This can be seen as a linear system of equations in $\SDerive{f_1},\ldots,\SDerive{f_d}$ with coefficients in the field of rational functions $\K(x,f_1,\ldots,f_d)$. As such, we can compute the determinant $\det(M)$, where $M$ represents the coefficient matrix of the system. Crucially we show  that $\det(M)$  is  a non-zero polynomial in $\K[x,f_1,\ldots,f_d]$. Hence, we have \[M^{-1} = \frac{1}{\det(M)}\,\mathrm{adj}(M)\, .\]
Furthermore, we prove that  $\det(M)$ does not vanish at $(0,f_1(0),\ldots,f_d(0))$, which in turn implies that
\[
\begin{bmatrix}
\SDerive{f_1}\\
\vdots\\
\SDerive{f_d}\\
\end{bmatrix}
= 
M^{-1}
\begin{bmatrix}
xQ_1\\
\vdots\\
xQ_d\\
\end{bmatrix}
= x \begin{bmatrix}
    U_1\\
    \vdots\\
    U_d\\
    \end{bmatrix}
\]
where the rational functions $U_1,\ldots,U_d \in \K(x,f_1,\ldots,f_d)$ are defined at $(0,f_1(0),\ldots,f_d(0))$. This concludes the proof.
\end{proof}

We complete the proof of \Cref{thm:MAIN1} by showing the equivalence \ref{enum:thm:MAIN1-ordinary}  $\Leftrightarrow$ \ref{enum:thm:MAIN1-labelled} in the theorem.
Given any multi-holonomic sequence $(\va_n)_{n=0}^\infty$  of order $d$ and arity $m$, defined by matrices $\mu_1$, \ldots, $\mu_m$, \Cref{prop:automata to recurrence} implies  that the 
sequence arises as 
 the ordinary  generating functions~$\vf_\A(x)$  of some differential tree automaton~$\A$. 
More precisely, define $\A= (d,\mu)$ over the alphabet $\Sigma = \{ \sigma_0, \ldots, \sigma_m\}$ such that $\sigma_i$ has arity $i$, $\mu(\sigma_0) = \va_0$, and
\[
\mu(\sigma_\ell)(x) = \mu_\ell(x)\quad 
\]
for all $\ell \in \{1, \ldots, m\}$.
Replacing  transition matrix  $\mu(\sigma_\ell)(x)$
in~$\A$ with $\mu(\sigma_\ell)(x) = x \cdot \mu_\ell(x)$ gives another automaton~$\B$ whose  
labelled generating functions~$\widetilde{\vf}_\B(x)$  is  induced by~$(\va_n)_{n=0}^\infty$. 

Since the proof of the equivalence \Cref{enum:thm:MAIN1-RDA} $\Leftrightarrow$ \Cref{enum:thm:MAIN1-ordinary} in the theorem relies on~\Cref{prop:CDA implies multi-holonomic degree 0}, we can now establish \ref{enum:thm:MAIN1-ordinary}  $\Rightarrow$ \ref{enum:thm:MAIN1-labelled} by restricting our attention to  differential tree automata in which the rational weight of 
 every transition  has the form $\frac{a}{x}$ for some~$a\in \K$.  
Given such a differential tree automaton $\A$
we apply \Cref{prop:automata to recurrence} to show that the vector of power series $\vf_{\A}(x) = \sum_{n=0}^\infty \va_n x^n$ induces a multi-holonomic sequence $(\va_n)_{n=0}^\infty$ in which every entry of the defining matrices also has the form $\frac{a}{x}$ for some~$a\in \K$. We can thus apply the above observation to show the existence of a $\K$-weighted automaton $\B$ whose labelled generating function $\widetilde{\vf}_\B(x)$ equals $\vf_{\A}(x)$. 
The converse direction \ref{enum:thm:MAIN1-labelled} $\Rightarrow$ \ref{enum:thm:MAIN1-ordinary} is analogous (application of \Cref{prop:automata to recurrence} followed by above observation).


\section{Proof of \texorpdfstring{\Cref{thm:MAIN2}}{Theorem~\ref{thm:MAIN2}}}

This theorem shows the equivalence of differentially algebraic power series, ordinary generating function of differential tree automata, and multi-holonomic sequences of arity~2. 

The most challenging development in this theorem is to establish  \ref{enum:thm:MAIN2-D-algebraic} $\Leftrightarrow$ \ref{enum:thm:MAIN2-recurrence}. 
Before discussing this equivalence, we outline the key    
 ingredients required to establish   \ref{enum:thm:MAIN2-ordinary} $\Leftrightarrow$ \ref{enum:thm:MAIN2-recurrence} in \Cref{thm:MAIN2}. 
 Detailed proofs can be found in Appendix~\ref{app:theorem2}.
 
By \Cref{prop:automata to recurrence}, the ordinary generating function of any differential tree automaton $\A$ is the first component of the vector of power series $\vf_\A(x) = \sum_{n=0}^\infty \va_n x^n$. The latter induces a multi-holonomic sequence $(\va_n)_{n=0}^\infty$. 
 By \Cref{thm:MAIN3}, this sequence is recognised by a multi-holonomic sequence of arity 2 which concludes the proof of \ref{enum:thm:MAIN2-ordinary} $\Rightarrow$ \ref{enum:thm:MAIN2-recurrence} in \Cref{thm:MAIN2}. 
We note in passing that the main milestone of the above argument is addressed in~\Cref{thm:MAIN3}. 

Conversely, consider the sequence $(\vb_n)_{n=0}^\infty$ of elements in~$\K^d$ satisfying for all $n \geq 1$ the recurrence 
\begin{equation}
\label{eq:th3:recurrence}
\vb_{n} = \vb_{n-1} \cdot P(n) + \sum_{k=0}^{n-1} (\vb_k \otimes \vb_{n-k-1})\cdot Q(n)
\end{equation}
with $P$ and $Q$ having entries defined on positive integers. We define the alphabet $\Sigma = \{ \sigma_0, \sigma_1, \sigma_2 \}$ with $\sigma_0$, $\sigma_1$ and $\sigma_2$ having, respectively, arity $0$, $1$ and $2$. Finally, we define the differential tree automaton $\A = (d,\mu)$ over $\Sigma$ such that $\mu(\sigma_0) = \vb_0$, and
\[
    \mu(\sigma_1)(x) = P(x)\quad \text{and}\quad \mu(\sigma_2)(x) = Q(x)\,.
\]
From \Cref{prop:automata to recurrence}, writing $\vf_\A(x)$ as $\sum_{n=0}^\infty \va_n x^n$, we directly obtain that the sequence $(\va_n)_{n=0}^\infty$ satisfies the recurrence relation for $(\vb_n)_{n=0}^\infty$ stated in \Cref{eq:th3:recurrence}. Moreover, the two sequences have the initial value, that is $\va_0 = \vb_0$. Thus $(\vb_n)_{n=0}^\infty = (\va_n)_{n=0}^\infty$, which concludes the proof of  the direction \ref{enum:thm:MAIN2-recurrence} $\Rightarrow$ \ref{enum:thm:MAIN2-ordinary} in \Cref{thm:MAIN2}. 


\subsection{From Multi-Holonomic Sequences to Differentially Algebraic Power Series}

This section focuses on proving~\ref{enum:thm:MAIN2-recurrence} $\Rightarrow$ \ref{enum:thm:MAIN2-D-algebraic} in \Cref{thm:MAIN2}. In other words, we show that  given a multi-holonomic sequence~$(\va_n)_{n=0}^\infty$, 
all induced 
coordinate power series~$\sum_{n=0}^\infty a_{n,i} x^n$, fpr $i \in \{1, \ldots, d\}$, are differentially algebraic. For this, we rely once again on the differential operator $\SDerive{f}$  defined by $\SDerive{f}(x) := x f'(x)$.
For all $i \in \N$ we write $\SDeriveN{f}{i+1}(x)$ for the $i+1$-th  iteration of this operator:
\[\SDeriveN{f}{i+1}(x) := \SDerive{(\SDeriveN{f}{i})}(x)\, .\]

Let $(\va_n)_{n=0}^\infty$ be a multi-holonomic sequence of arity~$m$ and degree~$r$,   defined by  matrices $\mu_1,\ldots,\mu_m$ that, without loss of generality, are assumed to share a common denominator $A(x) \in \K[x]$. Furthermore, as $(\va_n)_{n=0}^\infty$ has degree $r$, we may assume that
\[
A(x) = \sum_{i=0}^r c_i x^i
\]
for some $c_0,\ldots,c_r \in r$ and for all $\ell \in \{1, \ldots, m\}$, \[\mu_\ell(x) = \sum_{i=0}^r x^i \cdot \mu_{\ell,i}\]
for some $\mu_{\ell,i} \in \K^{d^\ell\times d}$. Recall that $J_n(\ell)$ is the set of tuples $(j_1,\ldots,j_\ell) \in \N^{\ell}$ such that $j_1 + \ldots + j_\ell = n-1$.

\begin{restatable}{proposition}{propsystemofequations}
\label{prop:system equations}
Let $(\va_n)_{n=0}^\infty$ be a multi-holonomic sequence, defined as above. The vector of power series $\vf(x) = \sum_{n=0}^\infty \va_n \,x^n$ satisfies the following system of differential equations:
\begin{align}
\sum_{i=0}^r c_i &\SDeriveN{\vf}{i} - c_0\va_0 \notag\\
=\ & x \sum_{\substack{\ell \in \{1, \ldots,m\}\\i \in \{0, \ldots,r\}\\\bm{j} \in J_{i+1}(\ell+1)}} \binom{i}{\bm{j}}  (\SDeriveN{\vf}{j_1} \otimes \ldots \otimes \SDeriveN{\vf}{j_\ell}) \cdot \mu_{\ell,i}\label{eq:system1}
\end{align}
\end{restatable}

\begin{proof}[Proof Sketch.]
The argument relies on the observation that for any power series $g(x) = \sum_{n=0}^\infty b_n x^n$, we have 
\[
\SDerive{g}(x) = \sum_{n=0}^\infty n b_n x^n\, \]
and more generally $\SDeriveN{g}{k}(x) = \sum_{n=0}^\infty n^k b_n x^n$. Thus, 
\[
\sum_{i=0}^r c_i \SDeriveN{f}{i}(x) = \sum_{n=0}^\infty A(n) \va_n x^n\, 
\]
Moreover, for $\bm{j} \in J_{i+1}(\ell+1)$,
\begin{align*}
(\SDeriveN{\vf}{j_1} \otimes &\ldots \otimes \SDeriveN{\vf}{j_\ell}) \cdot \mu_{k,i} \\
=\ & \sum_{\bm{n} \in \N^\ell} (\va_{n_1} \otimes \ldots \otimes \va_{n_\ell}) n_1^{j_1} \ldots n_\ell^{j_\ell} x^{n_1+\ldots+n_\ell}\cdot \mu_{k,i}
\end{align*}
By using the multinomial theorem, instantiated as follows,
\[
(n_1 + \ldots + n_\ell + 1)^i = \sum_{\substack{\bm{j} \in J_{i+1}(\ell+1)}} \binom{i}{\bm{j}} n_1^{j_1}\ldots n_\ell^{j_\ell}
\]
and by unfolding \Cref{eq:multiholonomic}, we conclude the proof.
\end{proof}

In \Cref{prop:system equations}, we showed that given a multi-holonomic sequence $(\va_n)_{n=0}^\infty$, the vector of power series $\vf(x) = \sum_{n=0}^\infty \va_n \,x^n$ is a solution of the system as in \Cref{eq:system1}. In the next proposition, we show that this vector is the unique solution of the system,  given the initial value $\va_0$.


\begin{restatable}{proposition}{propunicity}\label{pro: Unicity}
Let $(\va_n)_{n=0}^\infty$ be a multi-holonomic sequence and  $S$ be the system of differential equations derived from $(\va_n)_{n=0}^\infty$ as given in \Cref{eq:system1} of \Cref{prop:system equations}.

Given any solution $\overline{\vf}(x) = \sum_{n=0}^\infty \overline{\va}_n\,x^n$ of $S$, it is the unique power series solution with the initial value~$\overline{\va}_0$. 
Moreover, if $A(0) \neq 0$ then $\vf(x) = \sum_{n=0}^\infty \va_n \,x^n$ is the unique power series solution of~$S$.
\end{restatable}

\begin{proof}[Proof sketch]
The proof involves a fairly straightforward analysis of the $n$-th coefficients (for $n \geq 1$) of the two sides of \Cref{eq:system1} evaluated at the solution $\overline{\vf}(x)$. 
In particular, the $n$-th coefficient of the left-hand side can be written as 
\[
\sum_{i=0}^r c_i\, n^i\, \overline{\va}_n - c_0 \va_0 = A(n) \overline{\va}_n - c_0 \va_0
\]
As the right-hand side of~\Cref{eq:system1} is guarded by $x$, so for all  $n \geq 1$, the $n$-th coefficient of the right-hand side evaluated at $\overline{\vf}(x)$ is an expression in 
values $\overline{\va}_{n'}$ with $n' < n$.
Hence,  the term $\overline{\va}_n$ is uniquely defined by the values of its predecessors, which in turn implies that $\overline{\vf}(x)$ is uniquely defined by $\overline{\va}_0$.

Notice that if $c_0 \neq 0$ then the constant coefficient of the left-hand side is $c_0 \overline{\va}_{0} - c_0 \va_0$ whereas the constant coefficient of the right side is always $0$. Hence when $c_0 \neq 0$, the equality $\overline{\va}_{0} = \va_0$ holds, which in turn implies that $\vf(x) = \sum_{n=0}^\infty \va_n \,x^n$ is the unique power series solution of~$S$.
\end{proof}

We  use the uniqueness of the power series solution, given an initial value,  to prove that all induced coordinate power series of the solution are differentially algebraic. This is based on a differential version of the Artin approximation theorem given in \cite{denef1984power}.

\begin{theorem}\cite[Theorem 2.1]{denef1984power}\label{thm: approximation}
    Let $S$ be a system of differential polynomials in the differential variables $y_1, \ldots, y_k$ with coefficients in $\K[x]$. Let $(y_1(x), \ldots, y_k(x)) \in \K[\![x]\!]^k$ be a solution of $S$ and $m \in \N$.
    Then there exists another solution $(\Bar{y}_1(x), \ldots, \Bar{y}_k(x)) \in \K[\![x]\!]^k$ such that all $\Bar{y}_i(x)$'s are  differentially algebraic and  \[(y_1(x), \ldots, y_k(x)) \equiv (\Bar{y}_1(x), \ldots, \Bar{y}_k(x)) \pmod{x^m}\,.\]
\end{theorem}

This theorem states that any solution (not necessarily differentially algebraic) of a system of a differential equations can be approximated by a differentially algebraic solution, such that both solutions agree on the first $m$ coefficients. 
This differential version of the Artin approximation theorem is used in the following proposition to complete  
 the proof of the implication \ref{enum:thm:MAIN2-recurrence} $\Rightarrow$ \ref{enum:thm:MAIN2-D-algebraic} in \Cref{thm:MAIN2}. 


\begin{proposition}
\label{th:differentially algebraic}
Let $(\va_n)_{n=0}^\infty$ be a multi-holonomic sequence. All components of  $\vf(x) = \sum_{n=0}^\infty \va_n\,x^n$ are differentially algebraic power series.
\end{proposition}

\begin{proof}
From \Cref{prop:system equations}, $\vf(x) = \sum_{n=0}^\infty \va_n\,x^n$ is a solution of the system of differential equations given in \Cref{eq:system1}. From \Cref{thm: approximation}, we know that there exists another power series solution $\overline{\vf}(x)$  such that all the power series are differentially algebraic and they coincide in particular with the first coefficients of $\vf(x)$, i.e., $\overline{\va}_0 = \va_0$. From \Cref{pro: Unicity}, we deduce that $\overline{\vf}(x) = \vf(x)$ which allows us to conclude that all power series in $\vf(x)$ are differentially algebraic.
\end{proof}


\subsection{From Differentially Algebraic Power Series to Multi-Holonomic Sequences}

To complete the proof of \Cref{thm:MAIN2}, it remains to show that any differentially algebraic power series is the component of a multi-holonomic sequence of arity 2, that is, the implication \ref{enum:thm:MAIN2-D-algebraic} $\Rightarrow$ \ref{enum:thm:MAIN2-recurrence} in \Cref{thm:MAIN2}. Our proof builds on a construction of Denef and Lipschitz~\cite{denef1984power}, who showed that 
the sequence of coefficients of any differentially algebraic power series $f(x) = \sum_{n=0}^\infty a_n x^n$ satisfies a recurrence of the form:
\[P(n) \,  a_{n} = Q_n(a_0,\ldots,a_{n-1}) \, , \]
where $Q_n$ is a polynomial whose form depends on $n$. Our proof makes explicit this recurrence and shows that how it can be transformed into a generalised multi-holonomic sequence. The proof will then conclude by applying \Cref{thm:MAIN3}.

\medskip

Observe that $n$-th coefficient of  a power series $f(x) = \sum_{n=0}^\infty a_n x^n$ satisfies 
 $a_n = n!\,f^{(n)}(0)$.   
  Hence, to obtain a recurrence relation between the coefficients, we only need to obtain a recurrence relation between the derivatives of~$f$ evaluated at $0$. 
  One way to do so is to differentiate the differential polynomial $P$ on which the differentially algebraic power series $f(x)$, that is, the polynomial
  $P(y,y^{(1)},\ldots, y^{(d)})$ with coefficients in $\K[x]$ such that 
  $P(f(x),f'(x), \ldots, f^{(d)}(x)) = 0$.
In particular, notice that
\begin{equation}
\label{eq:denef}
P' = \frac{\partial P}{\partial y^{(d)}} y^{(d+1)} + R ',  ,
\end{equation}
with both $R$ and $\frac{\partial P}{\partial y^{(d)}}$ being polynomials in $\K[x][y,\ldots,y^{(d)}]$, thus of order strictly smaller than $d+1$. The partial derivative of $P$ by its highest-order variable, namely $\frac{\partial P}{\partial y^{(d)}}$, is called the \emph{separant of $P$}. Since $f(x)$ and its derivatives form a root of $P$ then they are also a root of $P'$. Hence, instantiating \Cref{eq:denef} with $f$ and its derivates evaluated at $0$  leads to the following equation.
\[
\frac{\partial P}{\partial y^{(d)}}(f(0),\ldots,f^{(d)}(0)) f^{(d+1)}(0) = - R(f(0),\ldots,f^{(d)}(0))
\]
We can see here how the recurrence takes shape as $f^{(d+1)}(0)$ is expressed using an equation that depends solely on $f^{(i)}(0)$ for $i \leq d$. However, the main hurdle comes from the fact that the separant can vanish at $0$.

As a first step in addressing this issue, 
in the next proposition, we show that 
one can always assume that 
 the separant of $P$ is not identically zero at the solution. In other words, 
\[
\frac{\partial P(y,y^{(1)},\ldots,y^{(d)})}{y^{(d)}}(f(x),f'(x), \ldots,f^{(d)}(x)) \neq 0 \,.
\]
This properties holds if $P$ has minimal order and degree among differential polynomials that annihilate $f$.


\begin{restatable}{proposition}{propvanishingorder}
\label{prop:vanishing order}
Let $f(x) \in \K[\![x]\!]$ be a differentially algebraic power series. There exists a differential polynomial $P \in \K[x][y,\ldots,y^{(d)}]$ such that
\begin{align*}
   P(f(x),\dots,f^{(d)}(x))=0, \text{ and }\\
   \frac{\partial P}{\partial y^{(d)}} (f(x),\ldots,f^{(d)}(x)) \neq 0 \,.
\end{align*}
\end{restatable}

\begin{proof}
By definition of differentially algebraic power series, there exists a non-zero polynomial $P(y^{(0)},\ldots,y^{(d)})$ with coefficients in $\K[x]$ such that $P(f(x),\ldots,f^{(d)}(x)) = 0$. We consider such polynomial $P$ with the smallest $(d,\mathrm{deg}(y^{(d)}))$ (using the lexicographic ordering) where $d$ is the highest order of $y$ in $P$ and  $\deg(y^{(d)})$ is the degree of $y^{(d)}$ in $P$. 
As the polynomial $\frac{\partial P}{\partial y^{(d)}}$ is not identically zero and has either the same order and strictly smaller degree than $y^{(d)}$, or a strictly small order,  by minimality of $P$ in the defined lexicographic ordering, we deduce  that $\frac{\partial P}{\partial y^{(d)}}(f(x),\ldots,f^{(d)}(x)) \neq 0$. 
\end{proof}

We may thus assume that  the separant of $P$ not identically zero at the solution $f$.
The core idea of the Denef and Lipschitz~\cite{denef1984power} construction consists of iteratively differentiating the polynomial $P$
until the resulting differential polynomial meets certain condition (requiring that the coefficient of highest order term 
never vanishes at~$0$).
The number of successive differentiations  depends on the separant.

 Given $k \in \N$ and differential polynomials $A \in \K[x][y^{(\infty)}]$,  denote by $A \mod y^{(k)}$ the polynomials obtained from $A$ by removing all monomials of order at least~$k$ (all monomials that are multiple of $y^{(o)}$ for some $o\geq k$).


\begin{restatable}{proposition}{lemRecursion}
\label{lem:RECUR1}
Let $f \in \K[\![x]\!]$ be a power series and $P \in \K[x][y,\ldots,y^{(d)}]$ a differential polynomial  such that
\[P(x,f,\dots,f^{(d)})=0 \quad \text{ and } \quad \frac{\partial P}{\partial y^{(d)}} (x,f,\ldots,f^{(d)}) \neq 0\,.\]
Then there exist $m,s,N\in \mathbb N$
and polynomial $A \in \K[x]$ such that for all~$n \geq N$ we have $A(n) \neq 0$ and
\begin{gather}
f^{(n+s+1)}(0)  = \frac{1}{A(n)} \, Q_n(0,f(0),\ldots,f^{(n+s)}(0))
\label{eq:KEY}
\end{gather}
where $Q_n = P^{(m+n)} \bmod y^{(n+s+1)}$.
\end{restatable}

\begin{proof}[Proof sketch]
Since the separant is not identically zero, after evaluating it at $f$ and its derivates, we get
\[
\frac{\partial P}{\partial y^{(d)}}\left(f,\ldots,f^{(d)}\right)=x^k \cdot \sum_{i=0}^\infty c_i x^i
\] 
for some $k$ with $c_0 \neq 0$.
In the sequel we refer to $k$ as the \emph{vanishing order} of the separant. 
 
By differentiating $P$ to order $2k+2$, we can show that
\[
    P^{(2k+2)} =  y^{(d+2k+2)}f_d + \cdots + y^{(d+k+2)}f_{d+k}+f_{d+k+1}
\]
where each $f_j$ has order at most $j$ and $f_d = \frac{\partial P}{\partial y^{(d)}}$. Continuing to differentiate, we obtain that for all $n \in \N$, $P^{(2k+2+n)}$ becomes 
\begin{equation}
\label{eq:proof-denef2}
y^{(d+2k+2+n)}S_0(n) + \cdots + y^{(d+k+2+n)}S_k(n)+ H
\end{equation}
where each $S_j(n)$ has order at most $d+j$ and $H$ has order at most $d+k+1+n$. Moreover, each \[S_j(n)(0,f(0),\ldots, f^{(d+j)}(0))\]
can be seen as a univariate polynomial in $n$.
In particular, as we  assumed that the separant has vanishing order $k$, there is at least one of these polynomials that are not identically zero. We take the smallest index $r \in \{0, \ldots, k\}$ such that \[- S_r(n)(0,f(0),\ldots, f^{(d+r)}(0))\] is not identically zero and we define $A(n)$ to be this polynomial.

By taking $m = 2k+2$ and $s = d + 2k + 1 - r$, the restriction of \Cref{eq:proof-denef2} to monomials of order strictly lower than $s+n+1$ corresponds to the polynomials \[y^{(d+k+2+n-r-1)}S_{r+1}(n) + \ldots + y^{(d+k+2+n)}S_k(n)+ H\,.\] Therefore, defining,  for all~$n \geq N$, 
\begin{align}
\label{def-Qn}
    Q_n := P^{(m+n)} \mod y^{(s+n+1)}\,,
\end{align} we obtain the desired result.
\end{proof}

To complete the proof of the implication \ref{enum:thm:MAIN2-D-algebraic} $\Rightarrow$ \ref{enum:thm:MAIN2-recurrence} in \Cref{thm:MAIN2}, we show that the polynomial $Q_n$ defined in~\eqref{def-Qn} naturally induces a generalised multi-holonomic sequence. Indeed, using Pascal's identity, we have for every monomial $M = y^{(i_1)}\ldots y^{(i_\ell)}$ in $P$ and for all~$n \in \N$, 
\[ M^{(n)} = \sum_{\bm{j} \in J_{n+1}(\ell)} \binom{n}{\bm{j}} \, y^{(i_1+j_1)} \cdots y^{(i_\ell+j_\ell)} \,.  \]
Computing  $M^{(m+n)}$ modulo $y^{(s+n+1)}$ intuitively corresponds to adding guard conditions on the indices $j_1,\ldots,j_\ell$. As such, $M^{(m+n)} \mod y^{(s+n+1)}$ becomes 
\[
    \sum_{\bm{j} \in G_{n+1}(\bm{m})} \binom{n+m}{\bm{j}} \, y^{(i_1+j_1)} \cdots y^{(i_\ell+j_\ell)}
\]
with $\bm{m} = (m,i_1-1-s,\ldots,i_\ell-1-s)$. The final generalised multi-holonomic sequence is obtained by summing over all monomials and applying an additional shift by $N$ since \Cref{eq:KEY} holds only when $n \geq N$. The proof concludes by invoking \Cref{thm:MAIN3}.  

\subsection{Equivalence of differential tree automata and their generating function}

We can use
\Cref{thm:MAIN2}  to decide the zeroness and equivalence of the 
generating functions and formal tree series associated to a differential tree automaton.


\begin{theorem}
\label{th:equivalence}
The problems of zeroness and equivalence of formal tree series for $\Q$-differential tree automata are decidable. 
The  problems of zeroness and equivalence of ordinary and labelled generating functions for $\Q$-differential tree automata are decidable.
\end{theorem}

\begin{proof}
Let $\A$ be a differential tree automaton. From \Cref{prop:system equations,pro: Unicity}, we know that the vector of $d$ power series $\vf_\A(x) = \sum_{n=0}^\infty \va_n x^n$ is the unique solution of a system of differential equations $S$ with $\va_0$ as initial value. Recall that the ordinary generating function of $\A$ is the first coordinate of $\vf_\A(x)$. We augment the system $S$ with
\begin{itemize}
\item an equation $f_1 = 0$ to test if the first coordinate is $0$;
\item $d$ equations $f_1 = a_{0,1} + x g_1$, \ldots, $f_d = a_{0,d} + x g_d$ with $g_1,\ldots,g_d$ fresh variables to enforce the initial values.
\end{itemize} 
Thus, deciding zeroness of the ordinary generating function 
is equivalent to determining solubility of a system of differential equations of $\Q[x]$, which is decidable~\cite[Theorem 3.1]{denef1984power}.

All other problems mentioned in this theorem can be reduced to the zeroness of ordinary generating functions: For zeroness of labelled generating functions, we rely on \Cref{prop:automata to recurrence}. The equivalence of generating functions can be reduced to the the zeroness of their difference thanks to their closure properties (see \Cref{prop:closure} in Appendix~\ref{sec:operations}). 
The same closure properties allow us to decide the zeroness of formal tree series $\sem{\A}$ by building an automaton $\A'$ such that $\sem{\A'}{t} = \sem{\A}{t}^2$ for all $t\in T_\Sigma$ and by testing the zeroness of the ordinary generating function of $\A'$.
\end{proof}


\section{Application to counting combinatorial structures}
\label{sec:SPECIES}
Combinatorial species are a formalism for defining classes 
of combinatorial structures~\cite{FlajoletZC94}.
 A species $A$ is a mapping (technically, a functor) that takes a finite set of labels $U$ and returns a set, denoted $A[U]$, of combinatorial structures labelled by $U$. For example, if $B$ is the species of plane rooted binary trees then $B[\{1, 4, 5\}]$ is the set of trees shown below:

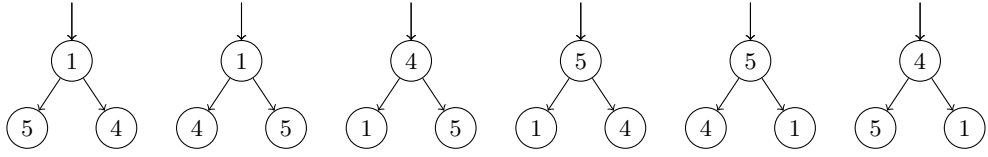
\begin{center}
\begin{tikzpicture}[
    nodet/.style = {draw,circle,font=\footnotesize,inner sep=0.1cm},
]
    \node[nodet,initial,initial text={},initial above,initial distance=0.25cm] (q_0) {$1$};
    \node[nodet] (q_2) [below right = 0.25cm and 0.2cm of q_0] {$4$};
    \node[nodet] (q_3) [below left = 0.25cm and 0.2cm of q_0] {$5$};

    \path[->] 
        (q_0) edge (q_2)
        (q_0) edge (q_3)
    ;

    \node[nodet,initial,initial text={},initial above,initial distance=0.25cm] (q_1) [right = 1.7cm of q_0]{$1$};
    \node[nodet] (q_2) [below right = 0.25 and 0.2cm of q_1] {$5$};
    \node[nodet] (q_3) [below left = 0.25 and 0.2cm of q_1] {$4$};

    \path[->] 
        (q_1) edge (q_2)
        (q_1) edge (q_3)
    ;

    \node[nodet,initial,initial text={},initial above,initial distance=0.25cm] (q_1') [right = 1.7cm of q_1]{$4$};
    \node[nodet] (q_2) [below right = 0.25 and 0.2cm of q_1'] {$5$};
    \node[nodet] (q_3) [below left = 0.25 and 0.2cm of q_1'] {$1$};

    \path[->] 
        (q_1') edge (q_2)
        (q_1') edge (q_3)
    ;

    \node[nodet,initial,initial text={},initial above,initial distance=0.25cm] (q_1'') [right = 1.7cm of q_1']{$5$};
    \node[nodet] (q_2) [below right = 0.25 and 0.2cm of q_1''] {$4$};
    \node[nodet] (q_3) [below left = 0.25 and 0.2cm of q_1''] {$1$};

    \path[->] 
        (q_1'') edge (q_2)
        (q_1'') edge (q_3)
    ;

    \node[nodet,initial,initial text={},initial above,initial distance=0.25cm] (q_1) [below = 1cm of q_1]{$5$};
    \node[nodet] (q_2) [below right = 0.25 and 0.2cm of q_1] {$1$};
    \node[nodet] (q_3) [below left = 0.25 and 0.2cm of q_1] {$4$};

    \path[->] 
        (q_1) edge (q_2)
        (q_1) edge (q_3)
    ;

    \node[nodet,initial,initial text={},initial above,initial distance=0.25cm] (q_1) [right = 1.7cm of q_1]{$4$};
    \node[nodet] (q_2) [below right = 0.25 and 0.2cm of q_1] {$1$};
    \node[nodet] (q_3) [below left = 0.25 and 0.2cm of q_1] {$5$};

    \path[->] 
        (q_1) edge (q_2)
        (q_1) edge (q_3)
    ;

\end{tikzpicture}
\end{center}
Note that all nodes do not have to be labelled. If $C$ is the species of plane rooted binary trees with labels only on external nodes then $C[\{1, 2\}]$ contains only two trees:
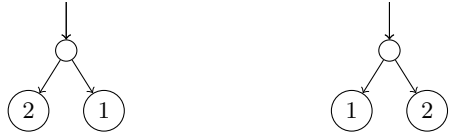
\begin{center}
    \begin{tikzpicture}[
        nodet/.style = {draw,circle,font=\footnotesize,inner sep=0.1cm},
    ]
        \node[nodet,initial,initial text={},initial above,initial distance=0.25cm] (q_0) {};
        \node[nodet] (q_2) [below right = 0.25 and 0.2cm of q_0] {$1$};
        \node[nodet] (q_3) [below left = 0.25 and 0.2cm of q_0] {$2$};
    
        \path[->] 
            (q_0) edge (q_2)
            (q_0) edge (q_3)
        ;
    
        \node[nodet,initial,initial text={},initial above,initial distance=0.25cm] (q_1) [right = 4cm of q_0]{};
        \node[nodet] (q_2) [below right = 0.25 and 0.2cm of q_1] {$2$};
        \node[nodet] (q_3) [below left = 0.25 and 0.2cm of q_1] {$1$};
    
        \path[->] 
            (q_1) edge (q_2)
            (q_1) edge (q_3)
        ;
    \end{tikzpicture}
\end{center}

To define species \cite{FlajoletZC94} relies on a collection of constructions:
\begin{inparaenum}[(i)]
\item The initial object $\vone$, which represents the empty structure without label.
\item The variable $X$, which represents a single labelled node.
\item The addition $A + B$, which represents the disjoints union of structures from $A$ and $B$. 
\item The product $A \cdot B$, which represents all the pairs of elements from $A$ and $B$ over the input set of labels.
\item The  $\sequence(A)$, which generates the sequences of elements of $A$.
\item The  $\set(A)$, which generates the sets of elements of $A$.
\item The  $\cycle(A)$, which generates the cycles of elements of $A$.
\end{inparaenum}


\begin{figure}[ht]
\small
\begin{center}
\begin{TAB}[2pt]{|c|c|c|c|}{|c|c|c|c|c|c|c|c|c|c|c|c|}
System & RDA Systems & \begin{tabular}{c}Initial\\values\end{tabular} & Spec\\
$\left\{
\begin{array}[1]{l}
y_a = x\, z\\
z' = z\, y_a'
\end{array}
\right.$ &
$\left\{
\begin{array}[1]{l}
y_a' = z + \frac{x\,z^2}{1-x\,z}\\
z' = \frac{z^2}{1-x\,z}
\end{array}
\right.$
&
---
&
A\\
$y_b = x + y_b^2$ &
$y_b' = \frac{1}{1-2\,y_b}$
&
$y_b(0)=0$
&
B
\\
$y_c = \frac{x}{1-y_c}$ &
$y_c' = \frac{(1-y_c)^3}{(1-y_c)^2-x}$
&
$y_c(0) = 0$
&
C\\
$\left\{
\begin{array}{l}
y'_d = y_d\, z'\\
z' = \frac{1}{1-x}
\end{array}
\right.$ &
$y_d' = \frac{y_d}{1-x}$
&
---
&
D\\
$\left\{\begin{array}{l}
y'_e = y_e\, z_2'\\
z_2' = \frac{y'_a}{1-y_a}\\
y_a = x\, z\\
z' = z\, y_a'
\end{array}\right.$ &
$\left\{
\begin{array}{l}
y'_e = y_e\, \frac{z}{1-y_a} +\\
\quad \frac{x\, z^2}{(1-x\,z)(1-y_a)}\\[2mm]
y_A' = z + \frac{x\,z^2}{1-x\,z}\\[1mm]
z' = \frac{z^2}{1-x\,z}
\end{array}
\right.$
&
$y_a(0) = 0$
&
E\\
$\left\{
\begin{array}{l}
y'_f = y_f\, z\\
z' = z
\end{array}
\right.$ &
$\left\{
\begin{array}{l}
y'_f = y_f\, z\\
z' = z
\end{array}
\right.$
&
---
&
F
\\
$\left\{
\begin{array}{l}
y_g = x + x\, z_2\\
z_2' = z_1 \, y'_g\\
z_1' = y_g\, y'_g 
\end{array}
\right.$ &
$\left\{
\begin{array}{l}
y'_g = \frac{1+z_2}{1-x\,z_1}\\[1mm]
z'_2 = \frac{y_1(1+z_2)}{1-x\,z_1}\\[1mm]
z_1' = y_g\frac{1+z_2}{1-x\,z_1}\\
\end{array}
\right.$
&
---
&
G
    \\
$\left\{
\begin{array}{l}
y_h = x + z_2\\
z_2' = z_1 \, y'_h\\
z_1' = z_0 \, y'_h\\
z_0' = z_0 \, y'_h
\end{array}
\right.$ &
$\left\{
\begin{array}{l}
y'_h = 1 + \frac{z_1}{1-z_1}\\[1mm]
z_2' = \frac{z_1}{1-z_1}\\[1mm]
z'_1 = z_0 + \frac{z_0\,z_1}{1-z_1}\\[1mm]
z_0' = z_0 + \frac{z_0\,z_1}{1-z_1}\\
\end{array}
\right.$
&
$
\begin{array}{@{}l@{}}
y_h(0) = 0\\
z_1(0) = 0\\
\end{array}
$
&
H\\
$\left\{
\begin{array}{l}
y'_k = y_k\, z'_1\\
z_1' = \frac{z'_2}{1-z_2}\\
z_2 = x \, z_3\\
z_3' = y_g \, y'_g\\
y_g = x + x\, z_4\\
z_4' = z_5 \, y'_g\\
z_5' = y_g \, y'_g 
\end{array}
\right.$ &
$\left\{
\begin{array}{l}
y'_k = y_k\frac{z_3}{1-x\,z_3} +\\[1mm]
\quad y_k\frac{x\,y_g(1+z_4)}{(1-x\,z_5)(1-x\,z_3)}\\[2mm]
z'_1 = \frac{z_3}{1-x\,z_3} + \\[1mm]
\quad \frac{x\,y_g(1+z_4)}{(1-x\,z_5)(1-x\,z_3)}\\[2mm]
z'_2 = z_3 + x\,\frac{y_g(1+z_4)}{1-x\,z_5}\\[1mm]
z'_3 = \frac{y_g(1+z_4)}{1-x\,z_5}\\[1mm]
y'_g = \frac{1+z_4}{1-x\,z_5}\\[1mm]
z'_4 = \frac{z_5(1+z_4)}{1-x\,z_5}\\[1mm]
z_5' = y_g\frac{1+z_4}{1-x\,z_5}\\
\end{array}
\right.$
&
---
&
K
\\
$\left\{
    \begin{array}{l}
    y'_\ell = y_\ell\, z_1 \,z'_2\\
    z'_1 = z_1 \, z'_2\\
    z_2' = z_3\\
    z_3' = y_3\\
    \end{array}
    \right.$ &
$\left\{
    \begin{array}{l}
    y'_\ell = y_\ell\, z_1\, z_3\\
    z'_1 = z_1 \,z_3\\
    z_3' = z_3\\
    \end{array}
    \right.$
&
---
&
L
\\
$\left\{
    \begin{array}{l}
    y_m = \frac{1}{1-z_1}\\
    z_1' = z_0\\
    z_0' = z_0\\
    \end{array}
    \right.$ &
$\left\{
    \begin{array}{l}
    y'_m = \frac{z_0}{(1-z_1)^2}\\
    z_1' = z_0\\
    z_0' = z_0\\
    \end{array}
    \right.$
&
$y_m(0) = 1$
&
M
\end{TAB}
\end{center}
\caption{Example of Species with RDA exponential generating series. Unspecified initial values indicate that all power series solutions are Rationally Dynamically Algebraic, no matter the initial values.}
\label{fig:RDS examples}
\vspace{-.5cm}
\end{figure}

The constructions $\set$, $\sequence$ and $\cycle$ also allow constraints on the cardinality. For example, $\set(A,card \geq 3)$ represents the sets of at least 3 elements of $A$.
A specification of species is then a (set of) equations that use these constructions. In this setting, \cite{FlajoletZC94} is interested in computing the exponential generating series of a species $A$, corresponding to the  series 
\[\sum_{n=0}^\infty \frac{\size{A[\{1,\ldots,n\}]}}{n!} x^n\,.\] 
In other words, the $n$-th coefficient of the power corresponds to the numbers of combinatorial structures in the species $A$ labelled by~$\{1, \ldots, n\}$.
For instance, \cite{FlajoletZC94} provides the following examples of specifications:
\begin{center}
\small
\begin{tabular}{@{}l@{\,}|@{\,}l}
Specifications & Objects\\
\hline
$A = X \cdot \set(A)$ & Non-plane trees\\
$B = X + B \cdot B$ & Plane binary trees with\\
& only external labels\\
$C = X \cdot \sequence(C)$ & Plane general trees \\
$D = \set(\cycle(X))$ & Permutations\\
$E = \set(\cycle(A))$ & Functional graph\\
$F = \set(\set(X,card \geq 1))$ & Set partitions\\
$G = X + X\cdot \set(G,card = 3)$ & Non-plane ternary trees\\
$H = X + \set(H, card \geq 2)$ & Hierarchies\\
$K = \set(\cycle(X \cdot \set(G,card = 2)))$ & 3-constrained functional\\
& graphs\\
$L = \set(\set(\set(X,card\geq 1),card \geq 1))$ & 3-balanced hierarchies\\
$M = \sequence(\set(X,card \geq 1))$ & Surjections
\end{tabular}
\end{center}

These specifications represents standard objects in the literature. Here, non-plane trees are trees in which the children of a node
are unordered, whereas in plane trees siblings are ordered. Functional graphs are directed graphs with every node having outdegree 1 and 3-functional graphs are function graphs where each node has indegree being 0 or 3. 

Fortunately, \cite{FlajoletZC94} also provides a way to translate every construct of the specification into a system of differential equations on the associated exponential generating functions. For instance, (i)~$B = \set(A)$ is translated to the equation $y_B' = y_B \cdot y_A'$; (ii) Z$B = \sequence(A)$ is translated into $y_B = \frac{1}{1-y_A}$; (iii)~$B = \cycle(A)$ is translated into $y_B' = \frac{y_A'}{1-y_A}$. We refer the reader to \cite[Thereom 2]{FlajoletZC94} for a more detailed description of this translation. 

We show in \Cref{fig:RDS examples} that each of the specifications presented above translates into a system of differential equations and in addition we show that their exponential generating series are all RDA power series.
Once the system of differential equations is transformed into a system as in \Cref{eq:RDA}, checking that the exponential generating series is RDA only requires checking that the rational functions are defined on the initial values. In some cases, all initial values yield RDA power series: for example for the systems of non-plane trees, permutations, 3-balanced hierarchies.  In the other cases, only specific initial values yield  RDA power series (e.g. Functional graphs).  We present in \Cref{fig:RDS examples} the admissible initial value of the exponential generating series of the species. 

We highlight the cases of the Hierarchies and Surjections.
For Hierarchies, the initial value $y_h(0) = 0$ is given by the specification. However, understanding the initial value of $z_1(x)$ is less evident. For this, we rely on the expression satisfied by $y_h(x)$ given in~\cite{FlajoletZC94}: $y_h(x) = x + e^{y_h(x)} - 1 - y_h(x)$. From this we have 
\[y'_h(x) = \frac{1}{2-e^{y_h(x)}} \qquad \text{ and } \qquad z'_2(x) = y'_h(x) - 1.\] Since $y_h(0) = 0$, we deduce that $y'_h(0) = 1$ and $z'_2(0) = 0$. Hence, $z_2'(x) = z_1(x) y'_h(x)$ yields $y_1(0) = 0$.
For Surjections, the initial value $y_m(0) = 1$ is given by the specification. Therefore, as $(1 - z_1(x)) y_m(x) = 1$, we deduce that $(1 - z_1(0)) \neq 0$.

\onecolumn 

\bibliographystyle{IEEEtran}
\bibliography{IEEEabrv,references}

\begin{thebibliography}{10}
\providecommand{\url}[1]{#1}
\csname url@samestyle\endcsname
\providecommand{\newblock}{\relax}
\providecommand{\bibinfo}[2]{#2}
\providecommand{\BIBentrySTDinterwordspacing}{\spaceskip=0pt\relax}
\providecommand{\BIBentryALTinterwordstretchfactor}{4}
\providecommand{\BIBentryALTinterwordspacing}{\spaceskip=\fontdimen2\font plus
\BIBentryALTinterwordstretchfactor\fontdimen3\font minus \fontdimen4\font\relax}
\providecommand{\BIBforeignlanguage}[2]{{%
\expandafter\ifx\csname l@#1\endcsname\relax
\typeout{** WARNING: IEEEtran.bst: No hyphenation pattern has been}%
\typeout{** loaded for the language `#1'. Using the pattern for}%
\typeout{** the default language instead.}%
\else
\language=\csname l@#1\endcsname
\fi
#2}}
\providecommand{\BIBdecl}{\relax}
\BIBdecl

\bibitem{bergeron1998combinatorial}
F.~Bergeron, G.~Labelle, and P.~Leroux, \emph{Combinatorial species and tree-like structures}.\hskip 1em plus 0.5em minus 0.4em\relax Cambridge University Press, 1998, no.~67.

\bibitem{flajolet2009analytic}
P.~Flajolet and R.~Sedgewick, \emph{Analytic combinatorics}.\hskip 1em plus 0.5em minus 0.4em\relax cambridge University press, 2009.

\bibitem{FORSMAN1992341}
K.~Forsman, ``On rational state space realizations,'' \emph{IFAC Proceedings Volumes}, vol.~25, no.~13, pp. 341--346, 1992, 2nd IFAC Symposium on Nonlinear Control Systems Design 1992, Bordeaux, France, 24-26 June.

\bibitem{ovchinnikov2022bounds}
A.~Ovchinnikov, G.~Pogudin, and T.~N. Vo, ``Bounds for elimination of unknowns in systems of differential-algebraic equations,'' \emph{International Mathematics Research Notices}, pp. 12\,342--12\,377, 2022.

\bibitem{fülöp2024weightedtreeautomata}
\BIBentryALTinterwordspacing
Z.~Fülöp and H.~Vogler, ``Weighted tree automata -- may it be a little more?'' 2024. [Online]. Available: \url{https://arxiv.org/abs/2212.05529}
\BIBentrySTDinterwordspacing

\bibitem{denef1984power}
J.~Denef and L.~Lipshitz, ``Power series solutions of algebraic differential equations,'' \emph{Mathematische annalen}, pp. 213--238, 1984.

\bibitem{Bergeron1990CombinatorialRO}
\BIBentryALTinterwordspacing
F.~Bergeron and C.~Reutenauer, ``Combinatorial resolution of systems of differential equations iii: a special class of differentially algebraic series,'' \emph{Eur. J. Comb.}, vol.~11, pp. 501--512, 1990. [Online]. Available: \url{https://api.semanticscholar.org/CorpusID:42547069}
\BIBentrySTDinterwordspacing

\bibitem{kauers2023d}
M.~Kauers, \emph{D-finite Functions}.\hskip 1em plus 0.5em minus 0.4em\relax Springer, 2023.

\bibitem{Reutenauer12}
C.~Reutenauer, ``On a matrix representation for polynomially recursive sequences,'' \emph{The Electronic Journal of Combinatorics [electronic only]}, vol.~19, 09 2012.

\bibitem{BunaMargineanCSW24}
\BIBentryALTinterwordspacing
A.~Buna{-}Marginean, V.~Cheval, M.~Shirmohammadi, and J.~Worrell, ``On learning polynomial recursive programs,'' \emph{Proc. {ACM} Program. Lang.}, vol.~8, no. {POPL}, pp. 1001--1027, 2024. [Online]. Available: \url{https://doi.org/10.1145/3632876}
\BIBentrySTDinterwordspacing

\bibitem{Clemente24}
\BIBentryALTinterwordspacing
L.~Clemente, ``Weighted basic parallel processes and combinatorial enumeration,'' in \emph{35th International Conference on Concurrency Theory, {CONCUR}}, ser. LIPIcs, R.~Majumdar and A.~Silva, Eds., vol. 311.\hskip 1em plus 0.5em minus 0.4em\relax Schloss Dagstuhl - Leibniz-Zentrum f{\"{u}}r Informatik, 2024, pp. 18:1--18:22. [Online]. Available: \url{https://doi.org/10.4230/LIPIcs.CONCUR.2024.18}
\BIBentrySTDinterwordspacing

\bibitem{Boreale19}
\BIBentryALTinterwordspacing
M.~Boreale, ``Algebra, coalgebra, and minimization in polynomial differential equations,'' \emph{Log. Methods Comput. Sci.}, vol.~15, no.~1, 2019. [Online]. Available: \url{https://doi.org/10.23638/LMCS-15(1:14)2019}
\BIBentrySTDinterwordspacing

\bibitem{CASTIGLIONE201774}
G.~Castiglione and P.~Massazza, ``On a class of languages with holonomic generating functions,'' \emph{Theoretical Computer Science}, vol. 658, pp. 74--84, 2017.

\bibitem{BostanCKN20}
A.~Bostan, A.~Carayol, F.~Koechlin, and C.~Nicaud, ``Weakly-unambiguous parikh automata and their link to holonomic series,'' in \emph{47th International Colloquium on Automata, Languages, and Programming, {ICALP}}, ser. LIPIcs, vol. 168.\hskip 1em plus 0.5em minus 0.4em\relax Schloss Dagstuhl - Leibniz-Zentrum f{\"{u}}r Informatik, 2020, pp. 114:1--114:16.

\bibitem{BellBY12}
\BIBentryALTinterwordspacing
J.~P. Bell, S.~Burris, and K.~A. Yeats, ``Monadic second-order classes of forests with a monadic second-order 0-1 law,'' \emph{Discret. Math. Theor. Comput. Sci.}, vol.~14, no.~1, pp. 87--108, 2012. [Online]. Available: \url{https://doi.org/10.46298/dmtcs.566}
\BIBentrySTDinterwordspacing

\bibitem{fischer2011application}
E.~Fischer, T.~Kotek, and J.~A. Makowsky, ``Application of logic to combinatorial sequences and their recurrence relations,'' \emph{Model Theoretic Methods in Finite Combinatorics}, vol. 558, pp. 1--42, 2011.

\bibitem{Comon1997TreeAT}
\BIBentryALTinterwordspacing
H.~Comon, ``Tree automata techniques and applications,'' 1997. [Online]. Available: \url{https://api.semanticscholar.org/CorpusID:2092186}
\BIBentrySTDinterwordspacing

\bibitem{BerstelR82}
J.~Berstel and C.~Reutenauer, ``Recognizable formal power series on trees,'' \emph{Theor. Comput. Sci.}, vol.~18, pp. 115--148, 1982.

\bibitem{Stanley80}
R.~P. Stanley, ``Differentiably finite power series,'' \emph{Eur. J. Comb.}, vol.~1, no.~2, pp. 175--188, 1980.

\bibitem{FlajoletZC94}
\BIBentryALTinterwordspacing
P.~Flajolet, P.~Zimmermann, and B.~V. Cutsem, ``A calculus for the random generation of labelled combinatorial structures,'' \emph{Theor. Comput. Sci.}, vol. 132, no.~2, pp. 1--35, 1994. [Online]. Available: \url{https://doi.org/10.1016/0304-3975(94)90226-7}
\BIBentrySTDinterwordspacing

\bibitem{Hurwitz1932}
\BIBentryALTinterwordspacing
A.~Hurwitz, \emph{Sur le d{\'e}veloppement des fonctions satisfaisant {\`a} une {\'e}quation diff{\'e}rentielle alg{\'e}brique}.\hskip 1em plus 0.5em minus 0.4em\relax Basel: Springer Basel, 1932, pp. 295--298. [Online]. Available: \url{https://doi.org/10.1007/978-3-0348-4161-0_16}
\BIBentrySTDinterwordspacing

\end{thebibliography}

\appendices


\section{Detailed proofs of Proposition~\ref{prop:generalised}}
\label{app:generalised}

The proof follows from the following result. 
\begin{proposition}
Let $\A = (d,\mu)$ be a generalised differential tree automaton over an alphabet $\Sigma$ with maximum arity~$m$. There exists a differential tree automaton $\B = (d',\mu')$ over $\Sigma$ such that  $[\![\mathcal A]\!]=[\![\mathcal B]\!]$. 
\end{proposition}

\begin{proof}
By definition, for all $\ell \in \{1,\ldots,m\}$ and $\sigma \in \Sigma_\ell$, each entry of $\mu(\sigma)$ lies in $\KK(x_0,\ldots,x_\ell)$.
Without loss of generality, we can assume that there are polynomials $Q_i \in \K[x_i]$ for $i \in \{0,\ldots,m\}$ such that
for all $\ell \in \{1,\ldots,m\}$ and $\sigma \in \Sigma_\ell$ we have 
\[
\mu(\sigma) = \frac{1}{Q_0(x_0) \cdots Q_\ell(x_\ell)} M(\sigma)
\]
where $M(\sigma) \in \K[x_0,\ldots,x_\ell]^{d^\ell \times d}$. 
Let $r$ be the maximum degree of any variable appearing in an entry of one of matrices $M(\sigma)$ and
define $I := \{1,\ldots,d\}\times \{0,\ldots,m\}\times \{0,\ldots,r\} \cup \{ 1 \}$.
Automaton $\B$ has dimension $d':=|I|$.  
The transition map $\mu'$ is defined to ensure that for all $t\in T_\Sigma$ 
and $(i,j,k) \in I$ we have 
\begin{gather}
\mu'_{(i,j,k)}(t) = \mu(t)_i \cdot \frac{\size{t}^k}{Q_j(\size{t})} \qquad\text{and}\qquad
\mu'_{1}(t) = \mu(t)_1 \, .
\label{eq:INV}
\end{gather}
Note that the second condition implies that $[\![\mathcal A]\!]=[\![\mathcal B]\!]$.
We maintain the invariant~\eqref{eq:INV} inductively.  To establish the base case, for $\sigma\in \Sigma_0$  we define
\[ \mu'(\sigma)_{(i,j,k)} := \left\{ \begin{array}{ll} \frac{\mu(\sigma)_i}{Q_j(0)} & \text{if $k=0$,}\\
0 & \text{otherwise}\end{array}\right.
\]
and $\mu'(\sigma)_{1} := \mu(\sigma)_1$.
For the inductive step, given tuples
$(i_1,1,k_1),\ldots,(i_\ell,\ell,k_\ell),(i,j,k) \in I$ and $k_0 \in \{1,\ldots,r\}$,
we define
\begin{align}
\mu'(\sigma)_{(i_1,1,k_1),\ldots,(i_\ell,\ell,k_\ell),(i,j,k)} &\, := \, \frac{P(x_0)}{Q_0(x_0)} \frac{x_0^{k}}{Q_j(x_0)}, \\
\intertext{and}
\mu'(\sigma)_{(i_1,1,k_1),\ldots,(i_\ell,\ell,k_\ell),1} &\, := \, \frac{R(x_0)}{Q_0(x_0)} \, ,
\end{align}
where $P(x_0) , R(x_0) \in \K[x_0]$ are the coefficients of the monomial $x_1^{k_1} \cdots x_\ell^{k_\ell}$ in $M(\sigma)_{i_1,\ldots,i_\ell,i}$
and $M(\sigma)_{i_1,\ldots,i_\ell,1}$ respectively.
All remaining entries of $\mu(\sigma)$ are zero.  
The invariant~\eqref{eq:INV} may now be established by an easy induction on $t \in T_\Sigma$.  
\end{proof}

\section{Detailed proofs of Theorem~\ref{thm:MAIN3}}
\label{app:generalised-multi-holonomic}

In the definition of multi-holonomic sequences, the matrices defining the sequences are all univariate, that is with elements in $\KK(x_0)$. On the other hand, in the generalised multi-holonomic, not only we added guards but the matrices are also multivariate, that is in $\KK(x_0,\ldots,x_\ell)$ for different values of $\ell$.  In the section, we show that these definition are equivalent. For this purpose, we introduce a subclass of generalised multi-holonomic sequences corresponding to the multivariate and unguarded sequences. Formally, a sequence $(\va_n)_{n=0}^\infty$ of elements of $\K^d$ is \emph{unguarded multi-holonomic} if there exists $m \in \N$ and for all $\ell \in \{1, \ldots, m\}$ a matrix of polynomials 
$\mu_\ell \in \KK(x_0,\ldots,x_\ell)^{d^{\ell} \times d}$ 
such that the following holds for all $n \geq 1$.
\[
  \boldsymbol a_n =  \, \sum_{\substack{\ell \in \{1,\ldots,m\}\\\bm{j} \in J_n(\ell)}} (\va_{j_1}\otimes \ldots \otimes \va_{j_\ell})  \cdot\mu_\ell(n,\bm{j})
\]

A generalised multi-holonomic of arity $m$ is defined by matrices whose elements belong to $\KK(x_0,\ldots,x_\ell)$ for some different values of $\ell$. However, the functions in the matrices do not necessarily rely on all variables. For example, we consider the case where we only rely on $x_0$, i.e. with elements in $\KK(x_0)$. In that case, we say that the sequence is univariate generalised multi-holonomic.


\begin{proposition}
\label{prop:unguarded to multi}
Let $(\va_n)_{n=0}^\infty$ be a unguarded  multi-holonomic sequence of order $d$ and arity $m$ and degree $r$. There exists a multi-holonomic sequence $(\vb_n)_{n=0}^\infty$ of order $D = d + \frac{d(r+1)m(m+1)}{2}$ such that for all $n \in \N$, for all $i \in \{1, \ldots, d\}$, $a_{n,i} = b_{n,i}$.
\end{proposition}

\begin{proof}
By definition there exists $m \in \N$ and for all $\ell \in \{1, \ldots, m\}$ a matrix of polynomials 
$\mu_\ell \in \KK(x_0,\ldots,x_\ell)^{d^{\ell} \times d}$ 
and the following holds for all $n \geq 1$.
\[
  \boldsymbol a_n =  \, \sum_{\substack{\ell \in \{1,\ldots,m\}\\\bm{j} \in J_n(\ell)}} (\va_{j_1}\otimes \ldots \otimes \va_{j_\ell})  \cdot\mu_\ell(n,\bm{j})
\]
Without loss of generality, we can assume that all elements in the matrices $\mu_1,\ldots, \mu_m$ share a common denominator in $x_0$. That is there exist $Q \in \K[x]$ with no positive zero and for all $\ell \in \{1, \ldots, m\}$ a matrix $M_\ell \in \K[x_1,\ldots,x_\ell]^{d^\ell\times d}$ and polynomials $Q_{\ell,i} \in \Q[x_i]$ for $i \in \{1, \ldots, m\}$ with no non-negative zero such that
\begin{align}
\mu_\ell(x_0,\ldots,x_\ell) = \frac{1}{Q(x_0)} \frac{1}{Q_{\ell,1}(x_1)\ldots Q_{\ell,\ell}(x_\ell)} M_\ell(x_0,\ldots,x_\ell)
\end{align}
As $r$ is the degree of $(\va_n)_{n=0}^\infty$, by unfolding the definition of the Kronecker product, we deduce that there exists univariate matrices $M_{\ell,(i_1,\ldots,i_\ell)}(x) \in \K[x]^{d^\ell \times d}$ for all $i_1,\ldots, i_\ell \in \{0, \ldots, r\}$ such that
\begin{align}
\mu_\ell(x_0,\ldots,x_\ell) = \sum_{i_1,\ldots,i_\ell \in \{0, \ldots, r\}} \frac{1}{Q(x_0)} \frac{x_1^{i_1}\ldots x_\ell^{i_\ell} }{Q_{\ell,1}(x_1)\ldots Q_{\ell,\ell}(x_\ell)}  M_{\ell,(i_1,\ldots,i_\ell)}(x_0)
\end{align}
Therefore, we obtain for all $n \geq 1$, 
\begin{align*}
\va_n &= \sum_{\substack{\ell \in \{1,\ldots,m\}\\\bm{j} \in J_n(\ell)}} (\va_{j_1}\otimes \ldots \otimes \va_{j_\ell})  \cdot\mu_\ell(n,\bm{j})\\
&= \frac{1}{Q(n)} \sum_{\substack{\ell \in \{1,\ldots,m\}\\\bm{j} \in J_n(\ell)}} \sum_{i_1,\ldots,i_\ell \in \{0,\ldots,r\}} (\frac{j_1^{i_1}}{Q_{\ell,1}(j_1)}\va_{j_1}\otimes \ldots \otimes \frac{j_\ell^{i_\ell}}{Q_{\ell,\ell}(j_\ell)}\va_{j_\ell}) \cdot M_{\ell,(i_1,\ldots,i_\ell)}(n)
\end{align*}
We introduce for all $i \in \{0, \ldots, r\}$ and for all $\ell \in \{1, \ldots, m\}$, for all $k \in \{1, \ldots, j\}$ a new sequences $(\bm{c}_{n,(i,\ell,k)})_{n=0}^\infty$ of order $d$ such that for all $n \in \N$,
\[
\bm{c}_{n,(i,\ell,k)} = \frac{n^{i}}{Q_{\ell,k}(n)}\va_n
\]
Notice that they are well defined since the polynomials $Q_{\ell,k}$ do not have non-negative zeros.
Hence the previous equation becomes:
\begin{align*}
\va_n &= \frac{1}{Q(n)} \sum_{\substack{\ell \in \{1,\ldots,m\}\\\bm{j} \in J_n(\ell)}} \sum_{i_1,\ldots,i_\ell \in \{0,\ldots,r\}} (\bm{c}_{j_1,(i_1,\ell,1)} \otimes \ldots \otimes \bm{c}_{j_\ell,(i_\ell,\ell,\ell)}) \cdot M_{\ell,(i_1,\ldots,i_\ell)}(n)
\end{align*}
Let us define a new sequence $(\vb_n)_{n=0}^\infty$ of order $D = d + \frac{d(r+1)m(m+1)}{2}$ such that
\[
\vb_n = \begin{bmatrix} \va_n & \bm{c}_{n,(0,1,1)} & \bm{c}_{n,(0,2,1)} & \bm{c}_{n,(0,2,2)} & \ldots & \bm{c}_{n,(r,m,m)} \end{bmatrix}
\]
To simplify access to this vector, let us define $\alpha : \{1, \ldots, d\} \rightarrow \{1, \ldots, D\}$, $\beta_\ell : \{0,\ldots,r\} \times \{1, \ldots, \ell\} \times \{1, \ldots, d\} \rightarrow \{1, \ldots, D\}$ for all $\ell \in \{1, \ldots, m\}$ such that 
\[
b_{n,\alpha(j)} = a_{n,j} \qquad b_{n,\beta_\ell(i,k,j)} = c_{n,(i,\ell,k),j}
\]
Thus, for all $\ell \in \{1, \ldots, m\}$, let us define a new matrix $\mu'_\ell \in \KK(x_0)^{D^\ell \times D}$ such that:
\begin{itemize}
\item $(\mu'_\ell)_{(\beta_\ell(i_1,1,j_1),\ldots,\beta_\ell(i_\ell,\ell,j_\ell)),\alpha(j)} = \frac{1}{Q(x)} M_{\ell,(i_1,\ldots,i_\ell)}(x)_{(j_1,\ldots,j_\ell),j}$
\item $(\mu'_\ell)_{(\beta_\ell(i_1,1,j_1),\ldots,\beta_\ell(i_\ell,\ell,j_\ell)),\beta_{\ell'}(i',k',j)} = \frac{x^{i'}}{Q_{\ell',k'}(x)} \frac{1}{Q(x)} M_{\ell,(i_1,\ldots,i_\ell)}(x)_{(j_1,\ldots,j_\ell),j}$
\item $0$ in all other cases.
\end{itemize}
By construction, we have that for all $n \geq 1$, 
\[
\vb_n = \sum_{\substack{\ell \in \{1,\ldots,m\}\\\bm{j} \in J_n(\ell)}} (\vb_{j_1} \otimes \ldots \otimes \vb_{j_\ell}) \cdot \mu'_\ell(n)
\]
which allows us to conclude.
\end{proof}


\begin{proposition}
\label{prop:linear map}
Let $(\va_n)_{n=0}^\infty$ an unguarded multi-holonomic sequence of elements of $\K^d$. Let $d' \in \N^+$. Let $M$ be a matrix in $\K^{d \times d'}$, i.e. a linear map from $\K^d$ to $\K^{d'}$. There exists a unguarded multi-holonomic sequence $(\vb_n)_{n=0}^\infty$ such that for all $n \in \N$, 
\[
\vb_n = \begin{bmatrix} \va_nM & \va_n \end{bmatrix}
\]
\end{proposition}

\begin{proof}
By definition, we know that $(\va_n)_{n=0}^\infty$ satisfies for all $n \geq 1$, 
\[
\va_n = \sum_{\substack{\ell \in \{1,\ldots,m\}\\\bm{j} \in J_n(\ell)}}
(\boldsymbol a_{j_1} \otimes \cdots \otimes \boldsymbol a_{j_\ell} ) \cdot \mu_\ell(n,\bm{j})
\]
Consider for all $\ell \in \{1, \ldots, m\}$, the linear maps $\overline{\mu}_\ell \in \KK(x_0,\ldots,x_\ell)^{(d+d')^\ell \times (d+d')}$ such that:
\[
\overline{\mu}_\ell = (\begin{bmatrix} \boldsymbol{u}_1 & \boldsymbol{v}_1 \end{bmatrix} \otimes \ldots \otimes \begin{bmatrix} \boldsymbol{u}_\ell & \boldsymbol{v}_\ell \end{bmatrix}) = \begin{bmatrix} (\boldsymbol{v}_1 \otimes \ldots \otimes \boldsymbol{v}_\ell)\, \mu_\ell\, M & (\boldsymbol{v}_1 \otimes \ldots \otimes \boldsymbol{v}_\ell)\mu_\ell\end{bmatrix}
\]
We show the desired result by induction on $n$. The base case $n = 0$ is trivial. For $n \geq 1$, by definition of $\overline{\mu}_1, \ldots, \overline{\mu}_m$, 
\begin{align*}
& \sum_{\substack{\ell \in \{1,\ldots,m\}\\\bm{j} \in J_n(\ell)}} (\vb_{j_1} \otimes \cdots \otimes \vb_{j_\ell} ) \overline{\mu}_\ell(n,\bm{j})\\
=\ & \sum_{\substack{\ell \in \{1,\ldots,m\}\\\bm{j} \in J_n(\ell)}}
(\begin{bmatrix} \va_{j_1}M & \va_{j_1} \end{bmatrix} \otimes \cdots \otimes \begin{bmatrix} \va_{j_\ell}M & \va_{j_\ell} \end{bmatrix} ) \overline{\mu}_\ell(n,\bm{j})\\
=\ & \sum_{\substack{\ell \in \{1,\ldots,m\}\\\bm{j} \in J_n(\ell)}}
\begin{bmatrix} (\va_1 \otimes \ldots \otimes \va_\ell)\, \mu_\ell(n,\bm{j})\, M & (\va_1 \otimes \ldots \otimes \va_\ell)\mu_\ell(n,\bm{j})\end{bmatrix}\\
=\ &\begin{bmatrix}\va_n\, M & \va_n\end{bmatrix}\qedhere
\end{align*}
\end{proof}

In a generalised multi-holonomic recurrence, the 
index set $I$ is a subset of $\Z \times \N^*$.  
This means that the shift for the sum $j_1 + \ldots + j_\ell$ (the difference between this sum and $n$) can be either positive or negative. We start by showing how to transform a generalised multi-holonomic sequences with negative shift to a multi-holonomic sequence.


\begin{proposition}
\label{prop:generalised guard to multi-holonomic}
Let $(\boldsymbol a_n)_{n=0}^\infty$ in $\K^d$ be a generalised multi-holonomic sequence where
$I \subseteq \mathbb N^+$ is a finite set, and for every $\boldsymbol{m} = (m_1,\ldots,m_\ell)$ in $I$, $\mu_{\boldsymbol{m}} \in \KK(x_0,\ldots,x_\ell)^{d^\ell \times d}$, and for all $n \geq 1$, 
\begin{gather}
    \boldsymbol a_n = \sum_{\substack{\bm{m} \in I\\\bm{j} \in G_n(\bm{m})}} 
(\boldsymbol a_{j_1} \otimes \cdots \otimes \boldsymbol a_{j_\ell} ) \cdot \mu_{\boldsymbol m}(n,\bm{j})
\label{eq:prop-GEN-HOL}
\end{gather}
Then there exists an unguarded multi-holonomic sequence $(\vb_n)_{n=0}^\infty$ of order $(M+1)\,d + 1$ where $M$ is the greater than any integer in $I$ and such that for all $i \in \{1, \ldots, d\}$, for all $n \in \N$, $\va_{n,i} = \vb_{n,i}$.
\end{proposition}

\begin{proof}
Let $M\in\mathbb N$ be such that $M>m_0,\ldots,m_\ell$ for all $(m_0,\ldots,m_\ell)\in I$
and let $\boldsymbol e_0,\ldots,\boldsymbol e_{M-1}$ be the canonical basis of $\K^{M}$. Let $D = d\,M+1$.
We will show that the sequence $(\overline{\va}_n)_{n=0}^\infty$ of vectors in $\K^{dM}$ defined by
\[ 
\overline{\va}_n := \begin{bmatrix} \boldsymbol a_n \otimes \boldsymbol e_{\min(n,M-1)}  & 1\end{bmatrix} \qquad(n\in\mathbb N)
\]
is a multi-holonomic.
To this end, given $\boldsymbol m=(m_0,\ldots,m_\ell) \in I$,  when $\ell \geq 1$, we define the linear maps
\[
\delta  : \K^{M^\ell} \rightarrow \K^M,\qquad
\varepsilon : \K^{M^{(\ell+1)}} \rightarrow \K,\qquad 
\widetilde{\mu}_{\boldsymbol m} : \K^{D^\ell} \rightarrow \K^{D}[x_0,\ldots,x_{\ell+1}]
\]
respectively by the following equations, where $j_1,\ldots,j_{\ell+1} \in \{0,\ldots,M-1\}$, and  $u_1,\ldots,u_{\ell+1} \in \K^d$,
\begin{align*}
    (\boldsymbol e_{j_1}\otimes \cdots \otimes \boldsymbol e_{j_l}) \, \delta  \, = \, & \boldsymbol e_{\min(j_1+\cdots+j_\ell+1,M-1)} \\
    (\boldsymbol e_{j_1} \otimes \cdots \otimes \boldsymbol e_{j_{\ell+1}}) \, \varepsilon  \, = \, & \begin{cases} 
        1 & \; \text{when }j_{\ell+1} = m_0 \wedge \bigwedge_{s=1}^\ell m_{s} \leq \sum_{\substack{i\in \{1,\ldots,\ell\}\setminus\{s\}}} j_i\\
        0 & \mbox{otherwise}
    \end{cases}
\end{align*}
and 
\begin{align*}
    &(\begin{bmatrix}\boldsymbol u_1\otimes \boldsymbol e_{j_1} & w_1 \end{bmatrix} \otimes \cdots \otimes \begin{bmatrix} \boldsymbol u_{\ell+1} \otimes \boldsymbol e_{j_{\ell+1}}& w_{\ell+1}\end{bmatrix}) \cdot \overline{\mu}_{\boldsymbol m}(x_0,\ldots,x_{\ell+1}) \, =\, \\ 
    &\qquad
    \begin{bmatrix}
    (\boldsymbol e_{j_1}\otimes \cdots \otimes \boldsymbol e_{j_{\ell+1}})\,\varepsilon \cdot 
     (\boldsymbol u_1\otimes \cdots \otimes\boldsymbol u_{\ell})\,\mu_{\boldsymbol m}(x_0,\ldots,x_{\ell}) \otimes 
    (\boldsymbol e_{j_1} \otimes\cdots \otimes  \boldsymbol e_{j_\ell}) \,\delta & 0
    \end{bmatrix} 
\end{align*}

Since $\boldsymbol{e}_0,\ldots,\boldsymbol{e}_{M-1}$ form a basis of $\K^M$, these linear maps are fully defined. 

By construction, for all $n \geq 1$, for all $j_1+\cdots + j_\ell=n-1 - m_0$, if
$j_1\leq n-1-m_1,\ldots,j_{\ell}\leq n-1-m_{\ell}$ and taking $j_{\ell+1} = m_0$ then for all $s \in \{1, \ldots, \ell\}$, $j_s \leq n-1-m_s$ implies $0 \leq (\sum_{\substack{i\in \{1,\ldots,\ell\}\setminus\{s\}}} j_i) - m_s$. Therefore, $(\boldsymbol e_{j_1} \otimes \cdots \otimes \boldsymbol e_{j_{\ell+1}}) \, \varepsilon = 1$ and so:
\[
(\overline{\boldsymbol a}_{j_1}\otimes\cdots\otimes \overline{\boldsymbol a}_{j_{\ell+1}}) \, \overline{\mu}_{\boldsymbol m}(j_1,\ldots,j_{\ell+1})
= \begin{bmatrix} {(\boldsymbol a_{j_1} \otimes \cdots \otimes \boldsymbol a_{j_\ell})\mu_{\boldsymbol m}}(j_1,\ldots,j_{\ell}) \otimes \boldsymbol e_{\min(j_1+\cdots+j_{\ell}+1,M-1)} & 0 \end{bmatrix}.
\]
Otherwise, when $j_{\ell+1} \neq m_0$ or $j_s > n -1 - m_s$ for some $s \in \{1, \ldots, \ell\}$, as $(\boldsymbol e_{j_1} \otimes \cdots \otimes \boldsymbol e_{j_{\ell+1}}) \, \varepsilon = 0$, we deduce that $(\overline{\boldsymbol a}_{j_1}\otimes\cdots\otimes \overline{\boldsymbol a}_{j_{\ell+1}}) \, \overline{\mu}_{\boldsymbol{m}}=0$.
Therefore, we deduce that:
\begin{align*}
    & \sum_{\substack{j_1+\cdots+j_{\ell+1} = n - 1}} 
    (\overline{\va}_{j_1} \otimes \cdots \otimes \overline{\va}_{j_{\ell+1}} ) \cdot \overline{\mu}_{\boldsymbol m}(n,j_1,\ldots,j_{\ell+1}) \\
    =\ & \sum_{\substack{j_1+\cdots+j_{\ell} = n - 1 +m_0 \\ j_1\leq n-1-m_1,\ldots,j_\ell \leq n-1-m_\ell}} 
    \begin{bmatrix} (\va_{j_1} \otimes \cdots \otimes \va_{j_{\ell}} ) \cdot \mu_{\boldsymbol m}(n,j_1,\ldots,j_{\ell}) \otimes \boldsymbol{e}_{\min(n,M-1)} & 0 \end{bmatrix}
\end{align*}

When $\ell =1$, that is $\bm{m} = (m_0)$, we have in fact that $\mu_{\bm{m}}$ is in fact a constant matrix. In such a case, we assume w.l.o.g. that there is a single such $\bm{m}$ the value of $m_0$ being unimportant (if there are several singleton in $I$, we can sum the constants). Thus, we consider the linear map $\overline{\mu}_{\bm{m}} : \K^D \rightarrow \K^D$ such that 
\[
\begin{bmatrix}\boldsymbol u\otimes \boldsymbol e_{j} & w \end{bmatrix} \cdot \overline{\mu}_{\bm{m}} = \begin{bmatrix}w\mu_{\bm{m}} \otimes \boldsymbol e_{j}\delta & w \end{bmatrix}
\]

Hence by the bilinearity of Kronecker product, we have for all $n \geq 1$,  
\begin{equation}
\overline{\va}_n = \sum_{\boldsymbol m = (m_1,\ldots,m_\ell) \in I }\sum_{j_1+\cdots+j_{\ell+1} = n - 1} 
(\overline{\va}_{j_1} \otimes \cdots \otimes \overline{\va}_{j_{\ell+1}} ) \cdot \overline{\mu}_{\boldsymbol m}(n,j_1,\ldots,j_{\ell+1})
\label{prop:eq min}
\end{equation}
To complete the proof, let $k$ be the size of the largest tuple in $I$. We partition $I$ into $I_1, \ldots, I_k$ such that $I_\ell$ is the set of $\bm{m} \in I$ such that $\overline{\mu}_{\bm{m}} \in \KK(x_0,\ldots,x_\ell)^{D^\ell \times D}$. We therefore obtain
\begin{align*}
\overline{\va}_n &= \sum_{\ell = 1}^k \sum_{\substack{\boldsymbol m \in I\\\bm{j} \in J_n(\ell)}} (\overline{\va}_{j_1} \otimes \cdots \otimes \overline{\va}_{j_\ell} ) \cdot \overline{\mu}_{\boldsymbol m}(n,\bm{j})\\
& = \sum_{\ell = 1}^k \sum_{\bm{j} \in J_n(\ell)} (\overline{\va}_{j_1} \otimes \cdots \otimes \overline{\va}_{j_\ell} ) \cdot \left(\sum_{\boldsymbol m \in I} \overline{\mu}_{\boldsymbol m}(n,\bm{j})\right)\,.
\end{align*}
We define the linear map $\gamma : \K^{D} \rightarrow \K^d$ such that for all $j \in \{0, \ldots, M-1\}$, $\begin{bmatrix} \bm{u} \otimes e_j & w \end{bmatrix} \gamma = \bm{u}$ and we conclude by applying \Cref{prop:linear map} with $(\overline{\va}_n)_{n=0}^\infty$ and $\gamma$.
\end{proof}

For the next proposition, recall that $G_n(s_0,\ldots,s_\ell)$ is the set of tuples $(j_1,\ldots, j_\ell) \in \N^\ell$ such that $j_1 + \ldots + j_\ell = n - 1 + s_0$ and for all $k \in \{1, \ldots, \ell\}$, $j_k \leq n-1-s_k$.


\begin{restatable}{proposition}{propgeneralisedmultiholonomic}
\label{prop:generalised to multiholonomic}
Let $(\boldsymbol a_n)_{n=0}^\infty$ be a generalised multi-holonomic sequence of elements of $\K^d$. There exists an unguarded multi-holonomic sequence $(\boldsymbol b_n)_{n=0}^\infty$ of order $D \geq d$ such that for all $i \in \{1, \ldots, d\}$, for all $n \in \N$, $\va_{n,i} = \vb_{n,i}$.
\end{restatable}

\begin{proof}
By definition of a generalised multi-holonomic sequence, we know that $(\va_n)_{n=0}^\infty$ satisfies for all $n \geq 1$, 
\[
\va_n = \sum_{\boldsymbol{m}\in I}\; \sum_{\bm{j}\in G_n(\bm{m})} 
(\boldsymbol a_{j_1} \otimes \cdots \otimes \boldsymbol a_{j_\ell} ) \cdot \mu_{\boldsymbol m}(n,\bm{j})
\]
with $I \subseteq \Z \times \N^*$. The main idea will be to create a new set $I'$ that produce the same sum but which fit the conditions of \Cref{prop:generalised guard to multi-holonomic}, that are only negative shift.

Let us take $M$ such that
\begin{equation}
\label{eq:prop:generalised to multiholonomic}
    M > \max_{\bm{m} \in I}\left\{ \sum_{m \in \bm{m}} |m| \right\}\,.
\end{equation}
We will define $(\vb_n)_{n=0}^\infty$ the sequence such that for all $n \in \N$, 
\[
\vb_n = \begin{bmatrix} \va_n & \ldots & \va_{n+M} \end{bmatrix}
\]
We define for all $i \in \{0, \ldots, M\}$ the linear maps $\pi_i: \K^{d(M+1)} \rightarrow \K^d$ for all $\boldsymbol{u}_0, \ldots, \boldsymbol{u}_M \in \K^d$, 
\[
\begin{bmatrix}
    \boldsymbol{u}_1 & \ldots & \boldsymbol{u}_M 
\end{bmatrix}\pi_i = \boldsymbol{u}_i
\]

Thus, for all $n \geq 1$, 
\[
\vb_n\pi_M = \sum_{\boldsymbol{m} \in I}\; \sum_{\bm{j}\in G_{n+M}(\bm{m})} 
(\boldsymbol a_{j_1} \otimes \cdots \otimes \boldsymbol a_{j_\ell} ) \cdot \mu_{\boldsymbol m}(n+M,\bm{j})
\]
Let us denote by $A(n,\bm{j})$ the expression $(\boldsymbol a_{j_1} \otimes \cdots \otimes \boldsymbol a_{j_\ell} ) \cdot \mu_{\boldsymbol m}(n+M,\bm{j})$.
Take $\boldsymbol{m}=(m_0,\ldots,m_\ell)\in I$. 
In the internal sum, given $(j_1, \ldots, j_\ell) \in G_{n+M}(\bm{m})$, we know that $j_1 + \ldots + j_\ell = n-1+M-m_0$ and $j_1 \leq n-1+M-m_1$, \ldots, $j_\ell \leq n-1+M-m_\ell$. We want to partition the tuples of $G_{n+M}(\bm{m})$ into $2^\ell$ subsets depending on whether each coordinate $j_s$ is either smaller or greater than $M-m_s$. Thus, given a subset $S \subseteq \{1, \ldots, \ell\}$, the restriction of $G_{n+M}(\bm{m})$ to $S$, denoted $G_{n+M}(\bm{m})_{|S}$ will corresponds to all tuples $(j_1, \ldots, j_\ell) \in G_{n+M}(\bm{m})$ such that $j_s \leq M - m_s$ if and only if $s \in S$. Therefore, we obtain:
\begin{align*}
\vb_n\pi_M &= \sum_{\boldsymbol{m} \in I}\; \sum_{\bm{j}\in G_{n+M}(\bm{m})} 
    (\boldsymbol a_{j_1} \otimes \cdots \otimes \boldsymbol a_{j_\ell} ) \cdot \mu_{\boldsymbol m}(n+M,\bm{j})
\\
&=\sum_{\bm{m}=(m_0,\ldots,m_\ell) \in I}\; \sum_{S \subseteq \{1,\ldots,\ell\}}\; \sum_{\bm{j} \in G_{n+M}(\bm{m})_{|S}} (\boldsymbol a_{j_1} \otimes \cdots \otimes \boldsymbol a_{j_\ell} ) \cdot \mu_{\boldsymbol m}(n+M,\bm{j})\\
&=\sum_{\bm{m}=(m_0,\ldots,m_\ell) \in I}\; \sum_{S \subseteq \{1,\ldots,\ell\}}\; \sum_{\bm{j} \in G_{n+M}(\bm{m})_{|S}} A(n,\bm{j})
\end{align*}

Let us look at one instance of $\bm{m} = (m_0,\ldots,m_\ell) \in I$ and $S \subseteq \{1, \ldots, \ell\}$. Let us define $k_1, \ldots, k_\ell$ such that $k_s = 0$ when $s \in S$ and $k_s = M - m_s$ otherwise. By definition, the internal sum can be split as follows:
\begin{align*}
& \sum_{\bm{j} \in G_{n+M}(\bm{m})_{|S}} A(n,\bm{j})\\
=\ & \sum_{\substack{\forall s \in S, j_s \leq M - m_s\\L = \sum_{s \in S} j_s}}\;
\sum_{\substack{\sum_{s \not\in S} j_s = n-1+M+m_0-L\\ \forall s \not\in S, M-m_s < j_s \leq n + M - 1 -m_s}}\; 
A(n,\bm{j}) \\
=\ & 
\sum_{\substack{\forall s \in S, j_s \leq M - m_s\\L = \sum_{s \in S} j_s}}\;
\sum_{\substack{\sum_{s \not\in S} j_s = n-1+M+m_0-L -(\ell-|S|)M + \sum_{s \not\in S} m_s \\ 
\forall s \not\in S', 0 < j_s \leq n-1}} A(n,j_1+k_1,\ldots,j_\ell+k_\ell)\notag
\end{align*}
Intuitively, we will build a new set $I' \subseteq \N^+$ with an element added for each $S = \{i_1,\ldots,i_k\} \subseteq \{1, \ldots, \ell\}$ and each $(j_{i_1},\ldots,j_{i_k})$ with $j_{i_1} \leq M - m_{i_1}, \ldots, j_{i_k} \leq M - m_{i_k}$. In particular, we will add (under some conditions described below) the element $\overline{\boldsymbol{m}} = (M+m_0-L-(\ell-|S|)M + \sum_{s\not\in S} m_s)$ associated to the transition matrix $\overline{\mu}_{\overline{\boldsymbol{m}}}$ such that:
\begin{align*}
(\boldsymbol{u}_{o_1} \otimes \ldots \otimes \boldsymbol{u}_{o_\epsilon}) \overline{\mu}_{\overline{\boldsymbol{m}}}(x_0,x_{o_1},\ldots,x_{o_\epsilon})\pi_M &= (\boldsymbol{v}_{1} \otimes \ldots \otimes \boldsymbol{v}_{\ell}) \mu_{\boldsymbol{m}}(x_0+M,c_1,\ldots,c_\ell)\\
(\boldsymbol{u}_{o_1} \otimes \ldots \otimes \boldsymbol{u}_{o_\epsilon}) \overline{\mu}_{\overline{\boldsymbol{m}}}(x_0,x_{o_1},\ldots,x_{o_\epsilon})\pi_i &= \boldsymbol{0}^d\qquad \forall i \in \{0,\ldots,M-1\}
\end{align*}
where $\{1, \ldots, \ell\} \setminus S = \{ o_1, \ldots, o_\varepsilon\}$ and for all $i \in \{1, \ldots, \ell\}$, if $i \in S$ then $\boldsymbol{v}_i = \va_{j_i}$ and $c_i= j_i$ else $\boldsymbol{v}_i = \boldsymbol{u}_i\pi_{M-m_i}$ and $c_i = x_i + M-m_i$.

Since $M$ satisfies \Cref{eq:prop:generalised to multiholonomic}, we deduce that $M - m_i > 0$ for all $i \in \{1, \ldots, \ell\}$ and so this transition matrix is fully and correctly defined. In particular, the projections $\pi_{M-m_i}$ exist. 

It remains to show that $N = M+m_0-\sum_{s\in S} j_s -(\ell-|S|)M + \sum_{s\not\in S} m_s$ is negative. We do a case analysis on $|S|$:
\begin{itemize}
\item When $|S| = \ell$: $N = M + m_0 - \sum_{s = 1}^m j_s$. However, we know that $n \in \N$ and by \Cref{eq:prop:generalised to multiholonomic}, we also have $N > 0$. Thus, $n \neq 1 - N$ which entails that the inner sum is always $0$.
\item When $|S| = \ell -1$: W.l.o.g., assume $1 \not\in S$. Hence $N = m_0 -\sum_{s = 2}^\ell j_s + m_1$. Hence, as the innermost sum contains the condition $j_1 \leq n-1$, the sum becomes in fact $0$ if $N > 0$.  Hence, we will only add the tuple $\overline{\boldsymbol{m}}$ to $I'$ if  $N \leq 0$.
\item When $|S| < \ell - 1$: Recall that by \Cref{eq:prop:generalised to multiholonomic}, 
\[
\max_{\bm{m} \in I}\left\{ \sum_{m \in \bm{m}} |m| \right\} < M\,.
\]
which entails $N \leq 0$.
\end{itemize}
Finally, we add to $I'$ an element $\boldsymbol{m} = (0,0)$ with the transition matrix $\overline{\mu}_{\overline{\boldsymbol{m}}}$ such that:
\begin{align*}
    \boldsymbol{u} \overline{\mu}_{\overline{\boldsymbol{m}}}(x_0,x_1) \pi_M &= \boldsymbol{0}^d\\
    \boldsymbol{u} \overline{\mu}_{\overline{\boldsymbol{m}}}(x_0,x_1) \pi_i &= \boldsymbol{u} \pi_{i+1} \qquad \forall i \in \{0,\ldots,M-1\}
    \end{align*}
This last element allow to \emph{shift} the values of $\va_{n+k}$ in the vector $\vb_{n}$. In other words, it will ensure that $\vb_{n+1}\pi_{i-1} = \vb_n\pi_i$.

By construction, we have shown that $(\vb_n)_{n=0}^\infty$ is a generalised multi-holonomic sequence with only negative shifts and such that for all $i \in \{1, \ldots, d\}$, for all $n \in \N$, $\va_{n,i} = \vb_{n,i}$. We conclude by applying \Cref{prop:generalised guard to multi-holonomic}.
\end{proof}

\begin{proposition}
\label{prop:multi to arity 2}
Let $(\va_n)_{n=0}^\infty$ be a multi-holonomic sequence of elements in $\K^d$ and arity $m$. There exists a multi-holonomic $(\vb_n)_{n=0}^\infty$ of arity at most $2$ such that for all $i \in \{1, \ldots, d\}$, for all $n \in \N$, $a_{n,i} = b_{n,i}$.
\end{proposition}

\begin{proof}
By definition, there exist for all $\ell \in \{1, \ldots, m\}$ a matrix of polynomials $\mu_\ell \in \KK(x)^{d^\ell \times d}$ such that for all $n \geq 1$, 
\begin{equation}
\label{eq:proof-arity 2}
\va_n = \sum_{\substack{\ell \in \{1, \ldots, m\}\\\bm{j} \in J_n(\ell)}} (\va_{j_1} \otimes \ldots \otimes \va_{j_\ell}) \cdot \mu_\ell(n)
\end{equation}
Let us assume that $m > 2$ as otherwise the result would directly hold with $(\va_n)_{n=0}^\infty = (\vb_n)_{n=0}^\infty$

We define new sequence $(\overline{\va}_n)_{n=0}^\infty$ such that for all $n \in \N$, $\overline{\va}_n = \va_{n+1}$. Moreover, we define $\sum_{\ell=2}^{m-1} d^\ell$ new sequences $c_{n,(j_1,\ldots,j_\ell)}$ for all $j_1,\ldots,j_\ell \in \{1, \ldots, d\}$ and $\ell \in \{2, \ldots, m-1\}$. Let us now define the value of these new sequences. For all $(j_1,j_2) \in \{1, \ldots, d\}$,
\begin{equation}
\label{eq:prop-arity2-1}
c_{n,(j_1,j_2)} := \sum_{k=0}^n a_{k,j_1} a_{n-k,j_2} = \overline{a}_{n-1,j_1}a_{0,j_2} + \sum_{k=0}^{n-1} a_{k,j_1} \overline{a}_{n-k-1,j_2}
\end{equation}
For all $\ell \in \{3, \ldots, m\}$, for all $j_1,\ldots,j_\ell \in \{1, \ldots, d\}$, for all $n \in \N$, 
\begin{equation}
\label{eq:prop-arity2-2}
c_{n,(j_1,\ldots,j_\ell)} = \sum_{k=0}^n c_{k,(j_1,\ldots,j_{\ell-1})} a_{n-k,j_\ell} = c_{n,(j_1,\ldots,j_{\ell-1})}a_{0,j_\ell} + \sum_{k=0}^{n-1} c_{k,(j_1,\ldots,j_{\ell-1})} \overline{a}_{n-k-1,j_\ell}
\end{equation}
We start by showing by induction on $\ell$ that for all $j_1, \ldots, j_\ell \in \{1, \ldots, d\}$, 
\begin{equation}
    \label{eq:prop-arity2-4}
c_{n,(j_1,\ldots,j_\ell)} = \sum_{(k_1,\ldots,k_\ell) \in J_{n+1}(\ell)} a_{k_1,j_1}\ldots a_{k_\ell,j_\ell}
\end{equation}
The base case $\ell = 2$ is given by definition of $c_{n,(j_1,j_2)}$. In the inductive step, we have:
\begin{align*}
&\sum_{(k_1,\ldots,k_\ell) \in J_{n+1}(\ell)} a_{k_1,j_1}\ldots a_{k_\ell,j_\ell} \\
=\ & \sum_{k_1+ \ldots +k_\ell = n} a_{k_1,j_1}\ldots a_{k_\ell,j_\ell}\\
=\ & \sum_{k=0}^n\ (\sum_{k_1+ \ldots +k_{\ell-1} = k} a_{k_1,j_1}\ldots a_{k_{\ell-1},j_{\ell-1}}) a_{n-k,j_\ell}\\
=\ & \sum_{k=0}^n\ c_{k,(j_1,\ldots,j_{\ell-1})} a_{n-k,j_\ell}\qquad \text{by inductive hypothesis}\\
=\ & c_{n,(j_1,\ldots,j_\ell)}
\end{align*}

We define the sequence $(\vb_n)_{n=0}^\infty$ of dimension $D = d(1 + \frac{d^\ell -1}{d-1})$ such that for all $n \in N$, 
\[
\vb_n = 
\begin{bmatrix}
    \va_n & \overline{\va}_n & c_{n,(1,1)} & c_{n,(1,2)} & \ldots & c_{n,(d,\ldots,d)}
\end{bmatrix}
\]
To simplify access to this vector, let us define $\alpha : \{1, \ldots, d\} \rightarrow \{1, \ldots, D\}$, $\beta : \{1, \ldots, d\} \rightarrow \{1, \ldots, D\}$ and $\gamma_\ell : \{1, \ldots, d\}^\ell \rightarrow \{1, \ldots, D\}$ for all $\ell \in \{2, \ldots, m-1\}$ such that 
\[
b_{n,\alpha(j)} = a_{n,j} \qquad b_{n,\beta(j)} = \overline{a}_{n,j} \qquad b_{n,\gamma_\ell(j_1,\ldots,j_\ell)} = c_{n,(j_1,\ldots,j_\ell)}
\]
Let us now show that $\vb_n$ is a multi-holonomic sequence of arity $2$. For all $n \geq 1$, for all $j \in \{1, \ldots, D\}$, we say that $b_{n,j}$ is partially multi-holonomic of arity 2 when there exists some $A_{i,j} \in \KK(x)$ and $B_{(i_1,i_2),j} \in \KK(x)$ for $i,i_1,i_2 \in \{1, \ldots, D\}$ such that
\begin{equation}
\label{eq:prop-arity2-3}
    b_{n,j} = \sum_{i = 1}^D b_{n-1,i} A_{i,j}(n) + \sum_{i_1,i_2 \in \{1, \ldots, D\}} \sum_{k = 0}^{n-1} b_{k,i_1}b_{n-k-1,i_2} B_{(i_1,i_2),j}(n)
\end{equation}
When $b_{n,j}$ is partially multi-holonomic of arity 2 for all $j \in \{1, \ldots, D\}$ then we directly obtain that $\vb_n$ is a multi-holonomic sequence of arity $2$. Indeed by creating two matrices $\mu \in \KK(x)^{D\times D}$ and $\mu' \in \KK(x)^{D^2\times D}$ such that 
\[
\mu_{i,j}(x) = A_{i,j}(x) \quad (i,j \in \{1, \ldots, D\}) \qquad \text{and} \qquad \mu'_{(i_1,i_2),j}(x) = B_{(i_1,i_2),j}(x) \quad (i_1,i_2,j \in \{1, \ldots, D\}) 
\]
we directly obtain that 
\[\vb_n = \vb_{n-1} \cdot \mu(n) + \sum_{k=0}^{n-1} (\vb_{k} \otimes \vb_{n-k-1}) \cdot \mu'(n)\]
which would allow us to conclude.

\medskip

Therefore, let us prove that for all $j \in \{1, \ldots, D\}$, $b_{n,j}$ is partially multi-holonomic of arity 2. We do a case analysis on the indices given by the function $\alpha, \beta$ and $\gamma_\ell$ for $\ell \in \{2, \ldots, m-1\}$.
\begin{itemize}
\item Indices $\alpha(j)$, $j \in \{1, \ldots, d\}$: By definition of $\overline{\va_n}$, we know that for all $n \geq 1$, for all $j \in \{1, \ldots, d\}$,
\[
    b_{n,\alpha(j)} = a_{n,j} = \overline{a}_{n-1,j} = b_{n-1,\beta(j)}\,.
\]
Hence the result holds.
\item We prove the result for the indices given by $\gamma_\ell$ by induction $\ell$. In the base case $\ell = 2$, we have for all $j_1,j_2 \in \{1, \ldots, d\}$, 
\begin{align*}
b_{n,\gamma(j_1,j_2)} &= c_{n,(j_1,j_2)} \\
&= \overline{a}_{n-1,j_1}a_{0,j_2} + \sum_{k=0}^{n-1} a_{k,j_1} \overline{a}_{n-k-1,j_2} \qquad \text{by \Cref{eq:prop-arity2-1}}\\
&= b_{n-1,\beta(j_1)} a_{0,j_2} + \sum_{k=0}^{n-1} b_{k,\alpha(j_1)} b_{n-k-1,\beta(j_2)}\,.
\end{align*}
Hence the result holds for the base case. In the inductive step $\ell > 2$, we have 
\begin{align*}
b_{n,\gamma(j_1,\ldots,j_\ell)} &= c_{n,(j_1,\ldots,j_\ell)} \\
&= c_{n,(j_1,\ldots,j_{\ell-1})}a_{0,j_\ell} + \sum_{k=0}^{n-1} c_{k,(j_1,\ldots,j_{\ell-1})} \overline{a}_{n-k-1,j_\ell} \qquad \text{by \Cref{eq:prop-arity2-2}}\\
&= b_{n,\gamma_{\ell-1}(j_1,\ldots,j_{\ell-1})}a_{0,j_\ell} + \sum_{k=0}^{n-1} b_{k,\gamma_{\ell-1}(j_1,\ldots,j_{\ell-1})} b_{n-k-1,\beta(j_\ell)}
\end{align*}
By inductive hypothesis on $\gamma_{\ell-1}$, we know that $b_{n,\gamma_{\ell-1}(j_1,\ldots,j_{\ell-1})}$ is partially multi-holonomic of arity 2 and we can replace $b_{n,\gamma_{\ell-1}(j_1,\ldots,j_{\ell-1})}$ by the sum given in \Cref{eq:prop-arity2-3}. This allows us to conclude that $b_{n,\gamma(j_1,\ldots,j_\ell)}$ is partially multi-holonomic of arity 2.
\item Indices $\beta(j)$, $j \in \{1, \ldots, d\}$: From \Cref{eq:proof-arity 2}, we know that for all $n \geq 1$, for all $j \in \{1, \ldots, d\}$
\begin{align*}
b_{n,\beta(j)} & = a_{n+1,j}\\
&= \sum_{\substack{\ell \in \{1, \ldots, m\}\\\bm{j} \in J_{n+1}(\ell)}} ((\va_{j_1} \otimes \ldots \otimes \va_{j_\ell}) \cdot \mu_\ell(n+1))_j\\
&= (\overline{\va}_{n-1} \mu_1(n+1))_j + \sum_{\substack{\ell \in \{2, \ldots, m\}\\\bm{j} \in J_{n+1}(\ell)}} ((\va_{j_1} \otimes \ldots \otimes \va_{j_\ell}) \cdot \mu_\ell(n+1))_j\\
&= \sum_{i = 1}^d b_{n-1,\beta(i)} \mu_1(n+1)_{i,j} + \sum_{\ell \in \{2, \ldots, m\}} \sum_{i_1,\ldots,i_\ell \in \{1, \ldots, d\}} \sum_{\bm{j} \in J_{n+1}(\ell)}  a_{j_1,i_1} \ldots a_{j_\ell,i_\ell} \mu_\ell(n+1)_{(i_1,\ldots,i_\ell),j}\\
&= \sum_{i = 1}^d b_{n-1,\beta(i)} \mu_1(n+1)_{i,j} + \sum_{\ell \in \{2, \ldots, m\}} \sum_{i_1,\ldots,i_\ell \in \{1, \ldots, d\}} c_{n,(i_1,\ldots,i_\ell)}  \mu_\ell(n+1)_{(i_1,\ldots,i_\ell),j} \qquad \text{by \Cref{eq:prop-arity2-4}}\\
&= \sum_{i = 1}^d b_{n-1,\beta(i)} \mu_1(n+1)_{i,j} + \sum_{\ell \in \{2, \ldots, m\}} \sum_{i_1,\ldots,i_\ell \in \{1, \ldots, d\}} b_{n,\gamma_\ell(i_1,\ldots,i_\ell)}  \mu_\ell(n+1)_{(i_1,\ldots,i_\ell),j}\,.
\end{align*}
Since we already proved that all $b_{n,\gamma_\ell(i_1,\ldots,i_\ell)}$ are partially multi-holonomic of arity 2, we can replace each $b_{n,\gamma_\ell(i_1,\ldots,i_\ell)}$ by their sum given in \Cref{eq:prop-arity2-3}. This allows us to conclude that $b_{n,\beta(j)}$ is partially multi-holonomic of arity 2.
\end{itemize}
As we have prove that for all $n \geq 1$, for all $j \in \{1, \ldots, D\}$, $b_{n,j}$ is partially multi-holonomic of arity 2, we conclude that $(\vb_n)_{n=0}^\infty$ is a multi-holonomic sequence of arity. Moreover, by construction, for all $n \in \N$, for all $i \in \{1, \ldots, d\}$, $a_{n,i} = b_{n,i}$. Hence we conclude.
\end{proof}

\thequivalencerecurrence*

\begin{proof}
The proof of 1 $\Rightarrow$ 2 and 2 $\Rightarrow$ 3 is direct by the definition. To prove 3 $\Rightarrow$ 2, take $(\vb_n)_{n=0}^\infty$ a generalised multi-holonomic sequence that recognises $(\va_n)_{n=0}^\infty$. By \Cref{prop:generalised to multiholonomic}, $(\va_n)_{n=0}^\infty$ is recognised by an unguarded multi-holonomic sequence. Then by \Cref{prop:unguarded to multi}, $(\va_n)_{n=0}^\infty$ is recognised by a multi-holonomic sequence. The proof of 2 $\Rightarrow$ 1 is given by \Cref{prop:multi to arity 2}.
\end{proof}

\section{Detailed proofs of Theorem~\ref{thm:MAIN1}}
\label{app:theorem1}
\propRDAimpliesCDA*

\begin{proof}
Let $f(x)$ be a RDA power series. By definition, there exists a system of differential equations of the form:
\[
y'_1 = \frac{P_1(y_1,\ldots,y_k)}{Q_1(y_1,\ldots,y_k)} \quad \ldots \quad y'_k = \frac{P_k(y_1,\ldots,y_k)}{Q_k(y_1,\ldots,y_k)}
\]
and $P_1,Q_1,\ldots,P_k,Q_k \in \K[y_1,\ldots,y_k]$ and $(y_1(x),\ldots,y_k(x))$ power series solutions of this system such that $y_1(x) = f(x)$ and for all $i \in \{1, \ldots, k\}$, $Q_i(y_1(0),\ldots,y_k(0)) \neq 0$. Since for all $i \in \{1, \ldots,k\}$, $Q_i$ is a polynomial, thus $Q_i(y_1(x),\ldots,y_k(x))$ is also power series. Moreover, as $Q_i(y_1(0),\ldots,y_k(0)) \neq 0$, we deduce that its inverse $z_i(x) = \frac{1}{Q_i(y_1(x),\ldots,y_k(x))}$ is also a power series.
However, if for all $i,j \in \{1, \ldots, k\}$, we denote by $R_{i,j}(y_1,\ldots,y_k)$ the polynomials $\frac{\partial Q_i(y_1,\ldots,y_k)}{\partial y_j}$, then we obtain that 
\begin{align*}
z'_i(x) = - \frac{1}{Q_i(y_1(x),\ldots,y_k(x))^2} \sum_{j=1}^k R_{i,j}(y_1(x),\ldots,y_k(x)) y_j'(x)\,.
\end{align*}
Therefore, $(y_1(x),\ldots,y_k(x),z_1(x),\ldots,z_k(x))$ are power series solutions of the system:
\begin{align*}
    y'_1 &= P_1(y_1,\ldots,y_k) z_1 \\
    &\vdots\\ 
    y'_k &= P_k(y_1,\ldots,y_k) z_k\\
    z'_1 &= -z_1^2 \sum_{j=1}^k R_{1,j}(y_1,\ldots,y_k)P_j(y_1,\ldots,y_k) z_j \\
    &\vdots\\
    z'_k &= -z_k^2 \sum_{j=1}^k R_{k,j}(y_1,\ldots,y_k)P_j(y_1,\ldots,y_k) z_j
\end{align*}
This conclude the proof that $y_1(x) = f(x)$ is CDA.
\end{proof}

Recall that given an automata $\A$ over $\Sigma$, $\vf_A(x) = \sum_{n=0}^\infty \sum_{\substack{t \in T_\Sigma \\\size{t} = n}} \hmu(t)x^n$ and $\widetilde{\vf}_A(x) = \sum_{n=0}^\infty \sum_{\substack{t \in T_\Sigma \\\size{t} = n}} \lambda(t)\hmu(t)x^n$.

\propautomatatorecurrence*

\begin{proof}
Notice that the only trees of size $0$ are the nullary function symbols, hence $\va_0 = \sum_{\sigma\in \Sigma_0} \mu(\sigma)$.
For all $n \geq 1$, splitting the sum into the possible shapes that trees can take, we obtain:
\begin{align*}
\va_n = \sum_{\substack{t\in T_\Sigma\\ \size{t}=n}} \hmu(t) = \sum_{\substack{k > 0\\\sigma \in \Sigma_k}}\; \sum_{\bm{n} \in J_n(k)}\; \sum_{\substack{t_1\in T_\Sigma\\ \size{t}=n_1}} \ldots \sum_{\substack{t_k\in T_\Sigma\\ \size{t}=n_k}} (\hmu(t_1) \otimes \ldots \otimes \hmu(t_k)) \cdot \mu(\sigma)(n)
\end{align*}
Using the bilinearity of the Kronecker product, we obtain:
\begin{align*}
\va_n &= \sum_{\substack{k > 0\\\sigma \in \Sigma_k}}\; \sum_{\bm{n} \in J_n(k)} \left(\sum_{\substack{t_1\in T_\Sigma\\ \size{t_1}=n_1}} \hmu(t_1) \otimes \ldots \otimes \sum_{\substack{t_k\in T_\Sigma\\ \size{t_k}=n_k}} \hmu(t_k)\right) \cdot \mu(\sigma)(n)\nonumber\\
&= \sum_{\substack{k > 0\\\sigma \in \Sigma_k}}\; \sum_{\bm{n} \in J_n(k)} (\va_{n_1} \otimes \ldots \otimes \va_{n_k}) \cdot \mu(\sigma)(n)
\end{align*}
This concludes the proof for $\va_n$.

Let us now focus on the labelled generating function. 
Let $t= \sigma(t_1,\ldots,t_\ell) \in T_\Sigma$ with $\ell \geq 1$ and $\sigma \in \Sigma_\ell$ such that size of $t$ is $n$ and $t_1,\ldots,t_\ell$ have respective sizes
$n_1,\ldots,n_\ell$. Hence $n = n_1 + \ldots + n_\ell + 1$. The number of labellings of $t$ is equal to the product of the multinomial coefficient
$\binom{n-1}{n_1,\ldots,n_\ell}$, which represents the number of $\ell$ partitions of size $n_1,\ldots,n_\ell$ in $n-1$ elements, with the respective numbers of labellings of $t_1,\ldots,t_\ell$.
It follows that
\[ \lambda(\sigma_\ell(t_1,\ldots,t_\ell))= \frac{1}{n} \lambda(t_1) \cdots \lambda(t_\ell) \, .\]
Moreover, for all $\sigma \in \Sigma_0$, $\lambda(\sigma) = 1$.

Therefore, similarly to above, we deduce that $\vb_0 = \sum_{\sigma \in \Sigma_0} \mu(\sigma)$ and for all $n \geq 1$, 
\begin{align*}
\vb_n 
&= \sum_{\substack{k > 0\\\sigma \in \Sigma_k}}\; \sum_{\bm{n} \in J_n(k)} \frac{1}{n}\left(\sum_{\substack{t_1\in T_\Sigma\\ \size{t_1}=n_1}} \lambda(t_1)\hmu(t_1) \otimes \ldots \otimes \sum_{\substack{t_k\in T_\Sigma\\ \size{t_k}=n_k}} \lambda(t_k)\hmu(t_k)\right) \cdot \mu(\sigma)(n)\\
&= \frac{1}{n}\sum_{\substack{k > 0\\\sigma \in \Sigma_k}}\; \sum_{\bm{n} \in J_n(k)} (\vb_{n_1} \otimes \ldots \otimes \vb_{n_k}) \cdot \mu(\sigma)(n)\qedhere
\end{align*}
\end{proof}


\begin{proposition}
\label{prop:multi to automaton}
Let $(\va_n)_{n=0}^\infty$ be a multi-holonomic sequence of order $d$ and arity $m$ defined by the matrices $\mu_1, \ldots, \mu_m$. Let $\Sigma = \{\sigma_0, \ldots, \sigma_m\}$ such that $\sigma_0, \ldots, \sigma_m$ have arity $0$, \ldots, $m$ respectively. 
\begin{itemize}
\item Let $\A = (d,\mu)$ be the differential tree automaton over $\Sigma$ such that $\mu(\sigma_0) = \va_0$ and for all $\ell \in \{1, \ldots, m\}$, $\mu(\sigma_\ell)(x) = \mu_\ell(x)$. Then we have $\vf_\A(x) = \sum_{n=0}^\infty \va_n x^n$.
\item Let $\A = (d,\mu)$ be the differential tree automaton over $\Sigma$ such that $\mu(\sigma_0) = \va_0$ and for all $\ell \in \{1, \ldots, m\}$, $\mu(\sigma_\ell)(x) = x \cdot \mu_\ell(x)$. Then we have $\widetilde{\vf}_\A(x) = \sum_{n=0}^\infty \va_n x^n$.
\end{itemize}
\end{proposition}

\begin{proof}
By definition, the sequence $(\va_n)_{n=0}^\infty$ satisfies for all $n \geq 1$, 
\[
\va_n = \sum_{\substack{\ell \in \{1, \ldots,m\}\\\bm{j} \in J_n(\ell)}} (\va_{j_1} \otimes \ldots \otimes \va_{j_\ell}) \cdot \sum_{\sigma \in \Sigma_\ell}\mu(\sigma)(n)
\]
Notice in this recurrence, the values of $\va_n$ are uniquely defined by the values of $\va_{n'}$ for $n' < n$. Thus, as by hypothesis, $\mu(\sigma_0) = \va_0$, the constant term of $\vf_\A$ is the same as $\va_0$, and by \Cref{prop:automata to recurrence}, the coefficients of $\vf_\A$ satisfy the same recurrence relation as $(\va_n)_{n=0}^\infty$, we deduce that $\vf_\A(x) = \sum_{n=0}^\infty \va_n x^n$. The same reasoning applies for the second statement of this proposition.
\end{proof}


\begin{restatable}{proposition}{propderive}
\label{prop:derive}
For all vectors of power series $\vf(x) = \sum_{n=0}^\infty \va_n x^n$, writing $0^0 = 1$, we have for all $i \in \N$, $\SDeriveN{\vf}{i}(x) = \sum_{n=0}^\infty n^i \va_n x^n$.
\end{restatable}

\begin{proof}
We prove this result by induction on $i$. The base case ($i = 0$) is trivial. In the inductive step ($i>0$), $\SDeriveN{\vf}{i-1}(x) = \sum_{n=0}^\infty n^{i-1} \va_n x^n$ by inductive hypothesis, hence 
\begin{align*}
&\SDerive{(\SDeriveN{\vf}{i-1})}(x) = x(\SDeriveN{\vf}{i-1})'(x) \\
=\ & x\sum_{n=1}^\infty n^i \va_n x^{n-1} = \sum_{n=1}^\infty n^i \va_n x^n
\end{align*} As $i>0$, we have $0^i = 0$ 
and so $\SDeriveN{\vf}{i}(x) = \sum_{n=0}^\infty n^i \va_n x^n$.
\end{proof}


\propCDAtoholonomic*

\begin{proof}
As a first step, if $P_1,\ldots,P_k$ contains the variable $x$ then we introduce a new variable $z$ with the equation $z' = 1$ and replacing every instance of $x$ in $P_1,\ldots,P_k$ by $z$. This increases the order by $1$ while simplifying the problem. We thus now consider that $P_1,\ldots,P_k \in \K[y_1,\ldots,y_k]$ and let $m$ be the maximal degree of all their monomials
    
Let $(y_1(x),\ldots,y_k(x))$ power series solution of \Cref{eq:prop:CDS implies multi} such that $y_1(x) = f(x)$.
Since $m$ is the maximal degree of all the monomials in $P_1,\ldots,P_k$, each $P_j$ has the form:
\[
P_j(y_1,\ldots,y_k) = \sum_{\ell=0}^m\; \sum_{\bm{i} \in \{1, \ldots, k\}^\ell} \alpha_{\bm{i},j}\; y_{i_1} \ldots y_{i_\ell} \, ,
\]
where the all coefficients $\alpha_{\bm{i},j}$ are in $\K$. The constant in the polynomial $P_j$ corresponds to the case when $\ell = 0$, that is represented by the coefficient $\alpha_{(),j}$. Note that the given system of equations can be transformed into:
\[
\SDerive{y_1} = xP_1(y_1,\ldots,y_k) \quad \ldots \quad
\SDerive{y_k} = xP_k(y_1,\ldots,y_k)
\]
Since $(y_1(x),\ldots,y_k(x))$ are power series, if we denote $y_i(x) = \sum_{n=0}^\infty a_{n,i}\,x^n$ for all $i \in \{1, \ldots, k\}$, then, relying on \Cref{prop:derive}, we deduce that the coefficients of the power series satisfy the following recurrence relations: For all $j \in \{1, \ldots, k\}$, for all $n\geq 1$, 
\begin{align}
a_{n,j} =& \frac{1}{n}\sum_{\ell=0}^m\; \sum_{\bm{j} \in J_n(\ell)} \sum_{\bm{i} \in \{1, \ldots, k\}^\ell}\; \alpha_{\bm{i},j}\, a_{j_1,i_1} \ldots a_{j_\ell,i_\ell}
\label{eq:CDAtoMultiHolonomic-Main}
\end{align}
Let us show that the sequence $(\va_n)_{n=0}^\infty = (\begin{bmatrix} a_{n,1} &  \ldots & a_{n,k} & 1\end{bmatrix})_{n=0}^\infty$ is multi-holonomic of degree $0$.
Define the matrices $\mu_1, \ldots, \mu_r$ as follows: for all $\forall \ell \in \{3, \ldots, m\}$, for all $i_1,\ldots,i_\ell, j \in \{1, \ldots, k+1\}$, 
\[
\begin{array}{@{}l@{}}
    (\mu_\ell)_{(i_1,\ldots,i_\ell),j}(x_0) = \left\{
        \begin{array}{@{\ }l@{\ }r}
            \frac{\alpha_{(i_1,\ldots,i_\ell),j}}{x_0} & \text{if }i_1,\ldots,i_\ell,j \leq k\\
            0 & \text{otherwise}
        \end{array}
    \right.\\[3mm]
    (\mu_\ell)_{(i_1,i_2),j}(x_0) = \left\{
        \begin{array}{lr}
            \frac{\alpha_{(i_1,i_2),j}}{x_0} & \text{if }i_1,i_2,j \leq k\\
            \frac{1}{x_0} & \text{if } i_1 = i_2 = j = k+1\\
            0 & \text{otherwise}
        \end{array}
    \right.\\[3mm]
   (\mu_1)_{(i),j}(x_0) = \left\{
        \begin{array}{lr}
            \frac{\alpha_{(i),j}}{x_0} & \text{if }i,j \leq k\\
            \frac{\alpha_{(),j}}{x_0} & \text{if } i = k+1, j \leq k\\
            0 & \text{otherwise}
        \end{array}
    \right.
\end{array}
\]

We now prove the desired property. Notice that by definition of the Kronecker product for all $\ell \in \{1, \ldots, m\}$, for all $j \in \{1, \ldots, k+1\}$, the $j$-th entry of $(\va_{j_1} \otimes \ldots \otimes \va_{j_\ell}) \cdot \mu_\ell(n)$ is the following sum:
\[
\sum_{\bm{i} \in \{1, \ldots, k+1\}^\ell} a_{j_1,i_1} \ldots a_{j_\ell,i_\ell} \cdot (\mu_\ell)_{\bm{i},j}(n)
\]
Therefore, by definition of $\mu_1,\ldots, \mu_m$, we have for $j \in \{1, \ldots,k\}$, the $j$-th entry of $(\va_{j_1} \otimes \ldots \otimes \va_{j_\ell}) \cdot \mu_\ell(n)$ is
\[
\frac{1}{n}\sum_{\bm{i} \in \{1, \ldots, k\}^\ell} a_{j_1,i_1}\ldots a_{j_\ell,i_\ell}\, \alpha_{\bm{i},j} 
\]
when $\ell \geq 2$ and 
\[
\frac{\alpha_{(),j}}{n} + \frac{1}{n}\sum_{i = 1}^k a_{j_1,i}\alpha_{(i),j}
\] 
when $\ell = 1$.
Furthermore, the $k+1$-th entry of $(\va_{j_1} \otimes \ldots \otimes \va_{j_\ell}) \cdot \mu_\ell(n) = 1$ when $\ell = 2$ and $0$ otherwise. Note that by convention, when $\ell = 0$, the two innermost sums of \Cref{eq:CDAtoMultiHolonomic-Main} are equal to $\alpha_{(),j}$.
Therefore, combining these equations with \Cref{eq:CDAtoMultiHolonomic-Main} gives us that for all $j \in \{1, \ldots, k\}$, 
\begin{align*}
a_{n,j} =& \frac{1}{n}\sum_{\ell = 0}^m\ \sum_{\bm{j} \in J_n(\ell)}\ \sum_{\bm{i} \in \{1, \ldots, k\}^\ell}\; 
 \alpha_{\tilde{i},j}\, a_{j_1,i_1}\ldots a_{j_\ell,i_\ell} \\
=& \sum_{\ell = 1}^m\ \sum_{\bm{j} \in J_n(\ell)}\; 
((\va_{n_1} \otimes \ldots \otimes \va_{n_\ell}) \cdot \mu_\ell(n))_j
\end{align*}
Furthermore, 
\begin{align*}
    &\sum_{\ell=1}^m\ \sum_{\bm{j} \in J_n(\ell)}\ ((\va_{n_1} \otimes \ldots \otimes \va_{n_\ell}) \cdot \mu_\ell(n))_{k+1}\\
    =\ &\frac{1}{n} \sum_{\bm{j} \in J_n(2)} 1\\
    =\ & 1 = \va_{n,k+1}
\end{align*}
This allows us to deduce that $(\va_n)_{n=0}^\infty = (\begin{bmatrix} a_{n,1} &  \ldots & a_{n,k} & 1\end{bmatrix})_{n=0}^\infty$ is multi-holonomic of with coefficients in matrices of the form $\frac{a}{x_0}$ with $a \in \K$. We conclude by applying \Cref{prop:multi to automaton}.
\end{proof}

Let $(\va_n)_{n=0}^\infty$ be a multi-holonomic of arity~$m$ and degree~$r$,   defined by  matrices $\mu_1,\ldots,\mu_m$ whom, without loss of generality, are assumed to share a common denominator $A(x) \in \K[x]$. Furthermore, as $(\va_n)_{n=0}^\infty$ has degree $r$, assume that
\[
A(x) = \sum_{i=0}^r c_i x^i
\]
for some $c_0,\ldots,c_r \in r$ and for all $\ell \in \{1, \ldots, m\}$, \[\mu_\ell(x) = \sum_{i=0}^r x^i \cdot \mu_{\ell,i}\]
for some $\mu_{\ell,i} \in \K^{d^\ell\times d}$. Recall that $J_n(\ell)$ is the set of tuples $(j_1,\ldots,j_\ell) \in \N^{\ell}$ such that $j_1 + \ldots + j_\ell = n-1$.


\propsystemofequations*

\begin{proof}
We perform the following computations:

\begin{align*}
&\sum_{i=0}^r c_i \SDeriveN{\vf}{i}(x) - c_0\va_0\\
=\ & c_0 \sum_{n=0}^\infty \va_n x^n - c_0\va_0 + \sum_{i=1}^u c_i \sum_{n=0}^\infty n^i \va_n x^n \qquad\text{by Prop.~\ref{prop:derive}}\\
=\ &c_0 \sum_{n=1}^\infty \va_n x^n + \sum_{i=1}^r c_i \sum_{n=1}^\infty n^i \va_n x^n \qquad\text{as $0^i = 0$ for $i\geq1$}\\
=\ &\sum_{i=0}^r c_i  \sum_{n=0}^\infty (n+1)^i \va_{n+1} x^{n+1}\\
=\ &\sum_{n=0}^\infty \sum_{i=0}^r c_i (n+1)^i \va_{n+1} x^{n+1}\\
=\ &\sum_{n=0}^\infty A(n+1) \va_{n+1} x^{n+1}\\
=\ &\sum_{n=0}^\infty \frac{A(n+1)}{A(n+1)} \sum_{\ell=1}^m\; \sum_{\bm{n} \in J_{n+1}(\ell)} (\va_{n_1} \otimes \ldots \otimes \va_{n_\ell}) \cdot \sum_{i=0}^r (n+1)^i \cdot \mu_{\ell,i} x^{n+1}\\
=\ &\sum_{n=0}^\infty \sum_{\ell=1}^m\; \sum_{\bm{n} \in J_{n+1}(\ell)} (\va_{n_1} \otimes \ldots \otimes \va_{n_\ell}) \cdot \sum_{i=0}^r (n_1+\ldots+n_\ell+1)^i \cdot \mu_{\ell,i} x^{n+1}\\
=\ & \sum_{n=0}^\infty \sum_{\ell=1}^m\;  \sum_{\bm{n} \in J_{n+1}(\ell)} (\va_{n_1} \otimes \ldots \otimes \va_{n_\ell}) \cdot \sum_{i=0}^r \sum_{\substack{\bm{j} \in J_{i+1}(\ell+1)}} \binom{i}{\bm{j}} n_1^{j_1}\ldots n_\ell^{j_\ell} \mu_{\ell,i}x^{n+1}\\
=\ & \sum_{\ell=1}^m\; \sum_{i=0}^r \sum_{\substack{\bm{j} \in J_{i+1}(\ell+1)}} \binom{i}{\bm{j}} \left( \sum_{n_1=0}^\infty n_1^{j_1}\va_{n_1}x^{n_1} \otimes \ldots \otimes \sum_{n_\ell=0}^\infty n_\ell^{j_\ell}\va_{n_\ell}x^{n_\ell} \right) \cdot \mu_{\ell,i}\ x\\
=\ & x \sum_{\ell=1}^m\; \sum_{i=0}^r \sum_{\substack{\bm{j} \in J_{i+1}(\ell+1)}} \binom{i}{\bm{j}}\; (\SDeriveN{\vf}{j_1}(x) \otimes \ldots \otimes \SDeriveN{\vf}{j_\ell}(x)) \cdot \mu_{\ell,i}
\end{align*}
This concludes the proof.
\end{proof}


\propHolonomicOneToRDA*

\begin{proof}
Assume that the sequence $(\va_n)_{n=0}^\infty$ has order $d$. Without loss of generality, we know that sequence $(\va_n)_{n=0}^\infty$ satisfies for all $n \geq 1$,
\[
\boldsymbol a_n = \frac{1}{n} \; \sum_{\ell=1}^m\; \sum_{\bm{j} \in J_n(\ell)} (\boldsymbol a_{j_1} \otimes \cdots \otimes \boldsymbol a_{j_k}) \cdot \mu_k(n)
\] 
where $\mu_k \in \K[x_0]^{d^k\times k}$ is of degree at most 1 for all $\ell \in \{1, \ldots, m\}$.
By \Cref{prop:system equations}, we know that $\vf(x)$ is solution of a system of differential equations of the form:
\[
\left\{
\begin{array}{l}
\SDerive{f_1} = x( Q_1(f_1,\ldots,f_d) + \sum_{i = 1}^d \SDerive{f_i }\cdot P_{1,i}(f_1,\ldots,f_d) )\\
\ldots\\
\SDerive{f_d} = x( Q_d(f_1,\ldots,f_d) + \sum_{i = 1}^d \SDerive{f_d}\cdot P_{d,i}(f_1,\ldots,f_d) )\\
\end{array}
\right.
\]
where all $Q_j$s and $P_{j,i}$s are polynomials in $\K[f_1,\ldots,f_d]$. We think of this system as a linear system with unknown $\SDerive{f_1}, \ldots, \SDerive{f_d}$ and with coefficients in the field $\K(f_1,\ldots,f_d)$. Consider the following matrix $M$:
\[
M=
\begin{bmatrix}
(1-xP_{1,1}) & -x P_{1,2} & \ldots & -x P_{1,d-1} & -x P_{1,d} \\
-x P_{2,1} & (1-xP_{2,2}) & \ldots & -x P_{2,d-1} & -x P_{2,d} \\
\vdots & \vdots & \vdots & \vdots & \vdots \\
-x P_{d,1} & \ldots  & \ldots & -x P_{d,d-1} & (1-x P_{d,d}) \\
\end{bmatrix}
\]
Thus, as $\det(M) = \sum_{\sigma \in S_d} \mathrm{sgn} \prod_{i=1}^d M_{i,\sigma(i)}$, we deduce that for all permutation $\sigma \in S_d$, if $\sigma$ is not the identity then $\prod_{i=1}^d M_{i,\sigma(i)} = x R$ for some polynomials $R \in \K[x,f_1,\ldots,f_d]$. Moreover, $\prod_{i=1}^d M_{i,i} = 1 + x R'$ for some polynomials $R' \in \K[x,f_1,\ldots,f_d]$. As such, $\det(M) = 1 + x(R' + R'') \in \K[f_1,\ldots,f_d]$ and $\det(0,f_1(0),\ldots,f_d(0)) \neq 0$. Therefore, we can compute $M^{-1} = \frac{1}{\det(M)}\mathrm{adj}(M)$ with $\mathrm{adj}(M) \in \K[x,f_1,\ldots,f_d]$ and we deduce that 
\[
\begin{bmatrix}
\SDerive{f_1}\\
\vdots\\
\SDerive{f_d}\\
\end{bmatrix}
= 
M^{-1}
\begin{bmatrix}
xQ_1\\
\vdots\\
xQ_d\\
\end{bmatrix}
= x \begin{bmatrix}
    U_1\\
    \vdots\\
    U_d\\
    \end{bmatrix}
\]
with $U_1,\ldots,U_d \in \K(x,f_1,\ldots,f_d)$ that are defined at $(0,f_1(0),\ldots,f_d(0))$. As $\SDerive{f}(x) = xf'(x)$ for all functions $f$, we deduce that $(x,f_1(x),\ldots,f_d(x))$ satisfies the system of differential equations:
\[
\left\{
\begin{array}{l}
t' = 1\\
f'_1 = U_1(t,f_1,\ldots,f_d)\\
\vdots \\
f'_d = U_d(t,f_1,\ldots,f_d)
\end{array}
\right.
\]
This concludes the proof that all power series in $\vf(x)$ are RDA.
\end{proof}

\thmainone*

\begin{proof}
The proof of \ref{enum:thm:MAIN1-RDA} $\Rightarrow$ \ref{enum:thm:MAIN1-ordinary} is given by \Cref{prop:RDA implies CDA,prop:CDA implies multi-holonomic degree 0}. The proof of \ref{enum:thm:MAIN1-RDA} $\Rightarrow$ \ref{enum:thm:MAIN1-labelled} is given by once again applying \Cref{prop:RDA implies CDA,prop:CDA implies multi-holonomic degree 0} to show that $f(x)$ is the ordinary generating function of a differential tree automaton in which the rational weight of every transition has the form $\frac{a}{x}$ for some $a \in \K$. We conclude by applying by applying \Cref{prop:automata to recurrence} followed by \Cref{prop:multi to automaton}. 
Similarly, the proof of \ref{enum:thm:MAIN1-labelled} $\Rightarrow$ \ref{enum:thm:MAIN1-ordinary} is given by applying \Cref{prop:automata to recurrence} followed by \Cref{prop:multi to automaton}. Finally, the proof of \ref{enum:thm:MAIN1-ordinary} $\Rightarrow$ \ref{enum:thm:MAIN1-RDA} is given by \Cref{prop:multi-holonomic degree 1 to RDA}.

\end{proof}

\section{Detailed proofs of Theorem~\ref{thm:MAIN2}}
\label{app:theorem2}
We detail in this section the proofs of the propositions needed to prove \Cref{thm:MAIN2}. Note that the proof of \Cref{prop:system equations} was already given in \Cref{app:theorem1}.


\propunicity*

\begin{proof}
We show by a simple inductive proof on $n$ that every $\overline{\va}_{n}$ are uniquely defined by $\overline{\va}_0$. The base case is trivial since $\overline{\va}_0$ is given. In the inductive step ($n>0$), we know from our inductive hypothesis that all $\overline{\va}_0,  \ldots, \overline{\va}_{n-1}$ are uniquely defined. 

Note that the $n^{th}$ coefficient of $\sum_{i=0}^r c_i \SDeriveN{\overline{\vf}}{i}(x) - c_0\overline{\va}_0 $ is $\sum_{i=0}^r c_i\, n^i\, \overline{\va}_n = A(n) \overline{\va}_{n}$ (since $n>0$). On the other hand, the $n^{th}$ coefficient of the right hand side of the \Cref{eq:system1} is the $(n-1)^{th}$ coefficient of
\[
    \sum_{\substack{\ell \in \{1, \ldots,m\}\\i \in \{0, \ldots,r\}\\\bm{j} \in J_{i+1}(\ell+1)}} \binom{i}{\bm{j}}  (\SDeriveN{\vf}{j_1} \otimes \ldots \otimes \SDeriveN{\vf}{j_\ell}) \cdot \mu_{\ell,i}
\]
which is determined uniquely by $\overline{\va}_0, \ldots, \overline{\va}_{n-1}$. Therefore, since $A(n)$ is never zero for all positive integer $n$ by hypothesis, $\overline{\va}_{n}$ is uniquely determined.

Finally, if $A(0) \neq 0$ then $c_0 \neq 0$. Notice by \Cref{prop:derive} that $\SDeriveN{\overline{\vf}}{i}(0) = 0$ for $i > 0$ and $\SDeriveN{\overline{\vf}}{0}(0) = \overline{\va}_0$. Moreover, as the right hand side of the \Cref{eq:system1} is guarded by $x$, we obtain that $c_0\,\overline{\va}_0 - c_0\,\va_0 = 0$ which implies $\overline{\va}_0 = \va_0$.
\end{proof}


\lemRecursion*

\begin{proof}
We follow ideas of~\cite{denef1984power}, which are in turn based on Hurwitz~\cite{Hurwitz1932}.

Let $k \in \mathbb N$ be the order of vanishing of the power series $\frac{\partial P}{\partial y^{(d)}} (x,f,\ldots,f^{(d)})$, that is, such that 
\begin{equation}
\frac{\partial P}{\partial y^{(d)}}\left(x,f,\ldots,f^{(d)}\right)=x^k \cdot \sum_{i=0}^\infty c_i x^i
\label{eq:vanishing}
\end{equation} 
with $c_0\neq 0$.

Differentiating $P$ to order $2k+2$ we have
 \begin{gather*}
        P^{(2k+2)} \; = \;  y^{(d+2k+2)}f_d + y^{(d+2k+1)}f_{d+1} + \cdots + y^{(d+k+2)}f_{d+k}+f_{d+k+1}
    \end{gather*}
    where  
    $f_{d},\ldots,f_{d+k+1} \in \K[x,y^{(\infty)}]$ are such that 
    $f_j$ has order at most $j$ and $f_d = \frac{\partial P}{\partial y^{(d)}}$. Continuing to differentiate, by Leibnitz's rule for all $n \in \mathbb N$ we have
    \begin{equation}
        P^{(2k+2+n)} = y^{(d+2k+2+n)}S_0(n) + y^{(d+2k+1+n)} S_1(n) + \cdots + y^{(d+k+2+n)}S_k(n)+ h_{d+k+1+n}
        \label{eq:Denef_eq}
    \end{equation}
where for all $j \in \{0,\ldots,k\}$ we have
\[
S_j(n):= \left(f_{d+j}+nf'_{d+j-1}+\cdots + \binom{n}{j}f_d^{(j)}\right) \, .
\]
and $h_{d+k+1+n} \in \K[x,y^{(\infty)}]$ has order at most $d+k+1+n$. We use the convention here that $\binom{n}{j} = 0$ when $n < j$.

Notice that $S_k(n)(0,f(0),\ldots,f^{(d+k)}(0))$ as a polynomial in $n$ has degree $k$ and the $k$-th coefficient is given by $\binom{n}{k}f_d^{(k)}(0,f(0),\ldots,f^{(d+k)}(0))$. Thus, since $f_d = \frac{\partial P}{\partial y^{(d)}}$ and by \Cref{eq:vanishing}, we deduce that $f_d^{(k)}(0,f(0),\ldots,f^{(d+k)}(0)) = k! c_0 \neq 0$. Thus, $S_k(n)(0,f(0),\ldots,f^{(d+k)}(0))$ as a polynomial in $n$ is not identically zero.

It implies that there exists a least element $r\in \{0,\ldots,k\}$ such that $S_r(n)(0,f(0),\cdots,f^{(d+r)}(0))$ is not identically zero as a polynomial  in $n$.  For this choice of $r$, setting $m:=2k+2$ and $s:=d+2k+1-r$, define:
\begin{align*}
    A(n) & := -S_r(n)(0,f(0),\ldots,f^{(d+r)}(0))\\
    Q_n & := P^{(m+n)} \mod y^{(s+n+1)}
\end{align*}
As $A(n)$ is not identically $0$, there exists $N \in \N$ such that for all $n \geq N$, $A(n) \neq 0$.
Since for all $j \in \{r+1, \ldots, k\}$, the order of $S_j(n)$ is at most $d+j$. Moreover, $h_{d+k+1+n}$ has order at most $d+k+1+n$. Therefore, by \Cref{eq:Denef_eq}, we deduce that 
\[
Q_n = \sum_{j=r+1}^k y^{(d+2k+2+n-j)}S_j(n) + h_{d+k+1+n}
\]
Finally, recall that $P(x,f(x),\ldots,f^{(d)}(x)) = 0$ and so $P^{(n+m)}(x,f(x),\ldots,f^{(d+m+n)}(x)) = 0$, and the polynomials $S_j(n)(0,f(0),\ldots,f^{(d+j)}(0))$ are identically zero for $j < r$. Thus, instantiating \Cref{eq:Denef_eq} at $0$, we obtain:
\[
Q_n(0,f(0),\ldots,f^{(n+s)}(0)) = f^{(n+s+1)}(0) \cdot A(n) \, .
\]
As for all $n \geq N$, $A(n) \neq 0$, the proof is complete.
\end{proof}


\begin{restatable}{proposition}{propMonomial}
Let 
    $M := y^{(i_1)} \cdots y^{(i_\ell)} \in  \K[y^{(\infty)}]$ be a monomial of order $d$ in $y$.  Then for all $n\in \mathbb N$
we have
\[ M^{(n)} = \sum_{\bm{j} \in J_{n+1}(\ell)} \binom{n}{\bm{j}} \, y^{(i_1+j_1)} \cdots y^{(i_\ell+j_\ell)} \,,  \]
\label{prop:monomial}
\end{restatable}
    
\begin{proof}
We prove the result by induction on $n$.
The base case is immediate.  The induction step is as follows.
We have
\[
\begin{array}{rcl}
    M^{(n+1)}  & = & 
    \displaystyle\Bigg( \, \sum_{\bm{j} \in J_{n+1}(\ell)} \binom{n}{\bm{j}} \, y^{(i_1+j_1)} \cdots y^{(i_\ell+j_\ell)}\Bigg)'
    \quad\text{(Induction hypothesis)}\\
    &=& \displaystyle \sum_{k=1}^\ell \,
    \sum_{\bm{j} \in J_{n+1}(\ell)} \binom{n}{\bm{j}} \, y^{(i_1+j_1)} \cdots
    y^{(i_k+j_k+1)} \cdots y^{(i_\ell+j_\ell)} \quad\text{(Leibnitz rule)}\\
    &=& \displaystyle \sum_{k=1}^\ell \,
    \sum_{\substack{j_1+\cdots+j_\ell=n+1\\j_1,\ldots,j_\ell \in \N, j_k > 0}} \binom{n}{j_1,\dots,j_k-1,\ldots,j_\ell} \, y^{(i_1+j_1)} \cdots y^{(i_\ell+j_\ell)}
    \end{array}
\] 
Taking the convention that $\binom{n}{j_1,\ldots,j_\ell} = 0$ when $j_i < 0$ for some $i$, we obtain:
\[
\begin{array}{rcl}
    M^{(n+1)} 
    &=& \displaystyle \sum_{k=1}^\ell \,
    \sum_{\bm{j} \in J_{n+2}(\ell)} \binom{n}{j_1,\dots,j_k-1,\ldots,j_\ell} \, y^{(i_1+j_1)} \cdots y^{(i_\ell+j_\ell)}\\
    &=& \displaystyle 
    \sum_{\bm{j} \in J_{n+2}(\ell)} \sum_{k=1}^\ell \, \binom{n}{j_1,\dots,j_k-1,\ldots,j_\ell} \, y^{(i_1+j_1)}\, \cdots y^{(i_\ell+j_\ell)}\\
    &=& \displaystyle 
    \sum_{\bm{j} \in J_{n+2}(\ell)} \binom{n+1}{\bm{j}} \, y^{(i_1+j_1)}\, \cdots y^{(i_\ell+j_\ell)}\quad\text{(Pascal's identity)}
\end{array}
\] 
This completes the proof.
\end{proof}


\begin{restatable}{proposition}{thDalgebraicToHolonomic}
\label{th:D-algebraic to multiholonomic}
Let $f \in \K[\![x]\!]$ be a differential algebraic power series. There exists a multi-holonomic sequence $(\va_n)_{n=0}^\infty$ such that $f(x) = \sum_{n=0} \va_{n,1} x^n$.
\end{restatable}

\begin{proof}
Let $f(x) = \sum_{n=0}^\infty u_n x^n$ be a differentially algebraic power series. Applying \Cref{prop:vanishing order,lem:RECUR1}, we deduce that there exist $P \in \K[x][y,\ldots,y^{(d)}]$ and $m,s,N \in \N$ and $A \in \K[x]$ such that for all $n \geq N$, 
\[
    f^{(n+s+1)}(0)  = \frac{1}{A(n)} \, Q_n(0,f(0),\ldots,f^{(n+s)}(0))
\]
where $Q_n = P^{(m+n)} \bmod y^{(n+s+1)}$.

By definition, $f^{k}(0) = k!\, u_k$ for all $k \in \N$. Let us denote by $I$ the set of tuples $(\alpha,i_1,\ldots,i_\ell)$ with $\alpha \in \K$
and $i_1 \leq \cdots \leq i_\ell \in \mathbb N$ such that $\alpha y^{(i_1)} \cdots y^{(i_\ell)}$ is a monomial of $P$. In other words $P = \sum_{(\alpha,i_1,\ldots,i_\ell) \in I} \alpha y^{(i_1)} \cdots y^{(i_\ell)}$. By \Cref{prop:monomial}, we deduce that:
\[
P^{(m+n)} = \sum_{(\alpha,i_1,\ldots,i_\ell) \in I}\; \sum_{\substack{j_1+\cdots+j_\ell=m+n\\j_1,\ldots,j_\ell \in \N}} \alpha  \binom{m+n}{j_1,\dots,j_\ell} \, y^{(i_1+j_1)} \cdots y^{(i_\ell+j_\ell)}
\]
As $Q_n = P^{(m+n)} \bmod y^{(n+s+1)}$, we obtain
\[
Q_n = \sum_{(\alpha,i_1,\ldots,i_\ell) \in I}\; \sum_{\substack{j_1+\cdots+j_\ell=m+n\\i_1+j_1, \ldots, i_\ell+j_\ell \leq n+s}} \alpha \binom{m+n}{j_1,\dots,j_\ell} \, y^{(i_1+j_1)} \cdots y^{(i_\ell+j_\ell)}
\]
Therefore, we obtain for all $n \geq N$, 
\[
u_{n+s+1} = \frac{(n+m)!}{(n+s+1)!\,A(n)} \sum_{(\alpha,i_1,\ldots,i_\ell) \in I} 
\sum_{\substack{j_1+\cdots+j_\ell=n+m\\i_1+j_1,\ldots,i_\ell+j_\ell\leq n+s}} 
\alpha \, (j_1+1)_{i_1} \cdots (j_\ell+1)_{i_\ell} \, u_{i_1+j_1} \cdots u_{i_\ell+j_\ell}
\]
Define the sequence $(\va_n)_{n=0}^\infty$ of order $N+s+m+1$ such that for all $n \in \N$, 
\[
\va_n = \begin{bmatrix} u_n & \ldots & u_{n+N+s+m}\end{bmatrix}\,.
\] 
Define the polynomials $B(n)$ and $C(n)$ such that if $s \geq m - 1$ then $B(n) = (n+N+s+m+1)\ldots(n+N+2m-1)A(n+N+m)$ and $C(n) = 1$ else $B(n) = A(n+N+m)$ and $C(n) = (n+N+2m)\ldots(n+N+s+m)$. Finally, for all $\boldsymbol{m} = (\alpha,i_1,\ldots,i_\ell) \in I$, define the polynomial $R_{\boldsymbol{m}}(x_1,\ldots,x_\ell) = \alpha(x_1+N+s+m-i_1+1)_{i_1}\ldots (x_\ell+N+s+m-i_\ell+1)_{i_\ell}$ and $M_{\boldsymbol{m}} = \ell(N+s+m) - (m + N + i_1 \ldots + i_\ell)$. 

Therefore, for all $n \geq 1$, 
\[
\va_{n+1,N+s+m+1} = \frac{C(n)}{B(n)} \sum_{\boldsymbol{m} = (\alpha,i_1,\ldots,i_\ell) \in I} 
\sum_{\substack{j'_1+\cdots+j'_\ell=n-M_{\boldsymbol{m}}\\j'_1,\ldots,j'_\ell \leq n-m}} 
R_{\boldsymbol{m}}(j'_1,\ldots,j'_\ell) \, \va_{j'_1,N+s+1} \cdots \va_{j'_\ell,N+s+1}
\]
Thus, we build set $I'$ by adding for each $\boldsymbol{m} = (\alpha,i_1,\ldots,i_\ell)$ a tuple $\overline{\boldsymbol{m}} = (-M_{\boldsymbol{m}},m,\ldots,m)$ and the transition matrix $\overline{\mu}_{\overline{\boldsymbol{m}}}(x_0,\ldots,x_\ell)$ such that for all $\boldsymbol{u}_1,\ldots, \boldsymbol{u}_\ell \in \K^{N+s+m+1}$, 
\begin{align*}
(\boldsymbol{u}_1 \otimes \ldots \otimes \boldsymbol{u}_\ell) \overline{\mu}_{\overline{\boldsymbol{m}}}(x_0,\ldots,x_\ell)_{N+s+m+1} &= \frac{C(x_0-1)R_{\boldsymbol{m}}(x_1,\ldots,x_\ell)}{B(x_0-1)}\boldsymbol{u}_{1,N+s+1} \ldots \boldsymbol{u}_{\ell,N+s+1}\\
(\boldsymbol{u}_1 \otimes \ldots \otimes \boldsymbol{u}_\ell) \overline{\mu}_{\overline{\boldsymbol{m}}}(x_0,\ldots,x_\ell)_{i} &= 0 \qquad \forall i \in \{1, \ldots, N+s+m\}
\end{align*}
To complete the generalised sequence, we add a final element $\overline{\boldsymbol{m}} = (0,0)$ and the associated transition matrix $\overline{\mu}_{\overline{\boldsymbol{m}}}(x_0,x_1)$ such that for all $\boldsymbol{u} \in \K^{N+s+m+1}$, 
\begin{align*}
\boldsymbol{u} \overline{\mu}_{\overline{\boldsymbol{m}}}(x_0,x_1)_{N+s+m+1} &= \boldsymbol{0}^d\\
\boldsymbol{u} \overline{\mu}_{\overline{\boldsymbol{m}}}(x_0,x_1)_i &= \boldsymbol{u}_{i+1} \qquad \forall i \in \{0,\ldots,M-1\}
\end{align*}
By construction, we have shown that $(\va_n)_{n=0}^\infty$ is a generalised multi-holonomic sequence with only negative shifts and such that for all $n \in \N$, $\va_{n,1} = u_n$. 
We conclude by applying \Cref{thm:MAIN3}.
\end{proof}


\thmaintwo*

\begin{proof}
From \Cref{prop:automata to recurrence}, we know that the ordinary generating function of any differential tree automaton $\A$ is the first component of the vector of power series $\vf_\A(x) = \sum_{n=0}^\infty \va_n x^n$. The latter induces a multi-holonomic sequence $(\va_n)_{n=0}^\infty$. Hence, from \Cref{thm:MAIN3}, we deduce that  $(\va_n)_{n=0}^\infty$ is recognised by a multi-holonomic sequence of arity 2 which concludes the proof of \ref{enum:thm:MAIN2-ordinary} $\Rightarrow$ \ref{enum:thm:MAIN2-recurrence}. 

Conversely, consider the sequence $(\vb_n)_{n=0}^\infty$ of elements in $\K^d$ satisfying for all $n \in \N$ the recurrence 
\begin{equation}
\label{eq:th3:recurrence-detailed}
\vb_{n} = \vb_{n-1} \cdot P(n) + \sum_{k=0}^{n-1} (\vb_k \otimes \vb_{n-k-1})\cdot Q(n)
\end{equation}
with $P$ and $Q$ having entries defined on positive integers. We define the alphabet $\Sigma = \{ \sigma_0, \sigma_1, \sigma_2 \}$ with $\sigma_0$, $\sigma_1$ and $\sigma_2$ having respectively arity $0$, $1$ and $2$. Finally, we define the differential tree automaton $\A = (d,\mu)$ over $\Sigma$ such that $\mu(\sigma_0) = \vb_0$, and
\[
    \mu(\sigma_1)(x) = P(x)\quad \text{and}\quad \mu(\sigma_2)(x) = Q(x)\,.
\]
Notice that $P(x) \in \KK(x)^{d\times d}$ and $Q(x) \in \KK(x)^{d^2\times d}$. From \Cref{prop:automata to recurrence}, writing $\vf_\A(x)$ as $\sum_{n=0}^\infty \va_n x^n$, we directly obtain that the sequence $(\va_n)_{n=0}^\infty$ satisfies the same recurrence relation as $(\vb_n)_{n=0}^\infty$, that is \Cref{eq:th3:recurrence-detailed}. Moreover, it also has the same initial value, that is $\va_0 = \vb_0$. These entail $(\vb_n)_{n=0}^\infty = (\va_n)_{n=0}^\infty$ which concludes the proof of \Cref{enum:thm:MAIN2-recurrence} $\Rightarrow$ \Cref{enum:thm:MAIN2-ordinary}.

The proof of \Cref{enum:thm:MAIN2-recurrence} $\Rightarrow$ \Cref{enum:thm:MAIN2-D-algebraic} is given by \Cref{th:differentially algebraic}. Finally, the proof of \Cref{enum:thm:MAIN2-D-algebraic} $\Rightarrow$ \Cref{enum:thm:MAIN2-recurrence} is given by \Cref{th:D-algebraic to multiholonomic}.
\end{proof}

\section{Algebraic operations on differential tree automata}
\label{sec:operations}

This appendix is dedicated to the proof of the following theorem.


\begin{theorem}[Closure properties]

\label{prop:closure}
The class of formal tree series recognisable by differential tree automata is closed under addition, scalar  multiplication  and product.

The class of generating functions of  differential tree automata is closed under addition, scalar  multiplication, product, derivative, integral, inverse, forward shift and backward shift.
\end{theorem}

For the class of formal tree series recognisable by  differential tree automata, the closure properties are given by the following propositions:
\begin{itemize}
\item addition: \Cref{lem:closure sem addition}
\item scalar  multiplication: \Cref{lem:closure sem scalar product}
\item product: \Cref{lem:closure sem product}
\end{itemize}

For the class of generating functions of differential tree automata, the closure properties are given by the following propositions:
\begin{itemize}
\item addition and scalar  multiplication: \Cref{lem:closure gen addition scalar}
\item product: \Cref{cor:closure gen product complete}
\item derivation and integral: \Cref{cor:closure gen derive and integral}
\item inverse: \Cref{prop:closure gen inverse}
\item backward shift: \Cref{lem:closure gen product}
\item forward shift: \Cref{prop:closure gen devise}
\end{itemize}


\subsection{The class of formal tree series recognisable by  differential tree automata}

In the proofs of this section, we denote by $\bm{1} = (1, \ldots, 1)$ the vector of 1. 


\begin{proposition}
\label{lem:vector e1}
Let $\A = (d,\mu)$ be a differential tree automaton over $\Sigma$. Let $\vbeta(x) \in \Q(x)^{d\times 1}$ defined on all non-negative integers. There exists an automaton $\A'$ over $\Sigma$ such that for all $t \in T_\Sigma$, $\sem{\A'}{t} = \hmu(t)\vbeta(\size{t})$.
\end{proposition}

\begin{proof}
We build the   differential tree automaton by incrementing the dimension of $\A$ by 1 and by building the weight function $\mu'$ such that for all $t \in T_\Sigma$, $\hmu'(t) = \begin{bmatrix} \hmu(t)\vbeta(\size{t}) & \hmu(t) \end{bmatrix}$. Formally, for all $a \in \Sigma_0$, we define $\mu'(\sigma) = \begin{bmatrix} \mu(\sigma)\beta(0) & \mu(\sigma)\end{bmatrix}$; and for all $\sigma \in \Sigma_k$ with $k>0$, $\mu'(\sigma) \in \K[x]^{(d+1)^k \times (d+1)}$ and, denoting $M(x) = \mu(\sigma)(x) \vbeta(x)$, we have:
\begin{itemize}
\item for all $j \in \{2, \ldots, d+1\}$, for all $\boldsymbol{i} = (i_1,\ldots,i_k) \in \{1,\ldots,d+1\}^k$, 
\[
\mu'(\sigma)_{\boldsymbol{i},j} = 
\left\{\begin{array}{lr}
    \mu(\sigma)_{\bm{i} - \bm{1},j-1} &  \text{if }\forall \ell,  i_\ell \in \{2,\ldots,d+1\}\\
    0 & \text{otherwise }
\end{array}
\right.
\]
\item for all $ \boldsymbol{i}=(i_1,\ldots,i_k) \in \{2,\ldots,d+1\}^k$, 
\[
\mu'(\sigma)_{\boldsymbol{i},1} = 
\left\{\begin{array}{lr}
    M_{\bm{i} - \bm{1},1} & \text{if }\forall \ell, i_\ell \in \{2,\ldots,d+1\}\\
    0 & \text{otherwise}
\end{array}
\right.
\]
\end{itemize}

\medskip

Take $\A' = (d',\mu')$ with $d' = d+1$, we prove by induction on the structure of trees that for all $t \in T_\Sigma$, $\hmu'(t) = \begin{bmatrix} \hmu(t)\vbeta(\size{t}) & \hmu(t) \end{bmatrix}$. In the base case, $t$ is a leaf $\sigma \in \Sigma_0$, meaning that $\hmu'(\sigma) = \mu'(\sigma) = \begin{bmatrix} \mu(\sigma)\beta(0) & \mu(\sigma)\end{bmatrix} = \begin{bmatrix} \hmu(\sigma)\beta(0) & \hmu(\sigma)\end{bmatrix}$. 

\medskip

In the inductive step, $t = \sigma(t_1,\ldots,t_k)$ for some $\sigma \in \Sigma_k$ and $t_1,\ldots,t_k \in T_\Sigma$. Applying our inductive hypothesis on $t_1,\ldots,t_k$, we have for all $i \in \{1, \ldots, k\}$, $\hmu'(t_i) = \begin{bmatrix} \hmu(t_i)\vbeta(\size{t_i}) & \hmu(t_i) \end{bmatrix}$. Hence, 
\[
\hmu'(t) = (\begin{bmatrix} \hmu(t_1)\vbeta(\size{t_1}) & \hmu(t_1) \end{bmatrix} \otimes \ldots \otimes \begin{bmatrix} \hmu(t_k)\vbeta(\size{t_k}) & \hmu(t_k) \end{bmatrix}) \mu'(\sigma)(\size{t})
\]
By construction, for all $j \in \{1, \ldots, d+1\}$, for all $ \boldsymbol{i}=(i_1,\ldots,i_k) \in \{1,\ldots,d+1\}^k$, if $i_\ell = 1$ for some $\ell$ then $\mu'(\sigma)_{\boldsymbol{i},j}(x) = 0$. Thus, for all $j \in \{2, \ldots, d+1\}$, \begin{align*}
\mu'(t)_j &= \sum_{\boldsymbol{i} \in \{2,\ldots,d+1\}^k} \left(\mu'(t_1)_{i_1} \ldots \mu'(t_k)_{i_k} \right) \mu'(\sigma)_{\boldsymbol{i},j}(\size{t})\\
&= \sum_{\boldsymbol{i} \in \{2,\ldots,d+1\}^k} \left( \mu(t_1)_{i_1 - 1} \ldots \mu(t_k)_{i_k-1}\right) \mu(\sigma)_{\bm{i} - \bm{1},j-1}(\size{t})\\
&= \sum_{\boldsymbol{i} \in \{1,\ldots,d\}^k} \left( \mu(t_1)_{i_1 } \ldots \mu(t_k)_{i_k}\right) \mu(\sigma)_{\boldsymbol{i},j-1}(\size{t})\\
&= \left( (\hmu(t_1) \otimes \ldots \otimes \hmu(t_k) ) \mu(\sigma)(\size{t}) \right)_{j-1}\\
&= \mu(t)_{j-1}\\
\end{align*}
Similarly, we also have:
\begin{align*}
\mu'(t)_1 &= \sum_{\boldsymbol{i} \in \{2,\ldots,d+1\}^k}
    \left( \mu(t_1)_{i_1 - 1} \ldots \mu(t_k)_{i_k-1}\right)     M_{\bm{i} - \bm{1},1}(\size{t})\\
&= \sum_{\boldsymbol{i} \in \{1,\ldots,d\}} \left( \mu(t_1)_{i_1 } \ldots \mu(t_k)_{i_k}\right) M_{\boldsymbol{i},1}(\size{t})\\
&= (\hmu(t_1) \otimes \ldots \otimes \hmu(t_k) ) \mu(\sigma)(\size{t}) \vbeta(\size{t})\\
&= \hmu'(t) \vbeta(\size{t})
\end{align*}
This concludes the proof of $\hmu'(t) = \begin{bmatrix} \hmu(t)\vbeta(\size{t}) & \hmu(t) \end{bmatrix}$ for all $t \in T_\Sigma$. As $\sem{\A'}{t} = \mu'(t)_1$, we conclude that $\sem{\A'}{t} = \hmu(t)\vbeta(\size{t}) = \sem{\A}{t}$ for all $t \in T_\Sigma$.
\end{proof}


\begin{proposition}
\label{lem:closure sem addition}
Let $\A_1$ and $\A_2$ be two differential tree automata over $\Sigma$. There exists a differential tree automaton $\A$ over $\Sigma$ such that for all $t \in T_\Sigma$, $\sem{\A}{t} = \sem{\A_1}{t} + \sem{\A_2}{t}$.
\end{proposition}

\begin{proof}
Assume that $\A_1 = (d_1,\mu_1)$ and $\A_2 = (d_2,\mu_2)$.
We build the automaton $\A' = (d,\mu)$ of dimension $d = d_1 + d_2$ such that
\[
\text{for all }t \in T_\Sigma, \hmu(t) = \begin{bmatrix} \hmu_1(t) & \hmu_2(t)\end{bmatrix}
\]
and we will then conclude by applying \Cref{lem:vector e1} with the automaton $\A'$ and the vector $\vbeta(x)$ defined as follows to obtain the desired automaton $\A$.
\[
\vbeta(x) = 
\begin{bmatrix} 
    1 \\ 
    \vzero_{(d_1-1)\times 1} \\
    1 \\ 
    \vzero_{(d_2-1)\times 1} \\
\end{bmatrix}
\]

The definition of $\hmu$ is fairly straightforward: For all $\sigma \in \Sigma_0$, $\mu(\sigma) = \begin{bmatrix} \mu_1(\sigma) & \mu_2(\sigma)\end{bmatrix}$. Moreover, for all $\sigma \in \Sigma_k$ with $k>0$, the matrix $\mu(\sigma)$ is in $\Q[x]^{d^k \times d}$ such that for all $\bm{i}= (i_1,\ldots,i_k) \in \{1, \ldots, d\}^k$,
\begin{itemize}
\item for all $j \in \{1, \ldots, d_1\}$,
\[
\mu(\sigma)_{\bm{i},j} = 
\left\{\begin{array}{lr}
    \mu_1(\sigma)_{\bm{i},j} & \text{if }\forall \ell, i_\ell \in \{ 1, \ldots, d_1 \}\\ 
    0 & \text{otherwise}
\end{array}
\right.
\]
\item for all $j \in \{ d_1+1, \ldots d\}$,
\[
\mu(\sigma)_{\bm{i},j} = 
\left\{\begin{array}{lr}
    \mu_2(\sigma)_{\bm{i} - d_1\bm{1},j-d_1} & \text{if }\forall \ell. i_\ell \in \{ d_1+1, \ldots d \}\\ 
    0 & \text{otherwise}
\end{array}
\right.
\]
\end{itemize}
The proof of $\hmu(t) = \begin{bmatrix} \hmu_1(t) & \hmu_2(t)\end{bmatrix}$ for all $t \in T_\Sigma$ is done by induction on the structure of the tree $t$ and follows by construction of $\mu$. The base case, that is $t = a$ with $a \in \Sigma_0$, is trivial by construction of $\mu(a)$. In the inductive step, $t = \sigma(t_1,\ldots,t_k)$ for some $\sigma \in \Sigma_k$ and $t_1,\ldots,t_k \in T_\Sigma$. By construction, for all $j \in \{1,\ldots,d_1\}$, for all $\bm{i}=(i_1,\ldots,i_k) \in \{1,\ldots,d\}^k$, $\mu(\sigma)_{\bm{i},j} = 0$ when there exists $\ell$ such that $i_\ell \not\in \{1,\ldots,d_1\}$. Thus, by applying our inductive hypothesis on $t_1,\ldots,t_k$, we obtain that for all $j \in \{1, \ldots, d_1\}$, 
\begin{align*}
\mu(t)_j &= \sum_{\bm{i} \in \{1,\ldots,d\}^k} 
\mu(t_1)_{i_1} \ldots  \mu(t_k)_{i_k}
\mu(\sigma)_{\bm{i},j}(\size{t})\\
&= \sum_{\bm{i} \in \{1,\ldots,d_1\}^k} \mu(t_1)_{i_1} \ldots  \mu(t_k)_{i_k} \mu(\sigma)_{\bm{i},j}(\size{t})\\
&= \sum_{\bm{i}\in \{1,\ldots,d_1\}^k} \mu_1(t_1)_{i_1} \ldots  \mu_1(t_k)_{i_k}\mu_1(\sigma)_{\bm{i},j}(\size{t})\\
&= \mu_1(t)_{j}
\end{align*}
Similarly, we have by construction that for all $j \in \{d_1+1,\ldots,d\}$, for all $\bm{i} \in \{1,\ldots,d\}^k$, $\mu(\sigma)_{\bm{i},j} = 0$ when there exists $\ell$ such that $i_\ell \not\in \{d_1 + 1,\ldots,d\}$. Hence, by applying our inductive hypothesis on $t_1,\ldots,t_k$, we obtain that for all $j \in \{d_1+1, \ldots, d\}$, 
\begin{align*}
\mu(t)_j &= \sum_{\bm{i} \in \{1, \ldots, d \}^k} \mu(t_1)_{i_1} \ldots  \mu(t_k)_{i_k} \mu(\sigma)_{\bm{i},j}(\size{t})\\
&= \sum_{\bm{i} \in \{d_1+1,\ldots,d\}^k} \mu_2(t_1)_{i_1-d_1} \ldots  \mu_2(t_k)_{i_k-d_1} \mu_2(\sigma)_{\bm{i} - d_1 \bm{1},j-d_1}(\size{t})\\
&= \sum_{\bm{i} \in \{1,\ldots,d_2\}^k} \mu_2(t_1)_{i_1}\ldots \mu_2(t_k)_{i_k} \mu_2(\sigma)_{\bm{i},j-n}(\size{t})\\
&= \mu_2(t)_{j-d_1}\qedhere
\end{align*}
\end{proof}


\begin{proposition}
\label{lem:closure sem scalar product}
Let $\A$ be a differential tree automaton over $\Sigma$. Let $\alpha \in \Q$. There exists a differential tree automaton $\A'$ over $\Sigma$ such that for all $t \in T_\Sigma$, $\sem{\A'}{t} = \alpha \sem{\A}{t}$.
\end{proposition}
\begin{proof}
It suffices to apply \Cref{lem:vector e1} on $\A$ with the vector $\vbeta(x) = \alpha\,\ve_1$, where $\ve_1$ is the canonical vector.
\end{proof}


\begin{proposition}
\label{lem:closure sem product}
Let $\A_1$ and $\A_2$ be two differential tree automata over $\Sigma$. There exists a differential tree automaton $\A$ over $\Sigma$ such that for all $t \in T_\Sigma$, $\sem{\A}{t} = \sem{\A_1}{t} \cdot  \sem{\A_2}{t}$.
\end{proposition}

\begin{proof}
We take $\A_1 = (d_1,\mu_1)$ and $\A_2 = (d_2,\mu_2)$. We build the automaton $\A = (d,\mu)$ of dimension $d = d_1d_2$ such that for all $t \in T_\Sigma$, 
\[
\hmu(t) = 
\begin{bmatrix} 
\hmu_1(t)_{1}\hmu_2(t)_{1} & \ldots & \hmu_1(t)_{1}\hmu_2(t)_{d_2} & \hmu_1(t)_{2}\hmu_2(t)_{1} & \ldots & \hmu_1(t)_{d_1}\hmu_2(t)_{d_2}
\end{bmatrix}
\]
For that purpose, we construct the function $\mu$ as follows. For all $a \in \Sigma_0$, for all $i \in \{1, \ldots, d_1\}$, for all $j \in \{1, \ldots, d_2\}$, $\mu(a)_{(i,j)} = \mu_1(a)_{i} \mu_2(a)_{j}$. For all $\sigma \in \Sigma_k$ with $k > 0$, for all $i'\in \{1, \ldots, d_1\}$,  $\bm{i} \in \{1, \ldots, d_1\}^k$ for all $j' \in \{1, \ldots, d_2\}$, $\bm{j} \in \{1, \ldots, d_2\}^k$
\[
\mu(\sigma)_{((i_1,j_1),\ldots,(i_k,j_k)),(i',j')} = \mu_1(\sigma)_{\bm{i},i'} \cdot \mu_2(\sigma)_{\bm{j},j'}
\]
By definition, $\hmu(\sigma(t_1,\ldots,t_k)) = (\hmu(t_1) \otimes \ldots \otimes \hmu(t_k)) \mu(\sigma)(\size{t})$. Therefore, applying our inductive hypothesis on each $t_1,\ldots, t_k$, we obtain that for all $i' \in \{1, \ldots, d_1\}$, for all $j' \in \{1, \ldots, d_2\}$, 
\begin{align*}
\hmu&(t)_{(i',j')} = \sum_{\substack{\bm{i} \in \{1,\ldots,d_1\}^k\\\bm{j} \in \{1,\ldots,d_2\}^k}} \mu(t_1)_{i_1,j_1} \ldots \mu(t_k)_{i_k,j_k} \cdot \mu(\sigma)(\size{t})_{((i_1,j_1),\ldots,(i_k,j_k)),(i',j')}\\
&= \sum_{\substack{\bm{i} \in \{1,\ldots,d_1\}^k\\\bm{j} \in \{1,\ldots,d_2\}^k}} 
\mu_1(t_1)_{i_1,j_1}\mu_2(t_1)_{i_1,j_1} \ldots \mu_1(t_k)_{i_k,j_k}\mu_2(t_k)_{i_k,j_k}
 \cdot \mu_1(\sigma)(\size{t})_{\bm{i},i'} \cdot \mu_2(\sigma)(\size{t})_{\bm{j},j'}\\
&= \sum_{\substack{\bm{i} \in \{1,\ldots,d_1\}^k\\\bm{j} \in \{1,\ldots,d_2\}^k}} 
\left(\mu_1(t_1)_{i_1,j_1} \ldots \mu_1(t_k)_{i_k,j_k}
 \cdot \mu_1(\sigma)(\size{t})_{\bm{i},i'}\right) \cdot \left(\mu_2(t_1)_{i_1,j_1}\ldots\mu_2(t_k)_{i_k,j_k}  \mu_2(\sigma)(\size{t})_{\bm{j},j'}\right)\\
&= \left((\hmu_1(t_1)\otimes \ldots \otimes \hmu_1(t_k)) \mu_1(\sigma)(\size{t})\right)_{i} \cdot \left((\hmu_2(t_1)\otimes \ldots \otimes \hmu_2(t_k)) \mu_2(\sigma)(\size{t})\right)_{j}\\
&=\mu_1(t)_{i} \cdot \mu_2(t)_{j}
\end{align*}
We conclude by noticing that $\sem{\A}{t} = \mu(t)_1 = \mu_1(t)_{1}\mu_2(t)_{1} = \sem{\A_1}{t} \cdot \sem{\A_2}{t}$.
\end{proof}

\subsection{Operations on the generating functions}

We start by showing that when considering generating functions, we can always restrict ourselves to alphabets $\Sigma$ where two different function symbols must have different arity. In other words, there can only be one nullary function symbol, one unary function symbol, one binary function symbol, etc. We say in this case that $\Sigma$ is \emph{arity distinct}.

\begin{proposition}
\label{prop:arity distinct}
Let $\A$ be a differential tree automaton over $\Sigma$. There exists a differential tree automaton $\A'$ over arity distinct $\Sigma'$ such that $f_{\A}(x) = f_{\A'}(x)$.
\end{proposition}

\begin{proof}
By \Cref{prop:automata to recurrence}, we know that, writing  $\vf_\A(x) = \sum_{n=0}^\infty \va_n x^n$, the sequence $(\va_n)_{n=0}^\infty$ satisfies for all $n \geq 1$,
\[
\va_n = \sum_{\ell = 1}^m \sum_{\bm{j}\in J_n(\ell)} (\va_{j_1} \otimes \ldots \otimes \va_{j_\ell})\cdot \sum_{\sigma \in \Sigma_\ell}(n)
\]
Thus, defining the matrices $\mu_1, \ldots, \mu_m$ as $\mu_i(x) = \sum_{\sigma \in \Sigma_i} \mu(\sigma)(x)$ for all $i \in \{0, \ldots, m\}$ we obtain that $(\va_n)_{n=0}^\infty$ is multi-holonomic. We conclude by applying \Cref{prop:multi to automaton}.
\end{proof}


\begin{proposition}
\label{prop:same alphabet}
Let $\A_1$ and $\A_2$ be two differential tree automata over $\Sigma_1$ and $\Sigma_2$ respectively. There exist an alphabet $\Sigma$ and two differential tree automata $\A'_1$ and $\A'_2$ over $\Sigma$ such that $f_{\A_1}(x) = f_{\A'_1}(x)$ and $f_{\A_2}(x) = f_{\A'_2}(x)$.
\end{proposition}

\begin{proof}
Take $\A_1 = (d_1,\mu_1)$ and $\A_2 = (d_2,\mu_2)$. By \Cref{prop:arity distinct}, we can also assume that $\Sigma_1$ and $\Sigma_2$ are arity distinct. 

Let $r_1$ and $r_2$ the maximum arity of symbols in $\Sigma_1$ and $\Sigma_2$. Let $\{k_1,\ldots,k_\ell\}$ be the set of arity in $\{0,\ldots, r_2\}$ such that for all $i \in \{1, \ldots, \ell\}$, $\Sigma_2$ contains a symbol of arity $k_i$ but not $\Sigma_1$. We define $\Sigma'_1$ by extending $\Sigma_1$ with the fresh symbols $\sigma_{k_1},\ldots,\sigma_{k_\ell}$ of arity $k_1,\ldots, k_\ell$ respectively and we extend $\A_1$ into $\A'_1 = (d_1,\mu'_1)$ such that $\mu'_1(\sigma_{k_i}) = \vzero_{d_1^{k_i} \times d_1}$ for all $i \in \{1, \ldots, \ell\}$. We trivially have $f_{\A_1}(x) = f_{\A'_1}(x)$. 
    
We extend in a similar fashion $\Sigma_2$ and $\A_2$ by computing $\{k'_1,\ldots,k'_{\ell'}\}$ to be the set of arities in $\{0,\ldots, r_1\}$ such that for all $i \in \{1, \ldots, \ell'\}$, $\Sigma_1$ contains a symbol of arity $k'_i$ but not $\Sigma_2$. This yields an automaton $\A'_2$ over $\Sigma'_2$ such that $f_{\A'_2}(x) = f_{\A_2}(x)$. As $\Sigma_1$ and $\Sigma_2$ are arity distinct, we have by construction that $\Sigma'_1$ is a renaming of $\Sigma'_2$. Therefore, we can fully rename the automaton $\A'_1$ to be over $\Sigma'_2$ with $f_{\A_1}(x) = f_{\A'_1}(x)$.
\end{proof}


\begin{corollary}
\label{lem:closure gen addition scalar}
Let $\A_1$ and $\A_2$ be two differential tree automata. Let $\alpha \in \K$.
\begin{itemize}
\item there exists a differential tree automaton such that $f_{\A}(x) = f_{\A_1}(x) + f_{\A_2}(x)$.
\item there exists a differential tree automaton such that $f_{\A}(x) = \alpha f_{\A_1}(x)$.
\end{itemize}
\end{corollary}

\begin{proof}
Direct from \Cref{prop:same alphabet,lem:closure sem addition,lem:closure sem scalar product}.
\end{proof}


\begin{proposition}
\label{lem:closure gen derive}
Let $\A$ be a differential tree automaton. There exists a differential tree automaton $\A'$ such that $f_{\A'}(x) = \SDerive{f_{\A}}(x) = x f'_\A(x)$.
\end{proposition}

\begin{proof}
Let $\A = (d,\mu)$ over the alphabet $\Sigma$. Let us consider the vector  $\vbeta(x) = \begin{bmatrix} x \\ \vzero_{(d-1)\times 1} \end{bmatrix}$. By \Cref{lem:vector e1}, there exists a differential tree automaton $\A'$ such that $\sem{\A'}{t} = \hmu(t)\vbeta(\size{t}) = \size{t} \hmu(t)_1 = \size{t} \sem{\A}{t}$ for all $t \in T_\Sigma$.
Hence, we directly obtain that $f_{\A'}(x) = \sum_{n=0} n \sum_{\substack{t \in T_\Sigma\\\size{t} = n}} \sem{\A}{t} x^n$. By \Cref{prop:derive}, we conclude that $f_{\A'}(x) = \SDerive{f_{\A}}(x)$.
\end{proof}


\begin{proposition}
\label{lem:closure gen integral}
Let $\A$ be a differential tree automaton. There exists a differential tree automaton $\A'$ such that $f_{\A'}(x) = \frac{1}{x}\int_{0}^{x} f_{\A}(x) dx$.
\end{proposition}

\begin{proof}
Let $\A = (d,\mu)$ over an alphabet $\Sigma$. Let us consider the vector $\vbeta(x) = \begin{bmatrix} \frac{1}{1+x} \\ \vzero_{(d-1)\times 1} \end{bmatrix}$. By \Cref{lem:vector e1}, there exists a differential tree automaton $\A'$ such that $\sem{\A'}{t} = \hmu(t) \vbeta(\size{t}) = \frac{1}{1+\size{t}} \hmu(t)_1 = \frac{1}{1+\size{t}}\sem{\A}{t}$. Hence, we directly obtain that:
\[
f_{\A'}(x) = \sum_{n=0} \frac{1}{1+n} \sum_{\substack{t \in T_\Sigma\\\size{t} = n}} \sem{\A}{t} x^n = \frac{1}{x}\int_{0}^{x} f_{\A}(x) dx\qedhere
\]
\end{proof}


\begin{proposition}
\label{lem:closure gen product}
Let $\A_1$ and $\A_2$ be two differential tree automata.
\begin{itemize}
\item There exists a differential tree automaton $\A$ such that $f_{\A}(x) = x \cdot f_{\A_1}(x)$.
\item There exists a differential tree automaton $\A$ such that $f_{\A}(x) = x \cdot f_{\A_1}(x) \cdot f_{\A_2}(x)$.
\end{itemize}
\end{proposition}

\begin{proof}
Without loss of generality, let us assume that $\A_1$ and $\A_2$ are automata over distinct alphabet $\Sigma^1$ and $\Sigma^2$ (we can always rename the function symbols otherwise). Let $\A_1 = (d_1,\mu_1)$ and $\A_2 = (d_2,\mu_2)$.

We start by building $\A$ such that $f_{\A}(x) = x \cdot f_{\A_1}(x)$. Let us take a new unary symbol $\sigma_1/1$ not already in $\Sigma^1$. We build the automaton $\A$ to have dimension $d_1 + 1$ over $\Sigma = \{ \sigma_1 \} \cup \Sigma^1$ with the weight function $\mu$ such that for all $\sigma \in \Sigma^1_0$, $\mu(\sigma) = \begin{bmatrix} 0 & \mu_1(\sigma) \end{bmatrix}$; and for all $k>0$, for all $\sigma \in \Sigma^1_k$, for all $ \bm{i} \in \{1, \ldots, d_1+1\}^k$, $ j \in \{1, \ldots, d_1+1\}$ 
\[
\mu(\sigma)_{\bm{i},j} = 
\left\{\begin{array}{lr}
    \mu_1(\sigma)_{\bm{i - 1},j-1} & \text{if }\forall \ell, i_\ell,j \in \{ 2, \ldots, d_1+1 \}\\
    0 & \text{otherwise}
\end{array}
\right.
\]
and
\[
\mu(\sigma_1)(x) = 
\begin{bmatrix} 0 & \vzero_{1\times d_1}\\
    1& \vzero_{1\times d_1}\\
    \vzero_{(d_1 - 1)\times 1} & \vzero_{(d_1 - 1)\times d_1}\\
\end{bmatrix}
\]
With such a construction, we can show the following property: for all $t \in T_\Sigma$, 
\begin{itemize}
\item if $t \in T_{\Sigma^1}$ then $\hmu(t) = \begin{bmatrix} 0 & \hmu_1(t) \end{bmatrix}$; 
\item if $t = \sigma_1(t')$ with $t' \in T_{\Sigma^1}$ then $\hmu(\sigma_1(t')) = \begin{bmatrix} \mu_1(t')_1 & \vzero_{1\times d_1} \end{bmatrix}$
\item and $\hmu(t) = \vzero_{1\times (d_1+1)}$ otherwise 
\end{itemize}
We prove these properties by induction on the structure of $t$. In the base case, $t$ is necessarily a nullary symbol $\sigma \in \Sigma^1_0$. By definition $\mu(\sigma) = \begin{bmatrix} 0 & \mu_1(\sigma) \end{bmatrix}$ hence the result holds. In the inductive step, we do a case analysis on $t$: 
\begin{itemize}
\item Case $t = \sigma(t_1,\ldots,t_k)$ with $\sigma \in \Sigma^1_k$ and $t_1,\ldots,t_k \in T_{\Sigma^1}$: By inductive hypothesis on $t_1,\ldots,t_k$, we have that for all $i \in \{1, \ldots, d_1+1\}$, 
\[
\mu(t)_i = (\begin{bmatrix} 0 & \hmu_1(t_1) \end{bmatrix} \otimes \ldots \otimes \begin{bmatrix} 0 & \hmu_1(t_k) \end{bmatrix}) \mu(\sigma)(\size{t})_i
\]
Thus, by definition of $\mu(\sigma)$, we have $\mu(t)_1 = 0$ and  for all $i \in \{2, \ldots, d_1+1\}$, \[\mu(t)_i = \left(\hmu_1(t_1) \otimes \ldots \otimes \hmu_1(t_k)\right) \mu_1(\sigma)(\size{t})_{i-1}\,.\] We thus obtain:
\[
    \hmu(\sigma(t_1,\ldots,t_k)) = \begin{bmatrix} 0 & \hmu_1(\sigma(t_1,\ldots,t_k)) \end{bmatrix}
\]
\item Case $t = \sigma(t_1,\ldots,t_k)$ with $\sigma \in \Sigma^1_k$ and there exists $\ell \in \{1, \ldots,k\}$ such that $t_\ell \not\in T_{\Sigma^1}$: By inductive hypothesis on $t_\ell$, we know that for all $i \in \{2, \ldots, d_1+1\}$, $\hmu(t_\ell)_i = 0$. Thus, we deduce from the definition of $\mu(\sigma)$ that:
\[
(\hmu(t_1) \otimes \ldots \otimes \hmu(t_k)) \mu(\sigma)(\size{t}) = \vzero_{1\times(d_1+1)}
\]
\item Case $t = \sigma_1(t')$ with $t' \in T_{\Sigma^1}$: By inductive hypothesis on $t'$, we know that $\hmu(t') = \begin{bmatrix} 0 & \hmu_1(t') \end{bmatrix}$. Thus, following the definition of $\mu(\sigma_1)$, we directly have that:
\[
\hmu(t) = \hmu(t') \mu(\sigma_1)(\size{t}) = \begin{bmatrix} \hmu_1(t)_1 & \vzero_{1\times d_1} \end{bmatrix}
\]
\item Case $t = \sigma_1(t')$ with $t' \not\in T_{\Sigma^1}$: By inductive hypothesis on $t'$, we know that for all $i \in \{2, \ldots, d_1+1\}$, $\mu(t')_i = 0$. Hence, by definition of $\mu(u)$, we obtain:
\[
\hmu(t) = \hmu(t') \mu(\sigma_1)(\size{t}) = \vzero_{1\times(d_1+1)}
\]
\end{itemize}
Let us now conclude the proof by computing the generating series. First notice that for all  $t \in T_\Sigma$, $\mu(t)_1 \neq 0$ implies $t = \sigma_1(t')$ with $t' \in T_{\Sigma^1}$. Hence,
\[
f_{\A}(x) = \sum_{n=0}^\infty \sum_{\substack{t \in T_{\Sigma}\\\size{t} = n}} \mu(t)_1  x^n
= \sum_{n=1}^\infty \sum_{\substack{t' \in T_{\Sigma^1}\\\size{t'} = n-1}} \mu(\sigma_1(t'))_1  x^n
= x\sum_{n=1}^\infty \sum_{\substack{t' \in T_{\Sigma^1}\\\size{t'} = n-1}} \mu_1(t')_1  x^{n-1} 
= x \cdot f_{\A_1}(x)
\]

\medskip

We now build a differential tree automaton $\A$ such that $f_\A(x) = x \cdot f_{\A_1}(x) \cdot f_{\A_2}(x)$. The proof is in fact very similar to the above proof. Instead of considering a new unary symbol, we consider a new binary symbol $\sigma_2$ (hence arity 2) not already in $\Sigma^1$ and $\Sigma^2$. We build the automaton $\A$ over $\Sigma = \{ \sigma_2 \} \cup \Sigma^1 \cup \Sigma^2$ with dimension $d = d_1 + d_2 + 1$, and the weight function $\mu$ such that: 
\begin{itemize}
    \item for all $\sigma \in \Sigma^1_0$, $\mu(\sigma) = \begin{bmatrix} 0 & \mu_1(\sigma) & \vzero_{1 \times d_2}\end{bmatrix}$
    \item for all $\sigma \in \Sigma^2_0$, $\mu(\sigma) = \begin{bmatrix} 0 & \vzero_{1 \times d_1} & \mu_2(\sigma) \end{bmatrix}$
    \item for all $k>0$, for all $\sigma \in \Sigma^1_k$, for all $\bm{i} \in \{1,\ldots,d\}^k,j \in \{1,\ldots,d\}$,
\[
\mu(\sigma)_{\bm{i},j} = 
\left\{\begin{array}{lr}
    \mu_1(\sigma)_{\bm{i - 1},j-1} & \text{if }\forall \ell, i_\ell,j \in \{ 2, \ldots, d_1+1 \}\\
    0 & \text{otherwise}
\end{array}
\right.
\]
\item for all $k>0$, for all $\sigma \in \Sigma^2_k$, for all $\bm{i} \in \{1,\ldots,d\}^k,j \in \{1,\ldots,d\}$, denoting $n = d_1 + 1$, 
\[
\mu(\sigma)_{\bm{i},j} = 
\left\{\begin{array}{lr}
    \mu_1(\sigma)_{\bm{i } -n \bm{1},j-n} & \text{if }\forall \ell, i_\ell,j \in \{ d_1+2, \ldots, d \}\\
    0 & \text{otherwise}
\end{array}
\right.
\]
\item $\mu(u)_{(1,1),0} = 1$ and otherwise $\mu(u)_{(i_1,i_2),j} = 0$. In other words:
\[
\mu(u) = 
\begin{bmatrix} 
    \vzero_{(d_1+1)\times (d_1+d_2)} & \vzero_{(d_1+1)\times (d_1+d_2)}\\
    1& \vzero_{1\times (d_1+d_2)}\\
    \vzero_{(d^2 - d_1 -2)\times 1} & \vzero_{(d^2 - d_1 -2)\times (d_1+d_2)}\\
\end{bmatrix}
\]
\end{itemize}
This construction naturally entails a similar property: for all $t \in T_\Sigma$,
\begin{itemize}
    \item if $t \in T_{\Sigma^1}$ then $\hmu(t) = \begin{bmatrix} 0 & \hmu_1(t) & \vzero_{1\times d_2}\end{bmatrix}$; 
    \item if $t \in T_{\Sigma^2}$ then $\hmu(t) = \begin{bmatrix} 0 & \vzero_{1\times d_1} & \hmu_2(t) \end{bmatrix}$; 
    \item if $t = \sigma_2(t_1,t_2)$ with $t_1 \in T_{\Sigma^1}$ and $t_2 \in T_{\Sigma^2}$ then \[ 
    \hmu(\sigma_2(t_1,t_2)) = \begin{bmatrix} \mu_1(t_1)_1 \cdot \mu_2(t_2)_1 & \vzero_{1\times (d_1 + d_2)} \end{bmatrix}\]
    \item and $\hmu(t) = \vzero_{1\times (d_1+d_2+1)}$ otherwise 
\end{itemize}
The proof of this property is done once again by induction on the structure of $t$ and is very similar to the above unary case, hence we omit the details. We now conclude the main proof by computing the generating series: Once again, notice that for all $t \in T_\Sigma$, $\hmu(t)_1 \neq 0$ implies $t = \sigma_2(t_1,t_2)$. Therefore, we have:
\begin{align*}
f_{\A}(x) &= \sum_{n=0}^\infty\; \sum_{\substack{t \in T_{\Sigma}\\\size{t} = n}} \mu(t)_1  x^n\\
&= \sum_{n=1}^\infty\; \sum_{k=0}^{n-1}\; \sum_{\substack{t_1 \in T_{\Sigma^1}\\\size{t_1} = k}}\; \sum_{\substack{t_2 \in T_{\Sigma^2}\\\size{t_2} = n-1-k}} \mu_1(t_1)_1 \cdot \mu_2(t_2)_1  x^n \\
&= x \sum_{n=1}^\infty\; \sum_{k=0}^{n-1}\; (\sum_{\substack{t_1 \in T_{\Sigma^1}\\\size{t_1} = k}} \mu_1(t_1)_1x^k) \cdot  (\sum_{\substack{t_2 \in T_{\Sigma^2}\\\size{t_2} = n-1-k}} \mu_2(t_2)_1x^{n-1-k}) \\
&= x \cdot (\sum_{n=0}^\infty \sum_{\substack{t_1 \in T_{\Sigma^1}\\\size{t_1} = n}} \mu_1(t_1)_1 x^n ) \cdot ( \sum_{n=0}^\infty \sum_{\substack{t_2 \in T_{\Sigma^2}\\\size{t_2} = n}} \mu_2(t_2)_1 x^n )\\
&= x \cdot f_{\A_1}(x) \cdot f_{\A_2}(x)\qedhere
\end{align*}
\end{proof}


\begin{proposition}
\label{prop:closure gen devise}
Let $\A$ be a differential tree automaton. Let $f_\A = \sum^\infty_{n=0} a_nx^n$. There exists a differential tree automaton $\A'$ such that $f_{\A'}(x) = \frac{f_\A(x) - a_0}{x} = \sum^\infty_{n=0} a_{n+1}x^n$.
\end{proposition}

\begin{proof}
Let $\A = (d,\mu)$ over $\Sigma$. Thanks for \Cref{prop:arity distinct}, we assume that $\Sigma$ is arity distinct. In other words, we can assume that $\Sigma = \{ \sigma_k \}_{k=0}^r$ for some $r$ with $\sigma_0, \ldots, \sigma_r$ of arity $0, \ldots, r$ respectively. The difficulty of this proof is that we need to \emph{decrease} the size of the trees of $\A$ by 1. We thus consider a new alphabet $\Sigma' = \Sigma \cup \bigcup_{k=1}^r \bigcup_{i=0}^k \{ \alpha_{k,i}\}$ where each $\alpha_{k,i}$ has arity $i$. Intuitively a term $\alpha_{k,i}(t_1,\ldots,t_{i-1},u)$ will have the same value in $\A'$ as the term $\sigma_k(t_1,\ldots,t_{i-1},t_i,\sigma_0,\ldots,\sigma_0)$ in $\A$ with $t_i$ having the same value as $u$ but with the size decreased by 1. Formally, we build an injective transformation $\Gamma$ from trees $t$ of $T_\Sigma$ of size $\size{t}>0$ to $T_{\Sigma'}$ as follows: For all $t \in T_\Sigma$, if $t = \sigma_k(t_1,\ldots,t_k)$ with $k>0$ and $i = \min(\{k\} \cup \{ j-1 \geq 0\mid \forall k \geq \ell \geq j, t_\ell = \sigma_0\})$ then
\[
\Gamma(t) = \alpha_{k,i}(t_1,\ldots,t_{i-1},\Gamma(t_i))
\]
A simple inductive proof allows us to show that for all $t \in T_\Sigma$, $\size{t} = \size{\Gamma(t)} + 1$. For this, notice that when $\size{t} = 1$, $t$ must be of the form $t = \sigma_k(\sigma_0,\ldots,\sigma_0)$ for some $k$. Hence, $\Gamma(t) = \alpha_{k,0}$.

We will build the automata $\A'$ to have dimension $d' = 2d$ over $\Sigma'$ with the weight function $\mu'$ such that for all $t \in T_{\Sigma'}$, 
\begin{itemize}
\item if $t \in T_\Sigma$ then $\hmu'(t) = \begin{bmatrix} \vzero_{d \times 1} & \hmu(t)\end{bmatrix}$
\item if there exists $t' \in T_\Sigma$ such that $\Gamma(t') = t$ then $\hmu'(t) = \begin{bmatrix} \hmu(t') & \vzero_{d \times 1}\end{bmatrix}$
\item $\hmu'(t) = \vzero_{2d \times 1}$ otherwise.
\end{itemize}
To achieve this property, we build $\mu'$ as follows:
\begin{itemize}
\item $\mu'(\sigma_0) = \begin{bmatrix} \vzero_{d \times 1} & \mu(\sigma_0)\end{bmatrix}$
\item for all $k \in \{1, \ldots, r\}$, for all $ \bm{i} \in \{1, \ldots, 2d\}^k  , j \in \{1, \ldots, 2d\}$,
\[
\mu'(\sigma_k)_{\bm{i},j} = 
\left\{\begin{array}{lr}
    \mu(\sigma)_{\bm{i} - d\bm{1},j-d} & \text{if }\forall \ell. i_\ell,j \in \{ d+1, \ldots, 2d \}\\
    0 & \text{otherwise}
\end{array}
\right.
\]
\item for all $k \in \{1, \ldots, r\}$, $\mu'(\alpha_{k,0}) = \begin{bmatrix} \hmu(\sigma_k(\sigma_0,\ldots,\sigma_0)) & \vzero_{d\times 1}\end{bmatrix}$.
\item for all $k \in \{1, \ldots, r\}$, for all $\ell \in \{1, \ldots, k\}$, denoting 
\[
M^{k,\ell}(x) = (I_{d^\ell} \otimes {\text{\Large $\otimes$}}_{j=\ell+1}^k \mu(\sigma_0))\mu(\sigma_k)(x + 1)
\]
we define for all for all $i_1,\ldots, i_\ell, j \in \{1, \ldots, 2d\}$,
\[
\mu'(\alpha_{k,\ell})_{\bm{i},j} = 
\left\{\begin{array}{l}
    M^{k,\ell}_{(i_1-d,\ldots,i_{\ell-1}-d,i_\ell),j} \quad
    \text{if }i_1,\ldots,i_{\ell-1} \in \{ d+1, \ldots, 2d \}, i_\ell,j \in \{1,\ldots,d\}\\
    0 \hfill\text{otherwise}
\end{array}
\right.
\]
\end{itemize}
We now prove the desired property by induction on the structure of $t$. In the base case, $t$ is a nullary symbol. Hence, either $t = \sigma_0$ or $t = \alpha_{k,0}$ for some $k \in \{1, \ldots, r\}$. In the former case, we have by definition that $\mu'(\sigma_0) = \begin{bmatrix} \vzero_{d \times 1} & \mu(\sigma_0)\end{bmatrix}$ hence the result directly holds. In the latter case, by definition of $\alpha_{k,0}$, we know that $\Gamma(\sigma_k(\sigma_0,\ldots,\sigma_0)) = \alpha_{k,0}$. Moreover, by definition, we have $\mu'(\alpha_{k,0}) = \begin{bmatrix} \hmu(\sigma_k(\sigma_0,\ldots,\sigma_0)) & \vzero_{d\times 1}\end{bmatrix}$.

\medskip

In the inductive step, we have that either $t = \sigma_k(t_1,\ldots,t_k)$ or $t = \alpha_{k,i}(t_1,\ldots,t_i)$ for some $k>0$ and $i>0$. 
\begin{itemize}
\item Case $t = \sigma_k(t_1,\ldots,t_k)$ when $t_1,\ldots,t_k \in T_\Sigma$: By inductive hypothesis on each $t_\ell$s, we have that for all $j \in \{1,\ldots,2d\}$, 
\[
\hmu'(t)_j = (\begin{bmatrix} \vzero_{1\times d} & \hmu(t_1) \end{bmatrix} \otimes \ldots \otimes \begin{bmatrix} \vzero_{1\times d} & \hmu(t_k) \end{bmatrix}) \mu'(\sigma_k)(\size{t})_j
\]
Thus, by definition of $\mu'(\sigma_k)$, we have $\mu'(t)_j = 0$ for $j \in \{1, \ldots, d\}$ and $\mu'(t)_j = (\hmu(t_1) \otimes \ldots \otimes \hmu(t_k)) \mu(\sigma_k)(\size{t})_{j-d}$ when $j \in \{d+1,\ldots,2d\}$. We deduce that $\hmu'(t) = \begin{bmatrix} \vzero_{d \times 1} & \hmu(t)\end{bmatrix}$.
\item Case $t = \sigma_k(t_1,\ldots,t_k)$ and there exists $\ell \in \{1, \ldots, k\}$ such that $t_\ell \in T_\Sigma$: By inductive hypothesis on $t_\ell$, we know that for all $j \in \{d+1,\ldots,2d\}$, $\mu(t_\ell)_j = 0$. Thus we deduce from the definition of $\mu'(\sigma_k)$ that 
\[
(\hmu'(t_1) \otimes \ldots \otimes \hmu'(t_k)) \mu(\sigma_k)(\size{t}) = \vzero_{1\times 2d}
\]
\item Case $t = \alpha_{k,i}(t_1,\ldots,t_i)$ and $t_1,\ldots,t_{i-1} \in T_\Sigma$ and there exists $t'_i \in T_\Sigma$ such that $\Gamma(t'_i) = t_i$: By inductive hypothesis on each $t_\ell$s, we know that for all $j \in \{1, \ldots, 2d\}$, 
\[
\mu'(t)_j = (\begin{bmatrix} \vzero_{1\times d} & \hmu(t_1) \end{bmatrix} \otimes \ldots \otimes \begin{bmatrix} \vzero_{1\times d} & \hmu(t_{i-1}) \end{bmatrix} \otimes \begin{bmatrix} \hmu(t'_i) & \vzero_{1\times d} \end{bmatrix}) \mu'(\alpha_{k,i})(\size{t})_j
\]
Recall that as $\Gamma(t'_i) = t_i$, we have $\size{t'_i} = \size{t_i} + 1$. Denoting $t' = \sigma_k(t_1,\ldots,t_{i-1}, t'_i,\sigma_0,\ldots,\sigma_0)$, we thus obtain $\size{t'} = \size{t} + 1$. Hence unfolding the definition of $\mu'(\alpha_{k,i})$ gives us for all $j \in \{d+1,\ldots, 2d\}$, $\mu'(t)_j = 0 $ and for all $j \in \{1,\ldots, d\}$,
\begin{align*}
\mu'(t)_j &= (\hmu(t_1) \otimes \ldots \otimes \hmu(t_{i-1}) \otimes \hmu(t'_i)) \cdot M^{k,i}_j(\size{t})\\
&=(\hmu(t_1) \otimes \ldots \otimes \hmu(t_{i-1}) \otimes \hmu(t'_i)) \cdot (I_{d^i} \otimes {\text{\Large $\otimes$}}_{\ell=i+1}^k \mu(\sigma_0)) \cdot \mu(\sigma_k)(\size{t'})_j
\end{align*}
Notice that $A = (\hmu(t_1) \otimes \ldots \otimes \hmu(t_{i-1}) \otimes \hmu(t'_i)) \in \K^{1\times d^i}$ and $B = {\text{\Large $\otimes$}}_{\ell=i+1}^k \mu(\sigma_0) \in \Q^{1\times d^{k-1}}$. Hence, by the mixed-product property of the Kronecker product, we obtain that $A(I_{d^i} \otimes B) = (AI_{d^i}) \otimes (I_1B) = A \otimes B$. Therefore:
\begin{align*}
\mu'(t)_j &= (\hmu(t_1) \otimes \ldots \otimes \hmu(t_{i-1}) \otimes \hmu(t'_i) \otimes \text{\Large $\otimes$}_{\ell=i+1}^k \mu(\sigma_0)) \cdot \mu(\sigma_k)(\size{t'})_j\\
&= \mu(\sigma_k(t_1,\ldots,t_{i-1},t'_i,\sigma_0,\ldots,\sigma_0))_j
\end{align*}
Notice that by definition of $\Gamma$, we have $\Gamma(t') = t$ which allows us to conclude that $t'\in T_\Sigma$, and $t = \Gamma(t')$, and $\hmu'(t) = \begin{bmatrix} \hmu(t') & \vzero_{d \times 1} \end{bmatrix}$.
\item Case $t = \alpha_{k,i}(t_1,\ldots,t_i)$ and either $t_\ell \not\in T_\Sigma$ with $\ell \in \{1, \ldots, i-1\}$ or for all $t'_i \in T_\Sigma$, $\Gamma(t'_i) \neq t_i$: In the former case, by inductive hypothesis on $t_\ell$, we deduce that for all $j \in \{d+1,\ldots, 2d\}$, $\mu'(t_\ell)_j = 0$. In the latter case, applying our inductive hypothesis on $t_i$ gives us that for all $j \in \{1,\ldots, d\}$, $\hmu'(t_i)_j = 0$. In both cases, we deduce from the definition of $\mu'(\alpha_{k,i})$ that 
\[
(\hmu'(t_1) \otimes \ldots \otimes \hmu'(t_i)) \mu'(\alpha_{k,i})(\size{t}) = \vzero_{1\times 2d}
\]
\end{itemize}
This conclude the proof of the desired property. Let us now conclude the main result by computing the power series. Thanks to our desired property, notice that for all $t \in T_{\Sigma'}$, $\sem{\A'}{t} \neq 0$ implies that there exists $t' \in T_\Sigma$ such that $\Gamma(t') = t$ and $\sem{\A'}{t} = \sem{\A}{t'}$. As $\Gamma$ is injective, we deduce that for all $n \in \N$, 
\[
\sum_{\substack{t \in T_{\Sigma'}\\ \size{t} = n}} \sem{\A'}{t} =
\sum_{\substack{t \in T_{\Sigma'}, \size{t} = n \\ \exists t' \in T_\Sigma. \Gamma(t') = t}} \sem{\A'}{t} = \sum_{\substack{t' \in T_{\Sigma}\\ \size{t'} = n+1}} \sem{\A}{t'}
\]
which allows us to conclude.
\end{proof}


\begin{corollary}
\label{cor:closure gen derive and integral}
Let $\A$ be a differential tree automaton. Let $f_A(x) = \sum_{n=0}^\infty a_nx^n$.
\begin{itemize}
\item There exists a differential tree automaton $\A'$ such that $f_{\A'}(x) = f_{\A}'(x)$.
\item There exists a differential tree automaton $\A'$ such that $f_{\A'}(x) = \int_{0}^{x} f_{\A}(x) dx$.
\end{itemize}
\end{corollary}

\begin{proof}
From \Cref{lem:closure gen derive}, we have a $f_{\A''}(x) = xf'_\A(x)$ for some $\A''$ and $f_{\A''}(x)$ has no constant term. Hence by \Cref{prop:closure gen devise}, we have $f_{\A'}(x) = \frac{f_{\A''}(x)}{x} = f'_\A(x)$. From \Cref{lem:closure gen integral,lem:closure gen product}, we directly obtain that $f_{\A'}(x) = \int_{0}^{x} f_{\A}(x) dx$ for some $\A'$.
\end{proof}


\begin{corollary}
\label{cor:closure gen product complete}
Let $\A_1$, $\A_2$ be two differential tree automata. There exists a differential tree automaton $\A$ such that $f_{\A}(x) = f_{\A_1}(x) \cdot f_{\A_2}(x)$.
\end{corollary}

\begin{proof}
Direct from \Cref{lem:closure sem product,prop:closure gen devise}.
\end{proof}

\begin{proposition}
\label{prop:closure gen inverse}
Let $\A$ be a differential tree automaton. Let $f_\A(x) = \sum_{n=0}^\infty a_nx^n$. If $a_0 \neq 0$ then there exists a differential tree automaton $\A'$ such that $f_{\A'}(x) = \frac{1}{f_{\A}(x)}$.
\end{proposition}

\begin{proof}
It is well known that when $a_0 \neq 0$, $f_A(x)$ admits an multiplicative inverse power series $f^{-1}_\A(x) = \sum_{n=0}^\infty b_nx^n$ where the coefficients of $f^{-1}_\A(x)$ satisfy the following property:
\[
\left\{
\begin{array}{l}
b_0 = \frac{1}{a_0}\\
b_n = - \frac{1}{a_0} \sum_{i=0}^{n-1} a_{i+1}b_{n-1-i}\quad \forall n > 0 
\end{array}
\right.
\]
Using \Cref{prop:closure gen devise}, we first build the automaton $\A''$ such that $f_{\A''}(x) = \sum^\infty_{n=0} a_{n+1}x^n$. By \Cref{prop:arity distinct}, we can assume that $\A'' = (d,\mu)$ over some alphabet $\Sigma$ that is arity distinct. Let us denote $\vf(x) = \sum_{n=0} \vc_nx^n$ the vector of power series corresponding to $\A''$, that is $\vc_n = \sum_{\substack{t \in T_\Sigma\\\size{t}=n}} \mu(t)$. Notice that $\vc_{n,1} = a_{n+1}$ for $n \in \N$.

We build the automaton $\A' = (d',\mu')$ over $\Sigma' = \Sigma \cup \{ \beta\}$ with $d' = d+1$, and $\beta$ a fresh binary function symbol and the weight function $\mu'$ such that for all $n \geq 0$, $\va'_n = \sum_{\substack{t \in T_{\Sigma'}\\\size{t}=n}} \mu'(t) = \begin{bmatrix} b_n & \vc_n \end{bmatrix}$. For that purpose, let us define $\mu'$ such that 
\begin{itemize}
\item $\mu'(a) = \begin{bmatrix} \frac{1}{a_0} & \mu(a) \end{bmatrix}$ where $a$ is the unique nullary function symbol in $\Sigma$ (recall that $\Sigma$ is arity distinct)
\item for all $k > 0$, for all $\sigma \in \Sigma_k$, for all $ \bm{i} \in \{1,\ldots,d+1\}^k ,j \in \{1,\ldots,d+1\} $, 
\[
\mu'(\sigma)_{\bm{i},j} = 
\left\{\begin{array}{lr}
    \mu(\sigma)_{\bm{i - 1},j-1} &  \text{if }\forall \ell. i_\ell,j \in \{2,\ldots,d+1\}\\
    0 & \text{otherwise }
\end{array}
\right.
\]
\item 
\[\mu'(\beta) = 
\begin{bmatrix}
0 & \vzero_{1\times d}\\
-\frac{1}{a_0} & \vzero_{1\times d}\\
\vzero_{({d'}^2-2) \times 1} & \vzero_{({d'}^2-2) \times d}
\end{bmatrix}
\]
\end{itemize}

In the inductive step $n > 0$, from our inductive hypothesis, we have
\begin{align*}
\va'_n &=  \sum_{\substack{k>0\\\sigma \in \Sigma'_k}}\; \sum_{\substack{n_1,\ldots,n_k \in \N \\ n_1 +\ldots n_k = n-1}} (\va'_{n_1} \otimes \ldots \otimes \va'_{n_k}) \mu'(\sigma)(n)\\
&= \sum_{\substack{k>0\\\sigma \in \Sigma'_k}}\; \sum_{\substack{n_1,\ldots,n_k \in \N \\ n_1 +\ldots n_k = n-1}} (\begin{bmatrix} b_{n_1} & \vc_{n_1} \end{bmatrix} \otimes \ldots \otimes \begin{bmatrix} b_{n_k} & \vc_{n_k} \end{bmatrix}) \mu'(\sigma)(n)\\
&= \sum_{\substack{k>0\\\sigma \in \Sigma_k}} \sum_{\substack{n_1,\ldots,n_k \in \N \\ n_1 +\ldots n_k = n-1}} (\begin{bmatrix} b_{n_1} & \vc_{n_1} \end{bmatrix} \otimes \ldots \otimes \begin{bmatrix} b_{n_k} & \vc_{n_k} \end{bmatrix}) \mu'(\sigma)(n)\\
&\quad + \sum_{i=0}^{n-1} (\begin{bmatrix} b_i & \vc_i \end{bmatrix} \otimes \begin{bmatrix} b_{n-1-i} & \vc_{n-1-i} \end{bmatrix}) \mu'(\beta)(n)
\end{align*}
Following the definition of $\mu'(\beta)(x)$, we deduce that
\[
(\begin{bmatrix} b_i & \vc_i \end{bmatrix} \otimes \begin{bmatrix} b_{n-1-i} & \vc_{n-1-i} \end{bmatrix}) \mu'(\beta)(n) = \begin{bmatrix} -\frac{1}{a_0}\vc_{i,1}\,b_{n-1-i} & \vzero_{1\times d}\end{bmatrix} = \begin{bmatrix} -\frac{a_{n+1}b_{n-1-i}}{a_0} & \vzero_{1\times d}\end{bmatrix}
\]
Moreover, by definition of $\mu'(\sigma)(x)$ for all $k>0$, for all $\sigma \in \Sigma_k$, we deduce that:
\[
(\begin{bmatrix} b_{n_1} & \vc_{n_1} \end{bmatrix} \otimes \ldots \otimes \begin{bmatrix} b_{n_k} & \vc_{n_k} \end{bmatrix}) \mu'(\sigma)(n) = \begin{bmatrix} 0 & (\vc_1 \otimes \ldots \otimes \vc_k) \mu(\sigma)(n) \end{bmatrix}
\]
We therefore obtain:
\begin{align*}
\va'_n &= \sum_{\substack{k>0\\\sigma \in \Sigma_k}}\; \sum_{\substack{n_1,\ldots,n_k \in \N \\ n_1 +\ldots n_k = n-1}} \begin{bmatrix} 0 & (\vc_1 \otimes \ldots \otimes \vc_k) \mu(\sigma)(n,n_1,\ldots,n_k) \end{bmatrix}  + \sum_{i=0}^{n-1} \begin{bmatrix} -\frac{a_{n+1}b_{n-1-i}}{a_0} & \vzero_{1\times d}\end{bmatrix}\\
&= \begin{bmatrix} 0 & \vc_n \end{bmatrix} + \begin{bmatrix} b_n & \vzero_{1\times d}\end{bmatrix}\\
&= \begin{bmatrix} b_n & \vc_n \end{bmatrix}
\end{align*}
As for all $n \geq 0$, $\va'_{n,1} = b_n$, we conclude that $f_{\A'}(x) = f^{-1}_{\A}(x)$.
\end{proof}

\end{document}